\documentclass[10pt,a4paper,final]{article}
\usepackage[english]{babel}%
\usepackage{amssymb,amsmath,ae,amsthm}%
\usepackage[notcite,notref]{showkeys}%
\usepackage[dvips]{graphicx}%
\usepackage{float, flafter}%
\usepackage{psfrag,epsfig}%
\usepackage[dvips]{color}
\theoremstyle{plain}
\newtheorem{theorem}{Theorem}[section]
\newtheorem{corollary}[theorem]{Corollary}
\newtheorem{lemma}[theorem]{Lemma}

\newtheorem{principle}[theorem]{Principle}
\theoremstyle{definition}
\newtheorem{definition}[theorem]{Definition}
\newtheorem{assumption}[theorem]{Assumption}
\newtheorem{remark}[theorem]{Remark}
\newtheorem{example}[theorem]{Example}

\newcommand{\tspace}[2]{{T}|_{#1}{#2}}

\newcommand{\OSPot}{\Phi_0}
\newcommand{\NNPot}{\Phi_1}
\newcommand{\cc}{\mathrm{c.c.}}

\renewcommand{\Re}{\mathop{\mathrm{Re}}}

\newcommand{\mhol}{\overline}

\newcommand{\mhintlimits}{\limits}%
\renewcommand{\d}{\mathrm d}
\newcommand{\vel}{{\mathrm{vel}}}

\newcommand{\con}{{\mathrm{con}}}

\newcommand{\bbTvp}{\bbT_\vel}

\newcommand{\skp}[2]{{\langle{#1},\,{#2}\rangle}}
\newcommand{\sskp}[2]{{\langle\!\langle{#1},\,{#2}\rangle\!\rangle}}
\newcommand{\DO}[1]{{{\mathcal{O}}\at{#1}}}

\newcommand{\laplace}{\triangle}
\newcommand{\Rset}{\mathbb{R}}
\newcommand{\Cset}{\mathbb{C}}

\newcommand{\Zset}{\mathbb{Z}}

\newcommand{\discr}{{\rm discr}}

\newcommand{\mtrc}{{\rm met}}

\newcommand{\red}{{\rm red}}

\newcommand{\fin}{{\rm fin}}

\newcommand{\ini}{{\rm ini}}

\newcommand{\const}{{\rm const}}
\newcommand{\can}{{\rm can}}

\newcommand{\trans}{{T}}
\newcommand{\tdots}{{...}}%

\newcommand{\Si}{{\Sigma}}
\newcommand{\Om}{{\Omega}}
\newcommand{\al}{{\alpha}}
\newcommand{\be}{{\beta}}
\newcommand{\ga}{{\gamma}}

\newcommand{\eps}{{\varepsilon}}

\newcommand{\la}{{\lambda}}
\newcommand{\si}{{\sigma}}

\newcommand{\om}{{\omega}}

\newcommand{\pair}[2]{{\left({#1},\,{#2}\right)}}
\newcommand{\npair}[2]{{({#1},\,{#2})}}

\newcommand{\Bpair}[2]{{\Big({#1},\,{#2}\Big)}}
\newcommand{\at}[1]{{\left({#1}\right)}}
\newcommand{\nnat}[1]{{({#1})}}
\newcommand{\bat}[1]{{\big(#1\big)}}
\newcommand{\Bat}[1]{{\Big(#1\Big)}}
\newcommand{\triple}[3]{{\left({#1},\,{#2},\,{#3}\right)}}

\newcommand{\dualpair}[2]{{\left\langle{#1},\,{#2}\right\rangle}}

\newcommand{\Phase}{\phi}
\newcommand{\MiTime}{t}

\newcommand{\MiLagr}{j}
\newcommand{\MiLagrC}{\eta}

\newcommand{\MaTime}{\tau}
\newcommand{\MaLagr}{y}

\newcommand{\MiLagrCDer}[1]{\partial_{\,\MiLagrC}{#1}}

\newcommand{\MaLagrDer}[1]{\partial_{\,\MaLagr}{#1}}
\newcommand{\MaTimeDer}[1]{\partial_{\,\MaTime}{#1}}
\newcommand{\PhaseDer}[1]{\partial_{\,\Phase}{#1}}

\newcommand{\MaLagrDerS}[1]{\partial_{\,\MaLagr\MaLagr}{#1}}
\newcommand{\MaTimeDerS}[1]{\partial_{\,\MaTime\MaTime}{#1}}

\newcommand{\ModPhase}{\Theta}

\newcommand{\ccinterval}[2]{[#1,\,#2]}%
\newcommand{\VarDer}[2]{\partial_{#2}{#1}}
\newcommand{\SVarDer}[2]{\partial^2_{#2}{#1}}

\newcommand{\bigpar}{\par$\;$\par\noindent}

\newcommand{\iu}{\mathtt{i}}
\newcommand{\mhexp}[1]{{{\mathtt{e}}^{#1}}}

\newcommand{\rra}{\rightarrow}

\newcommand{\abs}[1]{\left|{#1}\right|}

\newcommand{\calA}{\mathcal{A}}

\newcommand{\calE}{\mathcal{E}}
\newcommand{\calF}{\mathcal{F}}

\newcommand{\calH}{\mathcal{H}}
\newcommand{\calI}{\mathcal{I}}

\newcommand{\calK}{\mathcal{K}}
\newcommand{\calL}{\mathcal{L}}
\newcommand{\calM}{\mathcal{M}}

\newcommand{\calP}{\mathcal{P}}

\newcommand{\calR}{\mathcal{R}}
\newcommand{\calS}{\mathcal{S}}
\newcommand{\calT}{\mathcal{T}}

\newcommand{\calV}{\mathcal{V}}

\newcommand{\calZ}{\mathcal{Z}}

\newcommand{\bbE}{\mathbb{E}}

\newcommand{\bbH}{\mathbb{H}}
\newcommand{\bbI}{\mathbb{I}}
\newcommand{\bbJ}{\mathbb{J}}
\newcommand{\bbK}{\mathbb{K}}
\newcommand{\bbL}{\mathbb{L}}

\newcommand{\bbS}{\mathbb{S}}
\newcommand{\bbT}{\mathbb{T}}

\newcommand{\bbV}{\mathbb{V}}

\newcommand{\bsT}{\boldsymbol{T}}
\newcommand{\bsS}{\boldsymbol{S}}

\newcommand{\bssi}{\boldsymbol\sigma}
\newcommand{\bsSi}{\boldsymbol\Sigma}

\numberwithin{equation}{section}
\begin{document}
\title{%
Lagrangian and Hamiltonian two-scale reduction\thanks{%
%
This work was partially supported by the
Deutsche Forschungsgemeinschaft ({D}{F}{G}) %
within the Priority Program \emph{Analysis, Modeling and Simulation
of Multiscale Problems} (SPP 1095) under Mi 459/3-3 and within {\sc
Matheon} under D14.
}%
}%
\author{ %
Johannes Giannoulis\thanks{Zentrum Mathematik,
        Technische Universit\"at M\"unchen, Boltzmannstra{\ss}e 3, D-85747 Garching bei M\"unchen,
        Germany,
        {\tt giannoulis@ma.tum.de}} \and
Michael Herrmann\thanks{Institut f\"ur Mathematik,
        Humboldt-Universit\"at zu Berlin
        Unter den Linden 6, D-10099 Berlin, Germany,  {\tt
        michaelherrmanm@math.hu-berlin.de}} \and  
Alexander Mielke\thanks{Weierstra\ss-Institut f\"ur Angewandte Analysis und
        Stochastik, Mohrenstra\ss{}e 39, 10117 Berlin,
        and Institut f\"ur Mathematik, Humboldt-Universit\"at zu Berlin,
        Unter den Linden 6, D-10099 Berlin, Germany,
        {\tt mielke@wias-berlin.de}}
   }
%
%
%
%
\date{February 20, 2008}%
\maketitle
%
%
\begin{abstract}
Studying high-dimensional Hamiltonian systems with microstructure,
it is an important and challenging problem to identify reduced
macroscopic models that describe some effective dynamics on large
spatial and temporal scales. This paper concerns the question how
reasonable macroscopic Lagrangian and Hamiltonian structures can by
derived from the microscopic system.

In the first part we develop a general approach to this problem by
considering non-canonical Hamiltonian structures on the tangent
bundle. This approach can be applied to all Hamiltonian lattices (or
Hamiltonian PDEs) and involves three building blocks: (i) the
embedding of the microscopic system, (ii) an invertible two-scale
transformation that encodes the underlying scaling of space and
time, (iii) an elementary model reduction that is based on a
Principle of Consistent Expansions.

In the second part we exemplify the reduction approach and derive
various reduced PDE models for the atomic chain. The reduced
equations are either related to long wave-length motion or describe
the macroscopic modulation of an oscillatory microstructure.
\end{abstract}

%
%
%
\section{Introduction}\label{sec:Intro}
%
A major topic in the area of multi-scale problems is the derivation
of reduced or effective macroscopic models for a given microscopic
system. A prototype for this problem is the passage from discrete
lattice systems to continuum models which describe the
\emph{effective dynamics} on much larger spatial and temporal
scales. In this case, the microscopic dynamics is governed by a high
dimensional system of ODEs, whereas  the macroscopic models are
related to the PDEs of continuum mechanics or thermodynamics.
\par%
In the dynamical setting this problem can be stated as follows:
Choosing well-ordered microscopic initial data in a specified class
of functions, one hopes that the solution will stay close to this
class of functions. We can interpret the class of functions as an
approximate invariant manifold and aim to derive reduced equations
that govern the evolution on this manifold. Moreover, if the
original dynamics is related to underlying Lagrangian or Hamiltonian
structures, the question arises how these structures behave under
the reduction procedure. This approach is closely related to the
theory of modulation equations, see
\cite{Mie02,AMSMSP:GHM,SchUec07MLPD} for surveys, which describes
how an oscillatory microstructure is modulated on the macroscopic
space--time.

\bigpar%
In mathematically rigorous terms the transition from a microscopic
to a macroscopic scale can be described by a \emph{coarse graining
diagram}, which involves the \emph{scaling parameter} $\eps$, see
Figure \ref{fig:CoarseGraining}.
\begin{figure}[ht]
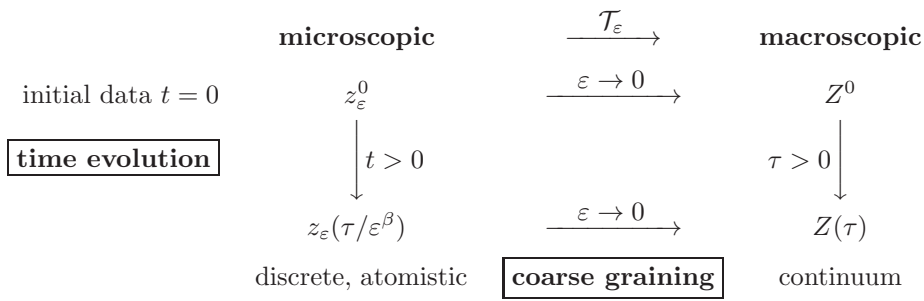

\[
\begin{array}{rccc}
& \text{\bf microscopic}&
\xrightarrow{\quad\mbox{$\calT_\eps$}\quad}%
&
 \text{\bf macroscopic} \\[0.5em]
\text{initial data }t=0&{z}_{\eps}^0& \xrightarrow{\quad\mbox{$\eps
\to 0$}\quad} &{ Z}^0
 \\[0.3em]
\text{\fbox{\bf time evolution}}&
\Bigg\downarrow\vcenter{\rlap{$t>0$}}&  &
\vcenter{\llap{$\tau>0$}} \Bigg\downarrow\\[0.3em]
&{z}_{\eps}(\tau/\eps^\beta) & \xrightarrow{\quad\mbox{$\eps \to
0$}\quad} & {Z }(\tau)
\\[0.5em]
&\text{ discrete, atomistic} &\text{ \fbox{\bf coarse graining} }
&\text{continuum}
\end{array}\qquad\qquad\mbox{ }
\]
\caption{The coarse graining diagram} \label{fig:CoarseGraining}
\end{figure}
The curve $\MiTime\mapsto{z}_\eps\at{\MiTime}\in{M_\eps}$ denotes
the solution of the microscopic model, i.e., it depends on the
microscopic time $t$, and takes values in the microscopic state
space $M_\eps$. On the other hand, the macroscopic trajectory
$\MaTime\mapsto{Z}\at{\MaTime}\in{N}$ is parametrized by the
macroscopic time $\MaTime$, and describes the evolution of the
macroscopic state $Z\at{\MaTime}\in{N}$. The two scales in this
problem are linked by a suitable \emph{two-scale ansatz}, which
consists of the time scaling $\MaTime=\eps^\beta\MiTime$, as well as
a scaling transformation $\calT_\eps:M_\eps\rra{N}$, which in
particular encodes the spatial scaling. In the best case the diagram
commutes, i.e., if the coarse graining
$z_\eps\at\MiTime\rightsquigarrow{Z}\at{\eps^\beta\MiTime}$ holds at
time $\tau=0$, then it holds true for a finite time interval
$\tau\in\ccinterval{0}{\tau_\fin}$. Any reasonable micro-macro
transition must provide an effective macroscopic evolution equation
for the macroscopic configuration $Z\at{\MaTime}\in{N}$. We can not
expect the macroscopic equation to provide exact solutions to the
microscopic system, but we can hope that it gives rise to
\emph{approximate solutions} that satisfy the microscopic law of
motion up to higher orders in $\eps$.

\bigpar
In the standard approach of model reduction one inserts a reasonable
two-scale ansatz into the microscopic law of motion and derives a
macroscopic evolution equation by means of formal expansions with
respect to the scaling parameter $\eps$.
\par%
However, this standard approach ignores the underlying Lagrangian
and Hamiltonian structures and therefore the following questions
arise naturally: $(i)$ Are there macroscopic Lagrangian and
Hamiltonian structures that correspond to the reduced macroscopic
equation? $(ii)$ If yes, how can one derive them and what is their
relation to the microscopic structures?
\bigpar%
The main issue of this paper is to develop a general framework for
micro-macro transitions that relies on a \emph{two-scale reduction}
of microscopic Lagrangian and Hamiltonian structures. To this end we
split our approach   into three steps, namely \emph{embedding},
\emph{exact two-scale transformation}, and \emph{reduction}, which
can be studied independently. Our point of view is strongly
motivated by the investigation of microscopic lattice systems, where
the micro-macro transition replaces a high dimensional system of
ODEs by a small number of macroscopic PDEs. Nevertheless, our
approach to Hamiltonian two-scale reduction can also be applied to
microscopic PDEs, see \S\ref{sec:Intro.ME} below. Note, that for us
a Hamiltonian structure consists of a Hamiltonian (function) and a
(non-canonical) symplectic form. Alternatively one could study the
reduction of Poisson structures.
\bigpar
The Hamiltonian two-scale reduction for lattices always involves the
scaling of space \emph{and} time variables. There exists a lot of
literature concerning solely the coupling of slow and fast time
scales in Hamiltonian systems with finite dimension, or fixed
spatial scales. The arising mathematical problems can be tackled by
means of averaging and adiabatic invariants, see for instance
\cite{Jar93,TT99,NV05}. Moreover, a lot of work has been done to
derive efficient schemes for the numerical integration of such
systems, compare \cite{AMSMSP:CJLL,HLW02}, and references therein.
\par%
A second class of micro-macro transitions is related to the passage
between different spatial scales. For instance, in the static case
it is a challenging problem to derive elastic energies from
atomistic lattice models, and to study the macroscopic convergence
of microscopic ground states and energies, see
\cite{FJ00,BG02a,BG02b,FT02,BG06,The06,Sch06,AMSMSP:BLM}. Another
kind of spatial reduction arises, when the microscopic model
combines both large and short space scales. Close to our point of
view, \cite{GKMS95} considers the Euler equations for an
incompressible fluid under gravity, and studies the limit of
vanishing height. It can be shown that the underlying
Poisson-structure converges to a limit that corresponds to the
shallow water equation. Moreover, using similar methods the
equations for shells and plates can be derived from the
three-dimensional models of nonlinear elasticity, see \cite{GKM96}.
%
%
\subsection{Motivating examples}\label{sec:Intro.ME}
%
Let us first discuss two simple examples related to microscopic PDEs
which highlight the essential features that arise in the general
setting. Below we will see that microscopic lattices can be treated
similarly if viewed as embedded into systems with continuous space
variable.
%
\bigpar%
The first example concerns the passage from the Boussinesq equation
to the Korteweg--de Vries (KdV) equation. Here, the microscopic
dynamics is governed by
\begin{align}
\label{MEX:kdv.Bous} x_{\MiTime\MiTime}=x_{\MiLagrC\MiLagrC}-
x_{\MiLagrC\MiLagrC\MiLagrC\MiLagrC}+x_\MiLagrC\,x_{\MiLagrC\MiLagrC},
\end{align}
where the unknown function $x$ depends on the \emph{microscopic
time} $\MiTime$ and the \emph{microscopic space variable}
$\MiLagrC\in\Rset$. Notice that $x_\MiTime$ and $x_\MiLagrC$
abbreviate $\partial_\MiTime{x}$ and $\partial_\MiLagrC{x}$,
respectively. For simplicity we ignore all boundary effects, so that
the microscopic configuration space is
$Q=L^2\at{\Rset;\,\d\MiLagrC}$. One particular macroscopic model for
\eqref{MEX:kdv.Bous} is related to the two-scale ansatz
\begin{align}
\label{MEX:kdv.TSA}
x\pair{\MiTime}{\MiLagrC}=\eps{X}\pair{\eps^3\MiTime}{\eps\at{\MiLagrC+\MiTime}}
\end{align}
where $\MaTime=\eps^3\MiTime$ and
$\MaLagr=\eps\at{\MiLagrC+\MiTime}$ denote the \emph{macroscopic}
time and space, respectively. The function $X$ is the macroscopic
configuration and for fixed $\MaTime$ it takes values in
$P=L^2\at{\Rset;\,\d\MaLagr}$. The \emph{scaling} parameter $\eps>0$
is assumed to be small and bridges the two appearing scales.
\bigpar%
The standard approach for model reduction works as follows: We plug
the two-scale ansatz \eqref{MEX:kdv.TSA} into the microscopic law of
motion \eqref{MEX:kdv.Bous},
use formal expansions with respect to $\eps$ and equate the terms of
leading order. For the example at hand one easily derives
\begin{align}
\label{MEX:kdv.KDV}
2{X}_{\MaTime\MaLagr}=-X_{\MaLagr\MaLagr\MaLagr\MaLagr}+
X_{\MaLagr}X_{\MaLagr\MaLagr},
\end{align}
which is a KdV equation for $X_\MaLagr$. As already mentioned, this
standard approach works very well but in general it is not clear at
all whether the derived macroscopic equation has its own Lagrangian
and Hamiltonian structures.
\par
We proceed with the Lagrangian and Hamiltonian two-scale reduction
for the Boussinesq example in order to illustrate the difficulties
that may arise in the general setting as well as the proposed
solutions. On the one hand, the microscopic Lagrangian $\calL$ for
\eqref{MEX:kdv.Bous} is given by $\calL=\calK-\calV$ with kinetic
energy $\calK$ and potential energy $\calV$ given by
\begin{align}
\label{MEX:kdv.Energies}
\calK\at{x_t}=\int\mhintlimits_{\Rset}\tfrac{1}{2}{x_\MiTime}^2\,\d\MiLagrC,\quad
\calV\at{x}=\int\mhintlimits_{\Rset}{\tfrac{1}{2}{x_\MiLagrC}^2+
\tfrac{1}{2}{x_{\MiLagrC\MiLagrC}}^2+\tfrac{1}{6}{x_\MiLagrC}^3}\,\d\MiLagrC.
\end{align}
Identifying the momenta $\pi=\VarDer{\calL}{x_\MiTime}$ with the
velocities $x_\MiTime$ we find that the microscopic Hamiltonian
$\calH$ equals the energy $\calE=\calK+\calV$. In particular, the
microscopic law of motion \eqref{MEX:kdv.Bous} equals the
Euler--Lagrange equations to $\calL$, and is moreover equivalent to
the canonical equations to $\calH$, which correspond to the
symplectic form
\begin{align}
\label{MEX:kdv:SymplForm1}
\Sigma=\left(%
\begin{array}{cc}
 0 & -1 \\
  1 & 0 \\
\end{array}%
\right)
\end{align}
with $1$ being the identity map $Q\to{Q}$. On the other hand, the
KdV equation \eqref{MEX:kdv.KDV} is the Euler-Lagrange equation to
$\bbL^\red\pair{X}{X_\MaTime}=\bbK^\red\pair{X}{X_\MaTime}-\bbV^\red\at{X}$
with
\begin{align*}
\bbK^\red\pair{X}{X_\MaTime}=
\int\mhintlimits_{\Rset}^{}{X_\MaTime}{X_\MaLagr}\,\d\MaLagr,\quad
\bbV^\red\at{X}=
\int\mhintlimits_{\Rset}^{}{\tfrac{1}{2}\at{X_{\MaLagr\MaLagr}}^2+
\tfrac{1}{6}\at{X_{\MaLagr}}^3}\,\d\MaLagr.
\end{align*}
Since $\bbL^\red$ depends linearly on the macroscopic velocities
$X_\tau$, the reduced macroscopic Hamiltonian structure is
non-canonical. In fact, the Hamiltonian $\bbH^\red$ equals the
potential energy $\bbV^\red$ and the symplectic structure
\begin{align}
\label{MEX:kdv:SymplForm2}
\bssi^\red\npair{\dot{X}}{\acute{X}}=\int\mhintlimits_{\Rset}^{}\dot{X}\acute{X}_\MaLagr\,\d\MaLagr
\end{align}
is a skew-symmetric 2-form on $P$. Consequently, the macroscopic law
of motion is given by
$\bssi^\red\pair{X_\MaTime}{\cdot}=\d\bbH^\red$, which is a
dynamical system on $P$, and not on $T P$ or $T^\ast{P}$.
\bigpar%
In order to describe how the microscopic Hamiltonian structure
reduces to the macroscopic one, we regard the two-scale ansatz
\eqref{MEX:kdv.TSA} as a time dependent transformation
$\bbT_\con\pair{\MiTime}{\eps}:X\in{}P\to{}Q$ with parameter $\eps$.
Its canonical lift $\bbT_\vel\pair{\MiTime}{\eps}:TP\to{}TQ$ to the
corresponding
tangent bundles reads
\begin{align}
\label{MEX:kdv:Trans}
\bbT_\vel\pair{\MiTime}{\eps}:\pair{X}{X_\MaTime}\rightsquigarrow
\pair{x}{x_\MiTime},\quad
\pair{x}{x_\MiTime}\at\MiLagrC=\pair{\eps{}X}{\eps^4X_\MaTime+\eps^2X_\MaLagr}
\at{\eps\MiLagrC+\eps\MiTime}.
\end{align}
This transformation comprises the crucial ingredients of our
approach: For fixed $\eps>0$ this transformation is \emph{exact},
this means invertible, but describes explicitly how the macroscopic
structures depend on $\eps$. Therefore, one can read-off the
\emph{effective} structures from the leading order terms in $\eps$.
\par
Applying the inverse transformation of \eqref{MEX:kdv:Trans} to the
energies from \eqref{MEX:kdv.Energies} we find
\begin{align*}
\bbK\triple{\eps}{X}{X_\MaTime}&=
\eps^{3}\int\mhintlimits_{\Rset}^{}\tfrac{1}{2}
\at{\eps^2X_\MaTime+X_\MaLagr}^2\,\d\MaLagr,\quad
\bbV\pair{\eps}{X}=\eps^{3}\tfrac{1}{2}
\int\mhintlimits_{\Rset}^{}{{X_\MaLagr}^2+\eps^2{X_{\MaLagr\MaLagr}}^2+
\tfrac{1}{3}\eps^2{X_\MaTime}^3}\,\d\MaLagr.
\end{align*}
Both transformed energies are of order $\eps^3$. However, the
transformed Lagrangian $\bbL$ is of order $\eps^5$, since the terms
of order $\eps^3$ vanish due to \emph{cancelation} via
$\bbL=\bbK-\bbV$. Thus, we find
\begin{align*}
\bbL\triple{\eps}{X}{X_\MaTime}=\eps^5\bbL^\red\pair{X}{X_\MaTime}+\DO{\eps^7}.
\end{align*}
The transformation of the Hamiltonian structure is not so simple,
since the transformation \eqref{MEX:kdv:Trans} involves a moving
frame. The macroscopic Hamiltonian $\bbH$, i.e.\ the Legendre
transform of $\bbL$, is given by
\begin{align*}
\bbH\triple{\eps}{X}{X_\MaTime}= \bbE\triple{\eps}{X}{X_\MaTime}-
\bbI\triple{\eps}{X}{X_\MaTime}=\eps^5\bbH^\red\pair{X}{X_\MaTime}+
\DO{\eps^7}.
\end{align*}
Here, $\bbE=\bbK+\bbV$ is the transformed energy and $\bbI$ is the
transform of $\calI$, where $\calI$ is the \emph{conserved quantity}
related to the moving frame by Noether's Theorem:
\begin{align*}
\calI\pair{x}{x_\MiTime}=\int\mhintlimits_\Rset\,x_\MiTime\,x_\MiLagrC\,\d\MiLagrC,\quad
\bbI\triple{\eps}{X}{X_\MaTime}=
\eps^{3}\int\mhintlimits_{\Rset}^{}\at{\eps^2X_\MaTime+X_\MaLagr}X_\MaLagr\,\d\MaLagr.
\end{align*}
We conclude that the transformation \eqref{MEX:kdv:Trans} provides
both the Lagrangian and the Hamiltonian for \eqref{MEX:kdv.KDV} to
leading order $\eps^5$. Moreover, it can be shown that the
symplectic form \eqref{MEX:kdv:SymplForm1}, considered as a $2$-form
on the \emph{tangent bundle} $TQ$ but \emph{not} on the cotangent
bundle $T^\ast{Q}$, transforms into
\begin{align*}
\bsSi=\eps^5\left(%
\begin{array}{cc}
  -2\MaLagrDer{} & 0 \\
  0 & 0 \\
\end{array}%
\right) +\eps^7\left(%
\begin{array}{cc}
  0 & -1 \\
  1 & 0 \\
\end{array}%
\right),
\end{align*}
which equals \eqref{MEX:kdv:SymplForm2} to leading order $\eps^5$.
Finally, the KdV equation is invariant under shifts in the
$\MaLagr$-direction, and this symmetry gives rise to the conserved
quantity
\begin{align*}
\bbI^\red\at{X}=\int\mhintlimits_\Rset\,X_\MaLagr^2\,\d\MaLagr,
\end{align*}
which turns out to be the lowest order expansions of
$\bbK(\eps,\cdot)$ and $\bbV(\eps,\cdot)$, namely
\begin{align*}
\eps^3\,\bbI^\red\at{X}=2\,\bbK\triple{\eps}{X}{X_\MaTime}+\DO{\eps^5}=
2\,\bbV\triple{\eps}{X}{X_\MaTime}+\DO{\eps^5}=
\bbI\triple{\eps}{X}{X_\MaTime}+\DO{\eps^5}.
\end{align*}
We conclude that the terms which vanish due to cancelation
correspond to a macroscopic integral of motion.
\bigpar%
As a second motivating example we study the macroscopic evolution of
a modulated pulse in the Klein-Gordon (KG) equation
\begin{equation}\label{KGpde}
x_{tt}=x_{\MiLagrC\MiLagrC} -\OSPot^\prime(x)
\end{equation}
with $x=x\pair{t}{\MiLagrC}$, $\MiLagrC\in\Rset$, and nonlinear
on-site potential $\OSPot$. A modulated pulse is an (approximate)
solution which satisfies the ansatz
\begin{align}\label{KGpdeNLSansatzExplicit} %
x(\MiLagrC,{\MiTime})
=\eps{A}\pair{\eps^2\MiTime}{\eps\MiLagrC-\eps{}c\MiTime}
\mhexp{\iu\at{\om\MiTime+\theta\MiLagrC}} +\cc + \DO{\eps^2}.
\end{align}
Here $\cc$ denotes the complex conjugate, the frequency $\om$ and
the wave number $\theta$ are fixed parameters, and $c$ is the
moving-frame speed. The plane waves
$\mhexp{\iu\at{\om\MiTime+\theta\MiLagrC}} $ in
\eqref{KGpdeNLSansatzExplicit} describe an \emph{oscillatory
microstructure} whose amplitude $A$ is modulated on the macroscopic
scale $\MaTime=\eps^2\MiTime$ and
$\MaLagr=\eps\at{\MiLagrC-c\MiTime}$.
\par%
A first necessary condition for \eqref{KGpdeNLSansatzExplicit} to
yield approximate solutions is that $\om$ and $\theta$ satisfy the
dispersion relation $\om^2=\theta^2+\OSPot^{\prime\prime}(0)$ and
$c=-\om^\prime$ is the associated \emph{group velocity}. Moreover,
the complex-valued amplitude $A$ must satisfy the nonlinear
Schr\"odinger (nlS) equation
\begin{align}
\notag
2\om\iu   A_\MaTime=\varrho_1 A_{\MaLagr \MaLagr}
-\varrho_2\abs{A}^2A,
\end{align}
where the constants $\varrho_1$ and $\varrho_2$ can be computed
explicitly. The validity of this macroscopic model has been proven
rigorously in \cite{KSM92} on the level of the equation of motion.
\par%
As in the Boussinesq example, both the microscopic and macroscopic
models have Lagrangian and Hamiltonian structures and so we are
interested in the question how these are related to each other. The
new feature in this example is the presence of microscopic
oscillations and the key idea is to introduce an additional
one-dimensional, periodic \emph{phase variable}
$\Phase\in{T^1}\simeq[0,\,2\pi]$. This new degree of freedom enables
us to find a suitable two-scale transformation such that \emph{all}
(transformed) oscillations are confined in the phase direction
$\Phase$. This suggests the two-scale ansatz
\begin{align}%
\label{KGpdeNLSansatzGeneral} %
x({\MiTime},\MiLagrC,\Phase)
=\eps{X}(\eps^2\MiTime,\eps(\MiLagrC-c\MiTime),\Phase+
\om\MiTime+\theta\MiLagrC),
\end{align}
which is similar to \eqref{KGpdeNLSansatzExplicit} but gives rise to
an invertible two-scale transformation.
\bigpar
The introduction of $\Phase$ can be viewed as an \emph{embedding} of
the microscopic system, such that \eqref{KGpde} becomes
\begin{align}
\notag
x_{\MiTime\MiTime}\triple{\MiTime}{\MiLagrC}{\Phase}=x_{\MiLagrC\MiLagrC}\triple{\MiTime}{\MiLagrC}{\Phase}
-\OSPot^\prime(x\triple{\MiTime}{\MiLagrC}{\Phase}).
\end{align}
This embedding does not affect the microscopic dynamics, since
$\Phase$ appears just as a parameter. The embedded system has
Lagrangian $\calL=\calK-\calV$ and Hamiltonian
$\calH=\calE=\calK+\calV$ with
\begin{align}
\label{KGpde.Energies} %
\calK(x_t)=\int\mhintlimits_{\Rset\times
T^1}\tfrac{1}{2}x_t^2\,\d\MiLagrC\d\Phase, \quad
\calV(x)=\int\mhintlimits_{\Rset\times
T^1}\tfrac12x_\MiLagrC^2+\OSPot(x)\,\d\MiLagrC\d\Phase
\end{align}
and corresponds to the symplectic form \eqref{MEX:kdv:SymplForm1}.
Moreover, we find two continuous symmetry groups related to shifts
with respect to $\MiLagrC$ and $\Phase$, which by Noether's theorem
correspond to the conserved quantities (integrals of motion)
\begin{align}
\label{KGpde.IntOfMotion} %
\calI_{\mathrm{space}}\pair{x}{x_\MiTime}=
\int\mhintlimits_{\Rset\times
T^1}x_\MiTime\,x_\MiLagrC\,\d\MiLagrC\d\Phase, \qquad
\calI_{\mathrm{phase}}\pair{x}{x_\MiTime}=
\int\mhintlimits_{\Rset\times
T^1}x_\MiTime\,x_\Phase\,\d\MiLagrC\d\Phase.
\end{align}
The second integral of motion arises only due to the embedding but
plays a prominent role in the two-scale reduction, since it is
needed for the derivation of the macroscopic Hamiltonian. In fact,
the moving frame in \eqref{KGpdeNLSansatzGeneral} involves drifts in
space and phase direction and the associated integral of motion
reads
\begin{align}
\notag%
\calI=-c\,\calI_{\mathrm{space}}+\om\,\calI_{\mathrm{phase}}.
\end{align}
Like for the Boussinesq example, we can use the transformation
\eqref{KGpdeNLSansatzGeneral} and our general approach described
below in order to derive the Lagrangian and Hamiltonian structures
for the nlS equation directly from their microscopic counterparts.
It comes out, that the leading order terms determine the
microstructure, the next-leading order terms give the moving frame
speed, and finally, the next-next leading order terms provide the
macroscopic law of motion. This will be explained in detail within
\S\ref{sec:ChainNew.nls}.
%
%
\subsection{General approach to
Lagrangian and Hamiltonian two-scale reduction}
%

%
The concepts arising in the above examples can be generalized to the
following abstract framework for a Lagrangian and Hamiltonian
two-scale reduction.
\bigpar%
The first step concerns the
{\bf{embedding}} %
%
of the microscopic system. We have seen above that the treatment of
models with microstructure requires the introduction of new phase
variables $\Phase$. Moreover, for discrete models like  chains we
replace the particle index $\MiLagr\in\Zset$ by a continuous
variable $\MiLagrC\in\Rset$. In all cases this embedding does not
change the microscopic dynamics, but it gives rise to new
\emph{continuous symmetry groups} and hence to additional integrals
of motion which contribute to the macroscopic Hamiltonian.
\par%
In what follows we always consider the Lagrangian $\calL$ of the
embedded system which is defined on the tangent bundle $TQ$ of the
microscopic configuration space $Q$. Then there exists an equivalent
Hamiltonian structure on the \emph{cotangent bundle} $T^\ast{Q}$
corresponding to the canonical symplectic form. However, for the
reduction step explained below it is essential to consider a
Hamiltonian $\calH$ as well as a symplectic form $\sigma$ both of
which are defined on the \emph{tangent bundle} $TQ$. To this end we
pull back the canonical structure from $T^\ast{Q}$ to $TQ$ via the
fiber derivative of $\calL$. This will be discussed in detail in
\S\ref{sec:HamStruct}.

\bigpar
The most important step in any two-scale reduction is the %
{\bf{transformation}} %
of the embedded system. For this purpose we introduce two-scale
transformations by composing elementary building blocks such as
\emph{(weak) symmetry transformations}, \emph{moving-frame
transformations}, and \emph{scaling transformations}. The first two
building blocks are well understood in classical mechanics, whereas
our concept of scaling transformations seems to be new, since it
involves the scaling of space \emph{and} time. The starting point
for any scaling transformation is a map $\calS_\con:Q\rra{P}$
bridging the microscopic and the macroscopic configuration spaces
$Q$ and $P$. The definition of such a map involves only the scalings
of the space coordinates, but its lift $\calS_\vel:TQ\rra{TP}$ to
the tangent bundles takes into account also the time scaling.
\par%
Two-scale transformations are in the heart of any two-scale
reduction, because they provide a macroscopic Lagrangian $\bbL$, a
macroscopic Hamiltonian $\bbH$, and a symplectic form $\bssi$ (all
defined on $TP$), which depend explicitly on the scaling parameter
$\eps$.
\bigpar%
The %
{\bf{reduction}} %
%
step starts with the formal expansions of the transformed Lagrangian
and Hamiltonian structures with respect to the scaling parameter
$\eps$, i.e.,
\begin{align*}
\bbL\at{\eps}=\eps^\kappa\,\at{\bbL_0+\eps\bbL_1+
\eps^2\bbL_2+\tdots},
\end{align*}
and
\begin{align*}
\bbH\at{\eps}=\eps^\kappa\,\at{\bbH_0+\eps\bbH_1+
\eps^2\bbH_2+\tdots},\quad
\bssi\at{\eps}=\eps^\kappa\,\at{\bssi_0+\eps\bssi_1+
\eps^2\bssi_2+\tdots}.
\end{align*}
A key feature of our approach is the \emph{Principle of Consistent
Expansions} which will be proved in \S\ref{sec:HamStruct} and
guarantees that $\pair{\bbH_i}{\bssi_i}$ is the Hamiltonian
structure corresponding to the Lagrangian $\bbL_i$. For this
principle to hold it is crucial to consider the Hamiltonian
structure on the \emph{tangent} (and not on the \emph{cotangent})
bundle.
\par%
For some examples the leading order Lagrangian $\bbL_0$ is
non-degenerate. Then the effective macroscopic model is completely
determined already by the leading order terms. However, whenever the
two-scale ansatz involves an oscillatory microstructure the leading
order terms turn out to be degenerate in the following sense: The
leading order Lagrangian $\bbL_0$ is quasi-stationary, i.e., it does
not depend on $X_\MaTime$, and this implies $\bbH_0=-\bbL_0$ and
$\bssi_0=0$. Moreover, there exists a sub-manifold $P_0$ of $P$ such
that the gradient of $\bbL_0$ vanishes on $P_0$. In this case we
restrict $\bbL-\bbL_0$ and $\bbH-\bbH_0$ and $\bssi$ to $TP_0$, and
derive the effective macroscopic model by expanding the restricted
structures.
\bigpar
The reduction procedure concerns the \emph{convergence of Lagrangian
and Hamiltonian structures} as $\eps\rra0$, but this does \emph{not}
necessarily imply the \emph{convergence of solutions}. Therefore
each reduced model must be \emph{justified}. In the general setting
the justification problem turns out to be very subtle and is not
addressed in this paper. However, for all examples presented here we
discuss the corresponding justification problem after having derived
the reduced Lagrangian and Hamiltonian structures. We also refer to
the surveys \cite{Mie02,AMSMSP:GHM,SchUec07MLPD} and to \cite{Mie07}
for an abstract theory using $\Gamma$-convergence for Hamiltonian
systems.
\bigpar%
The abstract framework for the two-scale reduction will be developed
in detail within \S\ref{sec:FoundNew}, where we prove the
transformation rules for Lagrangian and Hamiltonian structures and
discuss the reduction procedure in the various cases. Finally, in
\S\ref{sec:ChainNew} we apply this method to several micro-macro
transitions for the atomic chain.
%
%
\subsection{Two-scale reductions for the atomic chain}\label{sec:Intro.ExChain}
%
%
%
The nonlinear atomic chain consists of identical particles with unit
mass. The atoms are coupled to a background field by the
\emph{on-site potential} $\OSPot$ and nearest neighbors interact via
the \emph{pair potential} $\NNPot$. The microscopic dynamics is
governed by Newton's equations
\begin{align}
\label{Intro:AtomicChain} %
\ddot{x}_\MiLagr\at\MiTime= \NNPot^\prime\bat{{x}_{\MiLagr+1}
\at\MiTime-{x}_{\MiLagr}\at\MiTime}-
\NNPot^\prime\bat{{x}_{\MiLagr}\at\MiTime-
{x}_{\MiLagr-1}\at\MiTime}-\OSPot^\prime\at{{x}_{\MiLagr}\at\MiTime},
\end{align}
where $\MiLagr\in\Zset$ is the discrete particle index and
$x_\MiLagr\at{\MiTime}\in\Rset$ denotes the displacement of the
$j$-{t}{h} particle at time $\MiTime$. For $\OSPot\equiv0$ and
an-harmonic $\NNPot$ we obtain the Fermi--Pasta--Ulam (FPU) chain,
while Klein-Gordon (KG) chains correspond to harmonic $\NNPot$ but
have an-harmonic $\OSPot$.
\bigpar%
A general micro-macro transition for the atomic chain is related to
the two-scale ansatz
\begin{align}
\notag
x_\MiLagr\at\MiTime= \eps^\alpha{X}\pair{\eps^\beta\MiTime}
{\eps\at{\MiLagr-c\,\MiTime}}
\end{align}
with macroscopic time $\tau=\eps^\beta\MiTime$, macroscopic particle
index $\MaLagr=\eps\at{\MiLagr-c\MiTime}$ and macroscopic
configuration $X$. Notice that $\MaLagr$ is assumed to be a
continuous variable and can be interpreted as the \emph{coordinate
of a macroscopic material point}.
\bigpar
In the example part \S\ref{sec:ChainNew} we study the following
micro-macro transitions and discuss how the Lagrangian and
Hamiltonian structures that correspond to the effective macroscopic
equations can be derived directly from the Lagrangian and
Hamiltonian structure of the atomic chain. To this end we embed the
atomic chain \eqref{Intro:AtomicChain} into a microscopic system
with continuous particle index $\eta\in\Rset$, see
\S\ref{sec:Chain.Embed}.
%
%
\paragraph{Quasi-linear wave equation}
%
%
In \S\ref{sec:ChainNew.WE} we consider the FPU chain and rely on the
two-scale ansatz
\begin{align}
\label{Intro:WE.MSAnsatz}%
x_{\MiLagr}\at{\MiTime}=\eps^{-1}{X}\pair{\eps\MiTime}
{\eps\MiLagr},
\end{align}
which has no moving frame and corresponds to the \emph{hyperbolic
scaling} $\MaTime=\eps\MiTime$ and $\MaLagr=\eps\MiLagr$. In this
case the macroscopic evolution satisfies the nonlinear wave equation
\begin{align*}
\MaTimeDerS{}X-\MaLagrDer{}\Bat{\NNPot^\prime\at{\MaLagrDer{X}}}=0.
\end{align*}
%
%
\paragraph{KdV equation}
%
The second example, see \S\ref{sec:ChainNew.kdv}, concerns the
passage from FPU chains to a KdV equation by means of a two-scale
ansatz similar to \eqref{MEX:kdv.TSA}.
%
\paragraph{Modulated pulses and the nlS equation}
%
%
In analogy to the second motivating example, in
\S\ref{sec:ChainNew.nls} we study the macroscopic evolution of a
modulated pulse in the KG \emph{chain}. Similar to above, the
two-scale ansatz reads
\begin{align}
\label{Intro:nls.MSAnsatz}%
x_j\at\MiTime
=\eps{A}\pair{\eps^2\MiTime}{\eps\MiLagr-\eps{}c\MiTime}
\mhexp{\iu\at{\om\MiTime+\theta\MiLagr}} +\cc
\end{align}
and the macroscopic dynamics is described by an nlS equation. The
only difference as compared to the case of the \emph{continuous} KG
equation \eqref{KGpde} concerns the dispersion relation leading to
different coefficients in the macroscopic equation.
%
%
\paragraph{Three-wave-interaction}
%
%
The fourth example, see \S\ref{sec:ChainNew.threeWaves}, is the most
involved one and describes how three modulated pulses interact if
they are in resonance. This gives rise to the following ansatz
\begin{align}
\label{Intro:TW.MSAnsatz} %
x_{\MiLagr}\at{\MiTime} = \eps\sum_{n=1}^3 {A_n}\pair{\eps\MiTime}
{\eps\MiLagr}\mhexp{\iu\at{\om_n\MiTime+\theta_n\MiLagr}}+\cc
\end{align}
with three phases $\Phase_n=\om_n\MiTime+\theta_n\MiLagr$ and three
amplitudes $A_n$. All pairs $p_n=\pair{\theta_n}{\om_n}$ satisfy the
dispersion relation of the KG chain and are coupled via the
\emph{resonance condition} $p_1+p_2+p_3=0$ in
$ T^1 \times \Rset $,
where $T^k =\Rset^k/_{(2\pi\Zset)^k} $ is the $k$-dimensional torus.
This resonance condition shows that we have only two independent
phases. Moreover, the amplitudes are coupled on the hyperbolic
scaling $\MaTime=\eps\MiTime$, $\MaLagr=\eps\MiLagr$ via the
\emph{three-wave-interaction equations}
\begin{align}
\label{Intro:TW.MSequations} %
 \iu  \!
\begin{pmatrix}2\om_1\!\!&0&0\\0&\!\!2\om_2\!\!&0\\0&0&\!\!2\om_3\end{pmatrix}
\partial_\tau\!\!\begin{pmatrix}A_1 \\ A_2\\ A_3 \end{pmatrix}
=
\iu \! \begin{pmatrix}\!2\om_1\om_1^\prime\!\!&0&0\\
0&\!\!2\om_2\om_2^\prime\!\!&0\\0&0&\!\!2\om_3\om_3^\prime\!\end{pmatrix}
\partial_y\!\begin{pmatrix} A_1 \\ A_2\\ A_3 \end{pmatrix}
- \OSPot^{\prime\prime\prime}(0)
\begin{pmatrix}\mhol{A}_2\mhol{A}_3 \\
    \mhol{A}_1\mhol{A}_3 \\ \mhol{A}_1\mhol{A}_2 \end{pmatrix}.
 \end{align}
Finally, in \S\ref{sec:ChainNew.FurtherEx} we present further
examples for micro-macro transitions in the atomic chain. Although
they fit into the general framework they are postponed to a
forthcoming paper as they display additional complications.
%
%

\section{Lagrangian and Hamiltonian two-scale reduction}\label{sec:FoundNew}
%
%
In this section we describe the general framework for the two-scale
reduction of Lagrangian and Hamiltonian structures and present our
abstract results concerning two-scale transformations and the
problem of model reduction. Since we are mainly interested in
Hamiltonian PDEs and lattices we assume that the microscopic
configuration space $Q$ is a function space. Moreover, for
simplicity we suppose $Q$ to be a Hilbert space (usually some
$L^2$--space) with inner product $\skp{\cdot}{\cdot}$.
\par%
As a prototypical example for a microscopic Lagrangian $\cal$ we
consider a \emph{normal system}, where the Lagrangian $\calL$ is the
difference of quadratic kinetic energy $\calK$ and potential energy
$\calV$. More precisely, a normal Lagrangian $\calL$ satisfies
\begin{align} \label{TypEx:Normal.Lagrangian}
\calL\pair{x}{x_\MiTime}=\calK\at{x_\MiTime}-\calV\at{x},\qquad
\calK\at{x_\MiTime}=\tfrac{1}{2}\dualpair{x_\MiTime}{M\,x_\MiTime}
\end{align}
with symmetric mass matrix $M:Q\rra{Q}$. However, our approach is
not restricted to normal systems but can be applied to all
microscopic Lagrangian structures.
%
%
We start with some general remarks regarding Hamiltonian structures.
%
\subsection{Hamiltonian structures for given
Lagrangian}\label{sec:HamStruct}
%
%
In classical mechanics we have (at least) two possibilities to
introduce a Hamiltonian structure for a given Lagrangian
$\calL:TQ\rra\Rset$, where a Hamiltonian structure consists of both
a Hamiltonian (function) and a symplectic form.
\par%
The standard approach is related to the \emph{canonical Hamiltonian
structure} on the \emph{cotangent bundle} $T^\ast{Q}$. For its
definition we consider the Legendre transform $\calH:TQ\rra\Rset$ of
$\calL$ which is defined by
\begin{math}
\calH\pair{x}{x_\MiTime}=
\skp{\pi\pair{x}{x_\MiTime}}{x_\MiTime}-\calL\pair{x}{x_\MiTime}.
\end{math} %
Here, $\pi\pair{x}{x_\MiTime}=\mathfrak{F}\calL|_x\at{x_\MiTime}$ is
the canonical momentum associated to $x_\MiTime$ and is determined
by the \emph{fiber derivative} of the Lagrangian $\calL$. This fiber
derivative is given by
\begin{align*}
\mathfrak{F}\calL:TQ\rra{}T^{\ast}Q,\quad\pair{x}{x_\MiTime}
\mapsto\pair{x}{\VarDer{\calL}{x_\MiTime}\pair{x}{x_\MiTime}}=
\pair{x}{\pi\pair{x}{x_\MiTime}}.
\end{align*}
In the next step we replace the velocity $x_\MiTime$ by $\pi$,
assuming this is possible, and rewriting $\calH:TQ\rra\Rset$ in
terms of $x$ and $\pi$ we obtain the \emph{canonical Hamiltonian}
\begin{align*}
\mhol{\calH}:T^{\ast}Q\rra\Rset,\quad\mhol{\calH}\pair{x}{\pi}=
\calH\at{{\mathfrak{F}\calL}^{-1}\pair{x}{\pi}}.
\end{align*}
The Lagrangian equation to $\calL$, i.e.
\begin{math}
\tfrac{\d}{\d\MiTime}\pi\pair{x}{x_\MiTime}=\VarDer{\calL}{x}\pair{x}{x_\MiTime},
\end{math} %
is equivalent to the canonical equations
$x_\MiTime=\VarDer{\mhol{\calH}}{\pi}$,
$\pi_\MiTime=-\VarDer{\mhol{\calH}}{x}$, which can be written as
\begin{align*}
\mhol{\sigma}_\can|_{\mhol{z}}\pair{\mhol{z}_\MiTime}{\cdot}=
\d\mhol{\calH}|_{\mhol{z}}\at{\cdot}.
\end{align*}
Here, $\mhol{z}=\pair{x}{\pi}\in{}T^\ast{Q}$, and $\mhol{\si}_\can$
denotes the \emph{canonical symplectic form} on $T^{\ast}{Q}$ given
by
\begin{align}
\notag
\mhol\sigma_\can|_{\mhol{z}}\pair{\dot{\mhol{z}}}{\acute{\mhol{z}}}=
\skp{\acute{\pi}}{\dot{x}}- \skp{\dot{\pi}}{\acute{x}},
\end{align}
with $\dot{\mhol{z}}=\pair{\dot{x}}{\dot{\pi}}$ and
$\acute{\mhol{z}}=\pair{\acute{x}}{\acute{\pi}}$ being two
independent tangent vectors from $\tspace{\mhol{z}}{T^{\ast}Q}$.
\par
The second Hamiltonian structure lives on the \emph{tangent bundle}
$TQ$ and consists of the Hamiltonian $\calH:TQ\rra\Rset$ and a
non-canonical symplectic form $\sigma\in\Lambda^2\at{TQ}$ defined as
the pull-back of $\mhol{\sigma}_\can$ via $\mathfrak{F}\calL$, i.e.\
$\sigma=\at{\mathfrak{F}\calL}^\ast\mhol{\sigma}_\can$. This means
\begin{align}
\label{GS:CanSymplForm} \si|_z\pair{\dot{z}}{\acute{z}}=
\skp{D\pi|_z\at{\acute{z}}}{\dot{x}}-
\skp{D\pi|_z\at{\dot{z}}}{\acute{x}},
%
\end{align}
where $z=\pair{x}{x_\MiTime}$, $\dot{z},\acute{z}\in\tspace{z}{TQ}$,
and $D\pi|_z$ is the linearization of $\pi$ in $z$. Assuming that
$\mathfrak{F}\calL$ is differentiable, it can be shown, see
\cite{AM78} for a proof, that the Lagrangian equation for $\calL$ is
equivalent to the Hamiltonian system
\begin{align}
\label{GS:HamEqns} %
\sigma|_{z}\pair{z_\MiTime}{\cdot}=\d\calH|_{z}\at{\cdot}.
\end{align}
%
\begin{remark}%
The symplectic form $\si$ can be identified with a family
\begin{align*}
\Si:TQ\rra\mathrm{Lin}\at{Q{\times}Q,\,Q{\times}Q}
\end{align*}
of skew-symmetric and operator-valued matrices such that
\begin{math}
\sigma|_{z}\pair{\dot{z}}{\acute{z}}=
\skp{\Sigma|_{z}\dot{z}}{\acute{z}}_{Q{\times}{Q}}
\end{math}
for all states $z=\pair{x}{x_\MiTime}\in{Q}\times{Q}$ and arbitrary
tangent vectors $\dot{z},\,\acute{z}\in{Q}\times{Q}$. The components
$\Sigma_{\,{i}{j}}$, $i,j=1,2$, of $\Sigma$ are linear operators
$Q\rra{Q}$ and satisfy
$\Sigma_{\,{i}{j}}^\trans=-\Sigma_{\,{j}{i}}$. Consequently, the
Hamiltonian system \eqref{GS:HamEqns} is equivalent to
\begin{align*}
\Sigma|_{\pair{x}{x_\MiTime}}\,%
\frac{\d}{\d\MiTime}\begin{pmatrix}x\\x_{\MiTime}%
\end{pmatrix}=
\begin{pmatrix}\VarDer{\calH}{x}\pair{x}{x_\MiTime}\\%
\VarDer{\calH}{x_\MiTime}\pair{x}{x_\MiTime}%
\end{pmatrix}.
\end{align*}
\end{remark}
\begin{example}
\label{GS:Ex.Lagrangian}
On each Hilbert space $Q$ we can define the \emph{metric Lagrangian}
$\calL_\mtrc$ by
\begin{math}
\calL_\mtrc\pair{x}{x_\MiTime}=
\tfrac{1}{2}\skp{x_\MiTime}{x_\MiTime}.
\end{math} %
This implies $\calH_\mtrc=\calL_\mtrc$ and
\begin{align*}
\sigma_\mtrc|_\pair{x}{x_\MiTime}
\Bpair{\pair{\dot{x}}{\dot{x}_\MiTime}}
{\pair{\acute{x}}{\acute{x}_\MiTime}}=
\skp{\acute{x}_\MiTime}{\dot{x}}-
\skp{\dot{x}_\MiTime}{\acute{x}},\quad
\Sigma_\mtrc|_{\pair{x}{x_\MiTime}}=
\begin{pmatrix}
\scriptstyle0&\scriptstyle{-1}\\\scriptstyle1&\scriptstyle0%
\end{pmatrix},
\end{align*}
where $1$ denotes the identity map $Q\rra{Q}$. In what follows we
refer to $\si_\mtrc$ and $\Si_\mtrc$ as the \emph{metric symplectic
form} on $TQ$. Moreover, for a normal Lagrangian with
\eqref{TypEx:Normal.Lagrangian} we find
$\calH\pair{x}{x_\MiTime}=\calK\at{x_\MiTime}+\calV\at{x}$ as well
as
\begin{align*}
\Sigma=
\begin{pmatrix}\scriptstyle0&\scriptstyle{-M}\\\scriptstyle{\scriptstyle{M}}&\scriptstyle0\end{pmatrix}=
M\Sigma_{\mtrc},
\end{align*}
where we used $M=M^\trans$.
%
%
%
\end{example}

The tangent-bundle approach to Hamiltonian structures is more
general than the canonical one via the cotangent bundle, because it
works even if the map $x_\MiTime\mapsto\pi$ is not invertible, but
has the disadvantage that the symplectic form $\sigma$ depends
explicitly on the Lagrangian $\calL$. Consequently, the Hamiltonian
equations on $TQ$ do not arise in canonical form.
%
%
%
For the examples from \S\ref{sec:Intro.ME} and
\S\ref{sec:Intro.ExChain} we find $\pi\pair{x}{x_\MiTime}=x_\MiTime$
so that the Hamiltonian structures on $TQ$ and $T^\ast{Q}$ seem to
be equal. However, both structures \emph{transform differently}
under scaling transformations, see Principle \ref{POCE} and
\S\ref{sec:Trafos}.
\par%
For a first motivation why we prefer the tangent-bundle and avoid
the cotangent-bundle structures, let us study \emph{trivial
scalings}: Given a Lagrangian $\calL$ on $TQ$, we consider the
scaled Lagrangian $\calL_\eps=\eps\calL$, where $\eps>0$ is some
artificial small constant. The scaling of $\calH$ is given by
\begin{math}
\calH_\eps= %
\skp{\VarDer{\calL_\eps}{x_\MiTime}}{x_\MiTime}-
\calL_\eps=\eps\calH,
\end{math}
and similarly we find $\sigma_\eps=\eps\sigma$. On the other hand,
the standard (canonical) approach applied to $\calL_\eps$ yields
$\mhol{\calH}_\eps\pair{x}{\pi}=\eps\mhol\calH\pair{x}{\eps^{-1}\pi}$,
and the canonical equations
\begin{align}
\label{GS:ScaledCanonicalEqns}
x_\MiTime=+\,\VarDer{\mhol{\calH}_\eps}{\pi}\pair{x}{\pi}=
+\,\VarDer{\mhol{\calH}}{\pi}\pair{x}{\eps^{-1}\pi},\quad
\pi_\MiTime=-\,\VarDer{\mhol{\calH}_\eps}{x}\pair{x}{\pi}=
-\,\eps\VarDer{\mhol{\calH}}{x}\pair{x}{\eps^{-1}\pi}
\end{align}
again correspond to $\mhol{\si}_\can$, which does not depend on
$\eps$. Of course, as long as $\eps$ is fixed, both formulations are
completely equivalent, since we can replace $\pi$ by $\eps\pi$ in
\eqref{GS:ScaledCanonicalEqns}. However, if we try to identify
leading order dynamics by expansions in powers of $\eps$ we obtain
very different results. In fact, $\cal{L}_\eps$, $\calH_\eps$ and
$\sigma_\eps$ scale in the same way and, hence, $\eps$ drops out in
both the corresponding Lagrangian and Hamiltonian equations on $TQ$.
On the other hand, for a normal Lagrangian, as given in
\eqref{TypEx:Normal.Lagrangian}, we find
\begin{math}%
\mhol{\calH}_\eps=
\tfrac{1}{2\eps}\skp{\pi}{M^{-1}\pi}+\eps\calV\at{x}
\end{math}
and the formal expansion of $\mhol{\calH}_\eps$ gives
$\mhol{H}^\red\at{\pi}=\tfrac{1}{2\eps}\skp{\pi}{M^{-1}\pi}$ as
``leading order'' Hamiltonian on $T^{\ast}Q$. In particular, the
corresponding canonical equations  $x_\MiTime=\eps^{-1}M^{-1}\pi$
and $\pi_\MiTime=0$ do not recover the original dynamics.
\bigpar%
More generally, the key difference between tangent and cotangent
Hamiltonian structures is related to the following {\emph{\bf
Principle of Consistent Expansions}}:
\begin{principle}\label{POCE}
Suppose that the Lagrangian $\calL$ obeys a (formal) expansion in
powers of a parameter $\eps$, i.e.,
\begin{align}
\label{POCE:Eqn1} \calL\at{\eps}=\eps^\kappa\,\at{\calL_0+
\eps\calL_1+\eps^2\calL_2+\tdots}.
\end{align}
Then the Hamiltonian structure on $TQ$ obeys a corresponding
expansion
\begin{align*}
\calH\at{\eps}=\eps^\kappa\,\at{\calH_0+\eps\calH_1+
\eps^2\calH_2+\tdots},\quad
\si\at{\eps}=\eps^\kappa\,\at{\si_0+\eps\si_1+ \eps^2\si_2+\tdots}
\end{align*}
and all expansions are \emph{consistent}. This means, for each order
$\eps^i$ we have
\begin{align*}
\calH_{i}=\dualpair{\partial_{x_\MiTime}\calL_i}{x_\MiTime}-\calL_i,\qquad
\si_{i}=\at{\mathfrak{F}\calL_i}^\ast\,\mhol{\sigma}_\can.
\end{align*}
\end{principle}
\begin{proof}
Since the fiber-derivative operation acts linearly on the Lagrangian
we find
\begin{align*}
\mathfrak{F}\calL=\eps^\kappa\,\at{\mathfrak{F}\calL_0+\eps\mathfrak{F}\calL_1+
\eps^2\mathfrak{F}\calL_2+\tdots}
\end{align*}
and this implies both the existence and consistency of the expansion
of the Hamiltonian structure.
\end{proof}
The validity of Principle \ref{POCE} is a remarkable property of the
Hamiltonian structure on $TQ$ and has no analogue on $T^\ast{Q}$. In
fact, \eqref{POCE:Eqn1} implies a consistent expansion for the
canonical momentum $\pi$, i.e.\
\begin{math}
\pi\at{\eps}=\eps^\kappa\,\at{\pi_0+ \eps\pi_1+\eps^2\pi_2+\tdots}
\end{math} with $\pi_i=\VarDer{\calL_i}{x_\MiTime}$, but replacing
$x_\MiTime$ by $\pi$ we normally end up with a non-consistent
expansion for the Hamiltonian $\mhol{\calH}$ on $T^\ast{Q}$.
\bigpar%
In the context of this paper we do not apply Principle \ref{POCE} to
the microscopic Lagrangian or Hamiltonian structures, since usually
these do not depend on scaling parameters. However, the two-scale
transformations introduced in \S\ref{sec:Trafos} strongly depend on
$\eps$ and so do the \emph{transformed} Lagrangian and Hamiltonian
structures. Thus, for the purpose of model reduction the tangent
framework turns out to be very convenient as it provides the
consistency of the Lagrangian and Hamiltonian structures for
\emph{all} powers of $\eps$.
%
%
%
%
%
%
%
%
%
\subsection{Exact two-scale transformations}\label{sec:Trafos}
%
%
%
As mentioned in the introduction, any micro-macro transition relies
on an exact two-scale transformation which obviously changes the
Lagrangian and Hamiltonian structures. All of the two-scale
transformations considered in this paper are superpositions of
elementary building blocks, namely
\begin{enumerate}
\item (weak) symmetry transformations,%
\item moving-frame transformations,%
\item scalings of space and time coordinates.
\end{enumerate}
In this section we aim to describe how each of these building blocks
transforms the Lagrangian and Hamiltonian structures on $TQ$. The
concepts of symmetry and moving-frame transformations are well
established in the theory of Hamiltonian systems, but since they are
usually studied on the cotangent bundle we start with the
reformulation of standard results.
%
%
%
\subsubsection{Linear transformations}
%
%
%
Let $\calT_\con:Q\rra\tilde{Q}$ be a linear isomorphism between
${Q}$ and another Hilbert space $\tilde{Q}$ with inverse
$\tilde{\calT}_\con:\tilde{Q}\rra{Q}$. The canonical lifts of
$\calT_\con$ and $\tilde{\calT}_\con$ to the tangent bundles are
denoted by $\calT_\vel:TQ\rra{T\tilde{Q}}$ and
$\tilde{\calT}_\vel:T\tilde{Q}\rra{TQ}$, respectively, and satisfy
$\calT_\vel\pair{x}{x_\MiTime}=\pair{\calT_\con{x}}{\calT_\con{x}_\MiTime}$
as well as $\tilde{\calT}_\vel=\nnat{\calT_\vel}^{-1}$.
\begin{remark}
In what follows we use the inverse transformation
$\tilde{\calT}_\vel$ in order to pull back \emph{forms} from $Q$
(the pull-back with respect to $\tilde{\calT}_\vel$ is the
push-forward with respect to $\calT_\vel$). In particular, we pull
back functions $\calF$ ($0$-forms) and symplectic forms $\sigma$
($2$-forms). The images under this operation are denoted by
$\tilde{\calF}=\nnat{\tilde{\calT}_\vel}^\ast\calF$ and
$\tilde{\sigma}=\nnat{\tilde{\calT}_\vel}^\ast\sigma$, and satisfy
%
\begin{align*}
\tilde{\calF}\at{\tilde{z}}=\calF\nnat{\tilde{\calT}_\vel\tilde{z}},\quad\quad
\tilde{\sigma}|_{\tilde{z}}\npair{\dot{\tilde{z}}}{\acute{\tilde{z}}}=
\sigma|_{\tilde{\calT}_\vel\tilde{z}}\npair{\tilde{\calT}_\vel\dot{\tilde{z}}}{\tilde{\calT}_\vel\acute{\tilde{z}}},
\end{align*}
where $\tilde{z}\in{T\tilde{Q}}$ and
$\dot{\tilde{z}},\,\acute{\tilde{z}}\in\tspace{\tilde{z}}{T\tilde{Q}}$.
\end{remark}

\begin{theorem} %
\label{LT:Theo}%
Let $\tilde{\calL}=\calL\circ\tilde{\calT}_\vel$ be the transformed
Lagrangian and $\npair{\tilde{\calH}}{\tilde{\sigma}}$ the
associated Hamiltonian structure on $T\tilde{Q}$. Then,
$\tilde{\calH}$ and $\tilde\sigma$ equal the transformed Hamiltonian
and symplectic form, respectively.
\end{theorem}
\begin{proof}
Let $z=\pair{x}{x_\MiTime}$ be given, and
$\tilde{z}=\calT_\vel{z}=\pair{\tilde{x}}{\tilde{x}_\MiTime}$. The
definition of $\tilde{\calL}$ implies
\begin{align}
\label{LT:Theo.Eqn1} \tilde{\pi}\at{\tilde{z}}=
\VarDer{\tilde\calL}{\tilde{x}_\MiTime}\at{z}=
\nnat{\tilde{\calT}_\con}^\prime\pi\nnat{\tilde{\calT}_\vel\tilde{z}},
\quad \text{i.e.}\quad
\skp{\tilde\pi\nnat{\tilde{z}}}{\cdot}_{\tilde{Q}}=
\skp{\pi\nnat{\tilde{\calT}_\vel\tilde{z}}}{\tilde{\calT}_\con\cdot}_{{Q}},
\end{align}
where $\nnat{\tilde{\calT}_\con}^\prime$ is the adjoint operator to
$\tilde\calT_\con$. From this identity we derive
\begin{align*}
\tilde{\calH}\at{\tilde{z}}=
\skp{\tilde{\pi}\at{\tilde{z}}}{\tilde{x}_\MiTime}_{\tilde{Q}}-
\tilde{\calL}\nnat{\tilde{z}}=
\skp{\pi\nnat{\tilde{\calT}_\vel\tilde{z}}}{\tilde{\calT}_\con\tilde{x}_\MiTime}_{Q}-
\calL\nnat{\tilde{\calT}_\vel\tilde{z}}=
\calH\nnat{\tilde{\calT}_\vel\tilde{z}},
\end{align*}
as well as
\begin{align}
\label{LT:Theo.Eqn2}
\skp{D\tilde\pi|_{\tilde{z}}\at{\dot{\tilde{z}}}}{\cdot}_{\tilde{Q}}=
\skp{D\pi|_{\tilde{\calT}_\vel\tilde{z}}
\nnat{\tilde{\calT}_\vel\dot{\tilde{z}}}}{\tilde{\calT}_\con\cdot}_{{Q}}
\end{align}
for all $\dot{\tilde{z}}\in\tspace{\tilde{z}}{T\tilde{Q}}$. Finally,
combining \eqref{LT:Theo.Eqn2} with \eqref{GS:CanSymplForm} for
$z=\tilde{\calT}_\vel\tilde{z}$ we find
\begin{align*}
\si|_{\tilde{\calT}_\vel\tilde{z}}
\npair{\tilde{\calT}_\vel\dot{\tilde{z}}}{\tilde{\calT}_\vel\acute{\tilde{z}}}&=
\skp{D\pi|_{\tilde{\calT}_\vel\tilde{z}}
\nnat{\tilde{\calT}_\vel\acute{\tilde{z}}}}{\tilde{\calT}_\con\dot{\tilde{x}}}_{Q}-
\skp{D\pi|_{\tilde{\calT}_\vel\tilde{z}}
\nnat{\tilde{\calT}_\vel\dot{\tilde{z}}}}{\tilde{\calT}_\con\acute{\tilde{x}}}_{Q}
\\&=%
\skp{D\tilde{\pi}|_{\tilde{z}}\nnat{\acute{\tilde{z}}}}{\dot{\tilde{x}}}_{\tilde{Q}}-
\skp{D\tilde{\pi}|_{\tilde{z}}\nnat{\dot{\tilde{z}}}}{\acute{\tilde{x}}}_{\tilde{Q}}
=\tilde{\si}|_{\tilde{z}}\npair{\dot{\tilde{z}}}{\acute{\tilde{z}}},
\end{align*}
and the proof is complete.
\end{proof}
%
%
%
%
\begin{corollary} %
The following equivalences are satisfied:
\begin{enumerate}
\item
A curve $\MiTime\mapsto{x}\at\MiTime\in{Q}$ satisfies the Lagrangian
equation to $\calL$ if and only if the transformed curve
$\MiTime\mapsto\tilde{x}\at\MiTime=\calT_\con{x}\at\MiTime\in\tilde{Q}$
satisfies the Lagrangian equation to $\tilde\calL$.
\item
A curve $\MiTime\mapsto{z}\at\MiTime\in{TQ}$ satisfies the
Hamiltonian equation to $\pair{\calH}{\si}$ if and only if the
transformed curve
$\MiTime\mapsto\tilde{z}\at\MiTime=\calT_\vel{z}\at\MiTime\in{T\tilde{Q}}$
satisfies the Hamiltonian equation to
$\npair{\tilde\calH}{\tilde\si}$.
\end{enumerate}
\end{corollary}
%
%
%
%
%
%
%
\subsubsection{Weak symmetry transformations}
%
%
We introduce the notion of a \emph{weak symmetry transformation}
which describes a certain class of linear and invertible operators
from $Q$ into $Q$. Although both the Lagrangian and Hamiltonian
structures are not invariant they behave nicely under such
transformations. In particular, each weak symmetry transformation
changes neither the fiber derivative of $\calL$ nor the symplectic
form $\sigma$.
\begin{definition}
\label{SymmTrafo:Definition1}%
A \emph{weak symmetry transformation} (with respect to the
Lagrangian $\calL$) is a linear isomorphism $\calT_\con:Q\to{Q}$
with the following properties:
\begin{enumerate}

\item
$\calT_\con$ is unitary, this means
\begin{math}
\skp{\calT_\con{x}}{\tilde{x}}=
\skp{{x}}{\tilde{\calT}_\con\tilde{x}}
\end{math}
for all $x,\,\tilde{x}\in  Q$ .
\item
The canonical momentum $\pi=\VarDer{\calL}{x_\MiTime}$ commutes with
 $\calT_\con$ in the sense that
\begin{align}
\label{SymmTrafo:Definition1.Eqn2}
\pi\at{\calT_\vel{z}}=\calT_\con\pi\at{z}
\end{align}
holds for all $z=\pair{x}{x_\MiTime}\in{TQ}$ .
\end{enumerate}
Moreover, $\calT_\con$ is called  a \emph{symmetry transformation}
if it respects the Lagrangian, i.e.,
\begin{math}
\calL=\tilde\calL
\end{math}
in the sense of Theorem \ref{LT:Theo}.
\end{definition}
\begin{remark} \quad
\label{SymmTrafo.Remark1}%
$(i)$ Unitarity implies $\tilde{\calT}_\con={\calT}_\con^{\,\prime}$
and $\skp{x_1}{x_2}=\skp{{\calT_\con}x_1}{{\calT_\con}x_2}$ for all
$x_1,x_2\in{Q}$. $(ii)$ Condition \eqref{SymmTrafo:Definition1.Eqn2}
is equivalent to $\pi=\tilde{\pi}$, see \eqref{LT:Theo.Eqn1}, and
this implies
\begin{math}\mathfrak{F}{\calL}=\mathfrak{F}{\tilde{\calL}}\end{math} and
$\sigma=\tilde{\sigma}$.
$(iii)$ Each symmetry transformation satisfies $\calH=\tilde{\calH}$
and $\si=\tilde{\si}$.
$(iv)$ $\calL=\tilde{\calL}$ is sufficient for
\eqref{SymmTrafo:Definition1.Eqn2}.
\end{remark}
\begin{example}%
\label{SymmTrafo.Example1}%
Let $Q=L^2\at{\Rset{\times}T^1;\,\d\MiLagrC\d\Phase}$ be the
Lebesgue space of functions $x$ depending on $\MiLagrC\in\Rset$ and
a periodic phase variable $\Phase\in{T^1}\cong\ccinterval{0}{2\pi}$,
and let the unitary operator $\calT_\con$ be defined by
\begin{math}
\nnat{\calT_\con\,x}\pair{\MiLagrC}{\Phase}=
x\pair{\MiLagrC}{\Phase+s_0\eta}
\end{math}
for some $s_0$. The Lagrangian $\calL$ of the embedded Klein--Gordon
equation, cf.\ \eqref{KGpde.Energies}, is \emph{not} invariant under
the action of $\calT_\con$ as the differential operator
$\MiLagrCDer{}$ transforms into $\MiLagrCDer{}+s_0\PhaseDer{}$.
However, the condition \eqref{SymmTrafo:Definition1.Eqn2} is
satisfied.
\end{example}
%
%
%
\subsubsection{Groups of symmetry transformations}
%
%
The concept of symmetry groups is well established in mechanics and
mathematics and plays a fundamental role in the analysis of
Hamiltonian systems. Here we summarize the definitions and basic
properties.
\begin{definition}
\label{SymmTrafo:Definition2}%
A \emph{(weak) symmetry group} (with respect to the Lagrangian
$\calL$) is a one-parameter family of
$s\mapsto\calT_\con\at{s},\,s\in\Rset$, of (weak) symmetry
transformations that satisfies the following properties:
\begin{enumerate}
\item The family is a group of unitary transformations, i.e.,
${\calT}_\con\at{0}=\mathrm{Id}_{Q\to{Q}}$ and
\begin{align*}
{\calT}_\con\at{s{+}\tilde s} =  {\calT}_\con\at{s}
{\calT}_\con\at{\tilde s},\qquad\tilde{\calT}_\con\at{s}=
\nnat{\calT_\con\at{s}}^{-1}=\nnat{\calT_\con\at{s}}^{\prime}={\calT_\con\at{-s}}
\end{align*}
for all $s,\tilde s \in \Rset$.
\item
The generator $\calA_\con$ with
$\calA_\con{x}=\lim_{s\rra0}s^{-1}\at{\calT_\con\at{s}x-x}$ is
defined on a dense subset of $Q$. Consequently, the group
$s\mapsto\calT_\vel\at{s}$ is generated by
$\calA_\vel=\calA_\con\times\calA_\con$.

\end{enumerate}
\end{definition}
\begin{remark}
If $\calL$ is invariant under the action of a symmetry group
Noether's Theorem provides the integral of motion
\begin{align}
\label{SymmTrafo:IntOfMotion} %
\calI\pair{x}{x_\MiTime}=
\skp{\pi\pair{x}{x_\MiTime}}{\calA_\con{x}},\quad
\pi\pair{x}{x_\MiTime}=
\VarDer{\calL}{x_\MiTime}\pair{x}{x_\MiTime},
\end{align}
i.e., $\calI$ is conserved for any solution to the Hamiltonian
equation \eqref{GS:HamEqns}.

\end{remark}
\begin{example}
\label{SymmTrafo.Example2}
Let $Q$ and $\calL$ be as in Example \ref{SymmTrafo.Example1}, and
for fixed $\MiLagrC_0\in\Rset$, $\Phase_0\in\Rset$ and all
$s\in\Rset$ let
\begin{math}
\nnat{\calT_\con\at{s}\,x}\pair{\MiLagrC}{\Phase}=
x\pair{\MiLagrC+s\MiLagrC_0}{\Phase+s\Phase_0}.
\end{math}
Then, $s\mapsto\calT_\con\at{s}$ is a symmetry group with generator
$\calA_\con=\MiLagrC_0\MiLagrCDer{}+\Phase_0\PhaseDer{}$ and
integral of motion
\begin{math}
\calI= \eta_0\calI_{\mathrm{space}} +\Phase_0\calI_{\mathrm{phase}},
\end{math} %
where $\calI_{\mathrm{space}}$ and $\calI_{\mathrm{phase}}$ are
given by \eqref{KGpde.IntOfMotion}.
\end{example}
\begin{lemma}%
\label{SymmTrafo:Lemma1} %
Each (weak) symmetry group satisfies
\begin{math}
\sigma|_z\pair{\calA_\vel{z}}{\cdot}= \d{\calI}|_{z}\at{\cdot}
\end{math} %
for all $z\in{TQ}$ and $\calI$ from \eqref{SymmTrafo:IntOfMotion}.
\end{lemma}
\begin{proof}
For given $z=\pair{x}{x_{\MiTime}}\in{TQ}$ consider the curve
$t\mapsto{z}\at{\MiTime}=\calT_\vel\at{\MiTime}z\in{TQ}$, and its
image under $\mathfrak{F}\calL$, that is
$t\mapsto\mhol{z}\at\MiTime=\pair{{x}\at\MiTime}{{\pi}\at\MiTime}$
with ${x}\at\MiTime=\calT_\con\at\MiTime{x}$ and
${\pi}\at\MiTime=\pi\at{z\at{t}}$. Moreover, let
$\acute{\mhol{z}}=\pair{\acute{x}}{\acute{\pi}}$ be an arbitrary
tangent vector in $\tspace{\mhol{z}\at{0}}{T^\ast{Q}}$. Condition
\eqref{SymmTrafo:Definition1.Eqn2} implies
\begin{math}
\skp{\pi\at{\MiTime}}{\acute{x}}=
\skp{\pi\at{\calT_\vel\at{\MiTime}z\at{0}}}{\acute{x}}=
\skp{\pi\at{0}}{{\calT}_\con\at{-t}\,\acute{x}}
\end{math} %
and differentiation and evaluation for $\MiTime=0$ yield
\begin{math}
\skp{{\pi}_\MiTime\at0}{\acute{x}}=-\skp{\pi\at{0}}{{\calA_\con}\acute{x}}.
\end{math} %
This identity and the definition of $\mhol{\sigma}_\can$ provide
\begin{align*}
\mhol{\si}_\can|_{\mhol{z}\at{0}}%
\npair{\mhol{z}_\MiTime\at{0}}{\acute{\mhol{z}}}&=
\mhol{\si}_\can|_{\mhol{z}_0}\Bpair%
{\pair{\calA_\con{x}\at{0}}{{\pi}_\MiTime\at0}}%
{\pair{\acute{x}}{\acute{\pi}}}%
\\&= %
\skp{\acute{\pi}}{\calA_\con{x}\at{0}}-
\skp{{\pi}_\MiTime\at0}{\calA_\con\acute{x}}
= %
\skp{\acute{\pi}}{\calA_\con{x}\at{0}}+
\skp{{\pi}\at0}{\calA_\con\acute{x}}.
\end{align*}
Moreover, for $\mhol{\calI}\pair{x}{\pi}=\skp{\pi}{\calA_\con{x}}$
we find
\begin{math}
\d\mhol{\calI}|_{\mhol{z}\at0}\at{\acute{\mhol{z}}}=
\skp{\acute{\pi}}{{\calA_\con}\,{x}\at{0}}+
\skp{\pi\at{0}}{{\calA_\con}\,\acute{x}}
\end{math}
and, hence,
\begin{math}
\mhol{\si}_\can|_{\mhol{z}\at{0}}\pair{\mhol{z}_\MiTime\at{0}}{\cdot}=
\d\mhol{\calI}|_{\mhol{z}\at0}\at{\cdot}.
\end{math}
Finally, pulling back this identity via $\mathfrak{F}\calL$ and
using $z_\MiTime\at{0}=\mathcal{A}_\vel{z}\at{0}$ completes the
proof .
\end{proof}
%
%
%
\subsubsection{Moving frames}\label{sec:MFTrafo}
%
%
In this section we consider a time-parametrized family of invertible
transformations of the configuration space
$\calM_\con\at{\MiTime}:Q\rra{Q}$ and denote the family of inverse
transformations by
\begin{math}
\tilde{\calM}_\con\at{\MiTime}.
\end{math}
Taking into account the time dependence we shall lift this
transformation to the tangent bundle as follows: Each
time-parametrized curve $\MiTime\mapsto{x}\at\MiTime$ in $Q$
provides a lifted curve
\begin{math}
\MiTime\mapsto\pair{x\at\MiTime}{x_t\at\MiTime}
\end{math}
in $TQ$, where $x_\MiTime\at\MiTime$ denotes the tangent vector at
time $\MiTime$, i.e.\
\begin{math}
x_t\at\MiTime=\tfrac{\d}{\d\MiTime}x\at\MiTime.
\end{math}
Consequently, the lift of the transformed curve
\begin{math}
\MiTime\mapsto\tilde{x}\at\MiTime=
\calM_\con\at{\MiTime}\,x\at\MiTime
\end{math} %
is given by
\begin{align*}
\MiTime\mapsto\pair{\tilde{x}\at\MiTime} {\tilde{x}_t\at\MiTime}=
\pair{\tilde{x}\at\MiTime}{\tfrac{\d}{\d\MiTime}
\tilde{x}_t\at\MiTime}= \pair{\calM_\con\at{\MiTime}\,
x\at\MiTime}{\calM_\con\at{\MiTime}\,x_t\at\MiTime+
\at{\tfrac{\d}{\d\MiTime} \calM_\con\at{\MiTime}}\,x\at\MiTime},
\end{align*}
and we read-off the definition of $\calM_\vel\at{\MiTime}$, that is
\begin{align}
\label{MFTrafo:LiftedMap.Eqn1}%
\calM_\vel\at{\MiTime}:\pair{x}{x_\MiTime}
\mapsto\pair{\tilde{x}}{\tilde{x}_\MiTime}=
\pair{\calM_\con\at{\MiTime}\,x}{\calM_\con\at{\MiTime}\,x_t+
\at{\tfrac{\d}{\d\MiTime}\calM_\con\at{\MiTime}}\,x}.
\end{align}
The transformation of a Hamiltonian structure under a time-dependent
transformation is in general quite complicated. Therefore we solely
discuss time-dependent transformations that are related to
\emph{moving frames}.
\begin{definition}
\label{SymmTrafo:Definition3}%
The transformation $\calM_\con\at\MiTime$  is called a
\emph{moving-frame transformation} (with respect to the Lagrangian
$\calL$) if it is related to a symmetry group
$s\mapsto\calT_\con\at{s}$ via
$\calM_\con\at\MiTime=\calT_\con\at{t}$. This implies
$\calM_\con\at{0}=\mathrm{Id}_{Q\to{Q}}$ and
$\tilde{\calM}_\con\at{t}=\calM_\con\at{-t}$ for all $t$.
\end{definition}
\begin{example}
Let $Q$ be as in Example \ref{SymmTrafo.Example1}, and let
$\calL_\mtrc$ be the metric Lagrangian from Example
\ref{GS:Ex.Lagrangian}. Obviously, $\calL_\mtrc$ is invariant under
Galilean transformations
\begin{math}
\triple{\MiTime}{\MiLagrC}{\Phase}\mapsto
\triple{\MiTime}{\tilde{\MiLagrC}}{\Phase},
\end{math} %
where $\tilde{\MiLagrC}=\MiLagrC-c\,\MiTime$ denotes the spatial
coordinate in the moving frame. The corresponding time-dependent
transformations $\calM_\con\at\MiTime:x\mapsto\tilde{x}$ and
\begin{math}
\calM_\vel\at\MiTime:\pair{x}{x_\MiTime}
\mapsto\pair{\tilde{x}}{\tilde{x}_\MiTime}
\end{math}
can be read-off from the identification
$x\triple{\MiTime}{\MiLagrC}{\Phase}=\tilde{x}\triple{\MiTime}{\MiLagrC-c\MiTime}{\Phase}$
and are given by
\begin{align*}
\tilde{x}\pair{\MiLagrC}{\Phase}=
x\pair{\MiLagrC+c\MiTime}{\Phase},\quad
{\tilde{x}_\MiTime}\pair{\MiLagrC}{\Phase}=
x_\MiTime\pair{\MiLagrC+c\MiTime}{\Phase}+
c\,x_\MiLagrC\pair{\MiLagrC+c\MiTime}{\Phase},
\end{align*}
where $x_\MiLagrC$ abbreviates the derivative of $x$ with respect to
$\MiLagrC$. The underlying symmetry group
$\nnat{\calT_\con\at{s}{x}}\pair{\MiLagrC}{\Phase}=x\pair{\MiLagrC+cs}{\Phase}$
has the generator $\calA_\con=c\MiLagrCDer{}$ and the conserved
quantity
\begin{math}
\calI\pair{x}{x_\MiTime}= c\int_{\Rset}x_\MiTime\,
{x}_\MiLagrC\,\d\MiLagrC\in\Rset.
\end{math} %
The lifted transformations $\calM_\vel\at\MiTime$ and
$\calT_\vel\at{t}$ are really different because of
\begin{math}
\calT_\vel\at{t}= \calT_\con\at{t}\times\calT_\con\at{t}.
\end{math} %
\end{example}
For moving-frame transformations we can decompose the lifted map as
follows: Definition \ref{SymmTrafo:Definition3} implies
\begin{math}
\tfrac{\d}{\d\MiTime}\calM_\con\at\MiTime=
\calA_\con\calT_\con\at\MiTime,
\end{math} %
and using \eqref{MFTrafo:LiftedMap.Eqn1} we conclude that
\begin{align}
\notag
\begin{split}
\calM_\vel\at\MiTime&=
\calR_\vel\circ\calT_\vel\at\MiTime\quad\text{with}\quad
\calR_\vel:\pair{x}{x_\MiTime}
\mapsto\pair{x}{x_\MiTime+\calA_\con\,x},\\
\tilde{\calM}_\vel\at\MiTime&=
\tilde{\calT}_\vel\at\MiTime\circ\tilde{\calR}_\vel\quad\text{with}\quad
\tilde{\calR}_\vel:\pair{\tilde{x}}{\tilde{x}_\MiTime}
\mapsto\pair{\tilde{x}}{\tilde{x}_\MiTime-\calA_\con\,\tilde{x}}.
\end{split}
\end{align}
In what follows we denote by $\tilde{\calL}$ the transformed
Lagrangian, i.e.\
$\tilde{\calL}\at\MiTime=\calL\circ\tilde{\calM}_\vel\at\MiTime$,
and with $\npair{\tilde{\calH}}{\tilde{\sigma}}$ the Hamiltonian
structure corresponding to $\tilde{\calL}$. However, since the
Legendre transformation does not commute with $\calM_\vel\at\MiTime$
we can not expect  $\tilde{\calH}$, that is the Legendre transform
of $\tilde\calL$, to equal the transformed Hamiltonian. For this
reason we identify $\calH$ with $\calE$, and define
$\tilde{\calE}=\calE\circ\tilde{\calM}_\vel{\at\MiTime}$. This
notation is motivated by normal systems, see
\eqref{TypEx:Normal.Lagrangian}, for which the Hamiltonian $\calH$
equals the total energy $\calE=\calK+\calV$. Finally, we write
$\tilde{\calI}=\calI\circ\tilde{\calM}_\vel\at\MiTime$, where
$\calI\at{z}=\skp{\pi\at{z}}{\calA_\con{x}}$ is the integral of
motion associated to the symmetry group.
\par
Next we prove that all these quantities do not depend on time, as it
is already indicated by the notation, and derive the transformation
rules for the Hamiltonian structure.
\begin{theorem}
\label{MF:Lemma2}
Moving-frame transformations satisfy
\begin{align}
\notag
\tilde{\calL}=\calL\circ\tilde{\calR}_\vel,\quad
\tilde{\calE}=\calE\circ\tilde{\calR}_\vel,\quad
\tilde{\calI}=\calI\circ\tilde{\calR}_\vel.
\end{align}
Moreover, we have
\begin{align*}
\tilde{\calH}=\tilde{\calE}+\tilde{\calI},\quad\tilde\sigma=\nnat{\tilde{\calR}_\vel}^\ast\si,
\end{align*}
where $\npair{\tilde{H}}{\tilde\si}$ is the Hamiltonian structure
associated to $\tilde{\calL}$.
\end{theorem}
\begin{proof}
Let $\MiTime$ be fixed, and for arbitrary $z=\pair{x}{x_\MiTime}$
let
$\tilde{z}=\pair{\tilde{x}}{\tilde{x}_\MiTime}=\calM_\vel\at\MiTime{z}$.
Due to the invariance of $\calL$ under $\calT_\vel\at{t}$ we have
\begin{math}
\tilde{\calL}=
\at{\calL\circ\calT_\vel\at{-\MiTime}}\circ\tilde{\calR}_\vel=
\calL\circ\tilde{\calR}_\vel,
\end{math}
and this implies
\begin{align*}
\tilde{\pi}\at{\tilde{z}}=\VarDer{\tilde\calL}{\tilde{x_\MiTime}}\at{\tilde{z}}=
\VarDer{}{\tilde{x_\MiTime}}\at{\calL\pair{\tilde{x}}{\tilde{x}_\MiTime-\calA_\con\tilde{x}_\MiTime}}=
\pi\pair{\tilde{x}}{\tilde{x}_\MiTime-\calA_\con\tilde{x}}
\end{align*}
so that $\tilde\pi=\pi\circ\tilde{\calR}_\vel$.
We conclude that %
\begin{math}
\mathfrak{F}{\tilde{\calL}}=
\mathfrak{F}{\calL}\circ\tilde{\calR}_\vel=
\nnat{\tilde{\calR}_\vel}^\ast\mathfrak{F}{\calL}
\end{math} %
and hence $\tilde{\sigma}=\nnat{\tilde{\calR}_\vel}^\ast{\sigma}$.
The unitarity of $\tilde{\calT}_\con\at\MiTime$, the identity
$\calA_\con\tilde{\calT}_\con\at\MiTime=\tilde{\calT}_\con\at\MiTime\calA_\con$
and Formula \eqref{SymmTrafo:Definition1.Eqn2} yield
\begin{align*}
\tilde{\calI}\at{\tilde{z}}&=\calI\nnat{\tilde{\calM}_\vel\at\MiTime\tilde{z}}=
\skp{\pi\nnat{\tilde{\calM}_\vel\at\MiTime\tilde{z}}}{\calA_\con\tilde{\calT}_\con\at\MiTime\tilde{x}}=
\skp{\pi\nnat{\tilde{\calM}_\vel\at\MiTime\tilde{z}}}{\tilde{\calT}_\con\at\MiTime\calA_\con\tilde{x}}
\\&=%
\skp{\pi\nnat{\calT_\vel\at\MiTime\tilde{\calM}_\vel\at\MiTime\tilde{z}}}{\calA_\con\tilde{x}}
=%
\skp{\pi\nnat{\tilde{\calR}_\vel\tilde{z}}}{\calA_\con\tilde{x}}=
\calI\nnat{\tilde{\calR}_\vel\tilde{z}},
\end{align*}
the desired result for $\calI$. Analogously, with
$\at{\calE+\calL}\at{z}=\skp{\pi\at{z}}{x}$ we find
\begin{align*}
\nnat{\tilde{\calE}+\tilde{\calL}}\at{\tilde{z}}&=
\nnat{\calE+\calL}\nnat{\tilde{\calM}_\vel\at\MiTime\tilde{z}}=
\skp{\pi\nnat{\tilde{\calT}_\vel\at\MiTime\tilde{\calR}_\vel\tilde{z}}}%
{\tilde{\calT}_\con\at\MiTime\nnat{\tilde{x}_\MiTime-\calA_\con\tilde{x}}}
\\&=%
\skp{\pi\nnat{\tilde{\calR}_\vel\tilde{z}}}
{{\tilde{x}_\MiTime-\calA_\con\tilde{x}}}=
\nnat{\calE+\calL}\nnat{\tilde{\calR}_\vel\tilde{z}},
\end{align*}
which implies the formula for $\calE$. Finally,
\begin{align*}
\tilde{\calH}\at{\tilde{z}}&=\skp{\tilde{\pi}\nnat{\tilde{z}}}{\tilde{x}_\MiTime}-\tilde{\calL}\nnat{\tilde{z}}=
\skp{\pi\nnat{\tilde{\calR}_\vel\tilde{z}}}{\tilde{x}_\MiTime}-\tilde{\calL}\nnat{\tilde{\calR}_\vel\tilde{z}}
\\&=
\skp{\pi\nnat{\tilde{\calR}_\vel\tilde{z}}}{\tilde{\calR}_\vel\tilde{x}_\MiTime}+
\skp{\pi\nnat{\tilde{\calR}_\vel\tilde{z}}}{\calA_\con\tilde{x}}-\tilde{\calL}\nnat{\tilde{\calR}_\vel\tilde{z}}
=\calE\nnat{\tilde{\calR}_\vel\tilde{z}}+\calI\nnat{\tilde{\calR}_\vel\tilde{z}},
\end{align*}
and the proof is finished.
\end{proof}
\begin{figure}[ht!]
\centering{
\setlength{\unitlength}{1865sp}%
\begin{picture}(7000,3000)(3000,-5600)
\thicklines
{\put(11600,-3000){\vector( -1, 0){9300}}}%
{\put(6300,-5500){\vector( -1, 0){4000}}}%
{\put(1800,-3300){\vector( 0,-1){1600}}}%
{\put(7000,-2800){\makebox(0,0)[cb]{$\tilde{\calM}_\vel\at{\MiTime}$}}}%
{\put(4300,-5300){\makebox(0,0)[cb]{$\tilde{\calM}_\vel\at{\MiTime}$}}}%
\put(12100,-2900){\makebox(0,0)[cc]{$\tilde{\calL}$}}
\put(2000,-2900){\makebox(0,0)[rc]{$\calL$}}%
\put(2000,-5400){\makebox(0,0)[rc]{$\calH=\calE$}}%
\put(7000,-5400){\makebox(0,0)[cc]{$\tilde{\calE}$}}
{\put(12100,-3300){\vector( 0,-1){1600}}}%
{\put(7600,-5500){\vector( 1, 0){4000}}}%
\put(9600,-5300){\makebox(0,0)[cb]{$\tilde{\calH}=\tilde{\calE}+\tilde{\calI}$}}
\put(11900,-4100){\makebox(0,0)[rc]{$\tilde{\calH}=
\skp{\VarDer{\tilde{\calL}}{\tilde{x}_\MiTime}}{\tilde{x}_\MiTime}-\tilde{\calL}$}}
\put(12100,-5400){\makebox(0,0)[cc]{$\tilde{\calH}$}}
\put(2000,-4100){\makebox(0,0)[lc]{$\calH=\skp{\VarDer{\calL}{x_\MiTime}}{x_\MiTime}-\calL$}}
\end{picture}
}%
\caption{ $\calL$, $\calE$, $\calH$ and their transformed
counterparts in a moving frame. }\label{MF:Fig1}
\end{figure}
The results of Theorem \ref{MF:Lemma2} can be reinterpreted as the
\emph{transformation rule} for Hamiltonian structures, see Figure
\ref{MF:Fig1}. In fact, $\tilde{\si}$ equals the pull-back of $\si$,
and to obtain $\tilde{\calH}$ we pull back the sum of the
Hamiltonian $\calE=\calH$ and conserved quantity $\calI$. As a
consequence we gain the following result.
\begin{corollary}%
The following equivalences are satisfied.
\begin{enumerate}
\item%
A curve $\MiTime\mapsto{x}\at\MiTime\in{Q}$ solves the Lagrangian
equation to $\calL$ if and only if the transformed curve
$\MiTime\mapsto\tilde{x}\at\MiTime=\calM_\con\at\MiTime{x}\at\MiTime\in{Q}$
solves the Lagrangian equation to $\tilde{\calL}$.
\item%
A curve $\MiTime\mapsto{z}\at\MiTime\in{TQ}$ solves the Hamiltonian
equation to $\pair{\calH}{\si}$ if and only if the transformed curve
$\MiTime\mapsto\tilde{z}\at\MiTime=\calM_\vel\at\MiTime{z}\at\MiTime\in{TQ}$
solves the Hamiltonian equation to
$\npair{\tilde{\calH}}{\tilde{\si}}$.
\end{enumerate}
\end{corollary}
\begin{proof}
Since $\npair{\tilde{\calH}}{\tilde{\si}}$ is the Hamiltonian
structure associated to $\tilde{\calL}$ it is sufficient to prove
the equivalence in the Hamiltonian framework. Let
$\hat{z}\at\MiTime=\tilde{\calR}_\vel\tilde{z}\at\MiTime$ such that
$z\at\MiTime=\tilde{\calM}_\vel\at\MiTime\tilde{z}\at\MiTime$
satisfies $\hat{z}\at\MiTime=\calT_\vel\at\MiTime{z}\at\MiTime$. Now
suppose that $\MiTime\mapsto\tilde{z}\at\MiTime$ solves the
Hamiltonian equation to $\npair{\tilde{\calH}}{\tilde\si}$. This
means

\begin{math}
\tilde{\si}|_{\hat{z}\at{\MiTime_0}}\pair{\tilde{z}_\MiTime\at\MiTime_0}{\cdot}=%
\d\tilde{\calH}|_{\tilde{z}\at{\MiTime_0}}\at{\cdot}%
\end{math}
for arbitrary but fixed $\MiTime_0$, and Theorem \ref{MF:Lemma2}
provides
\begin{align*}
\si|_{\hat{z}\at{\MiTime_0}}\pair{\hat{z}_\MiTime\at{\MiTime_0}}{\cdot}=%
\d\calE|_{\hat{z}\at{\MiTime_0}}\at{\cdot}+ %
\d\calI|_{\hat{z}\at{\MiTime_0}}\at{\cdot}.%
\end{align*}
By construction we have
$\hat{z}_\MiTime\at{\MiTime_0}=\calT_\vel\at{\MiTime_0}{z_\MiTime\at{\MiTime_0}}
+\calA_\vel{z}\at{\MiTime_0}$, and exploiting Lemma
\ref{SymmTrafo:Lemma1} we find
\begin{align*}
\si|_{\calT_\vel\at{\MiTime_0}{z}\at{\MiTime_0}}\pair{\calT_\vel\at{\MiTime_0}{z_\MiTime\at{\MiTime_0}}}{\cdot}=%
\d\calE|_{\calT_\vel\at{\MiTime_0}{z}\at{\MiTime_0}}\at{\cdot}=\d\calH|_{\calT_\vel\at{\MiTime_0}{z}\at{\MiTime_0}}\at{\cdot},%
\end{align*}
and the invariance of $\calL$, $\calH$ and $\sigma$ under
$\calT_\vel\at{\MiTime_0}$ (cf.\ Remark \ref{SymmTrafo.Remark1})
shows that $\MiTime\mapsto{z}\at\MiTime$ solves the Hamiltonian
equation to $\npair{\calH}{\si}$. Finally, in order to establish the
equivalence we argue in the reverse direction.
\end{proof}
%
%
%
%
%
%
\subsubsection{Scaling transformations}\label{sec:ScalTrafo}
%
The two-scale problems considered in \S\ref{sec:ChainNew} involve
suitable scalings of space and time variables. We always suppose
that there exist positive constants $\be$ and $\ga$ such that
$\MaTime=\eps^\be\MiTime$ and $\MaLagr=\eps^\ga\MiLagrC$. In
particular, $\tfrac{\d\MaTime}{\d\MiTime}=\eps^\beta$ and
$\tfrac{\d\MaLagr}{\d\MiLagrC}=\eps^\gamma$ are the scaling
constants for time and space, respectively.
\bigpar%
The spatial scaling can be encoded in a linear and invertible
\emph{scaling transformation} $\calS_\con:Q\rra{P}$ that maps the
microscopic configuration space ${Q}$ into $P$, the space of all
macroscopic configurations. In what follows $P$ is always a Hilbert
space, usually some $L^2$--space, with inner product
$\sskp{\cdot}{\cdot}$ and $\bbS_\con:P\rra{}Q$ is the inverse to
$\calS_\con$. The elements of $P$ are denoted by $X$ and are
functions of the macroscopic space variable $\MaLagr$.
\begin{definition}
A \emph{scaling transformation} is a scaled isometry
$\calS_\con:Q\rra{P}$, i.e.,
\begin{align*}
\skp{x}{\tilde{x}}=
\eps^\mu\sskp{\calS_\con\,x}{\calS_\con\,\tilde{x}}
\end{align*}
for some exponent $\mu$ and all $x,\,\tilde{x}\in{Q}$.
\end{definition}
Notice that $P$ does not depend on the scaling parameter $\eps$,
whereas the transformations $\calS_\con$ and $\bbS_\con$ as well as
the scaled Lagrangian and Hamiltonian structures will strongly
depend on $\eps$. Nevertheless, for the moment $\eps$ is an
arbitrary but fixed parameter and hence we do not denote explicitly
the dependence on $\eps$.
\begin{example}
\label{ST:Example1}
Let $Q=L^2\at{\Rset;\,\d\MiLagrC}$ and
$P=L^2\at{\Rset;\,\d\MaLagr}$, and consider the two-scale ansatz
$x\pair{\MiTime}{\MiLagrC}=\eps^\alpha{X}\pair{\eps^\beta\MiTime}{\eps^\ga\MiLagrC}$.
In this case we have
\begin{align*}
\bat{\calS_\con\,x}\at{\MaLagr}=\eps^{-\al}{x}\at{\eps^{-\ga}\MaLagr},\quad
\bat{\bbS_\con\,X}\at{\MiLagrC}=\eps^\al{X}\at{\eps^\ga\MiLagrC},
\end{align*}
providing
\begin{align*}
\skp{x}{\tilde{x}}=\int\mhintlimits_\Rset{x}\at{\MiLagrC}\tilde{x}\at{\MiLagrC}\,\d\MiLagrC=
\eps^{2\alpha}\int\mhintlimits_\Rset{X}\at{\eps^\gamma\MiLagrC}\tilde{X}\at{\eps^\gamma\MiLagrC}\,\d\MiLagrC=
\eps^{2\alpha-\gamma}\sskp{X}{\tilde{X}}.
\end{align*}
\end{example}
The transformation $\calS_\con:Q\rra{}P$ does not take into account
the time scaling, since this is related to \emph{reparametrization
of curves} as follows: Let $\MiTime\mapsto{x}\at\MiTime$ be any
curve in $Q$ with tangent vectors
$\MiTime\mapsto{x}_\MiTime\at\MiTime$, and let
$\MiTime\mapsto\tilde{X}\at\MiTime=\calS_\con\,x\at\MiTime$ be the
transformed curve in $P$, which has tangent vectors
\begin{math}%
\MiTime\mapsto\tilde{X}_\MiTime\at\MiTime=
\tfrac{\d}{\d\MiTime}\tilde{X}\at\MiTime=
\calS_\con\,x_\MiTime\at\MiTime.
\end{math} %
In view of the time scaling we are not interested in
$\tilde{X}\at\MiTime$, but refer to the reparametrized curve
\begin{align*}
\MaTime\mapsto{X}\at{\MaTime}= \tilde{X}\at{\MiTime\at{\MaTime}}=
\calS_\con\,x\at{\MiTime\at{\MaTime}}
\end{align*}
with rescaled tangent vectors
\begin{align}
\label{ST:TransTangential} %
\MaTime\mapsto{}X_\MaTime\at\MaTime=
\tfrac{\d}{\d\MaTime}X\at\MaTime= \tfrac{\d\MiTime}{\d\MaTime}\,
\calS_\con\,x_\MiTime\at{\MiTime\at{\MaTime}}.
\end{align}
For this reason we denote elements of $TP\cong{P}{\times}{P}$ by
$\pair{X}{X_\MaTime}$ instead of $\pair{X}{X_\MiTime}$. Moreover, we
must take into account this reparametrization when defining
$\calS_\vel$, i.e.\ the lift of $\calS_\con$ to a map $TQ\rra{TP}$.
In fact, using \eqref{ST:TransTangential} we find
\begin{align*}
\calS_\vel:\pair{x}{x_\MiTime}\mapsto\pair{X}{X_\MaTime}=
\pair{\calS_\con\,x}{\eps^{-\beta} \calS_\con\,x_\MiTime}.
\end{align*}
\begin{example}
\label{ST:Example2} Using the notations from Example
\ref{ST:Example1} we obtain
\begin{align*}
\bat{\calS_\vel\pair{x}{x_\MiTime}}\at{\MaLagr}=
\pair{\eps^{-\alpha}{x}}{\eps^{-\alpha-\beta}x_\MiTime}\at{\eps^{-\ga}\MaLagr},\quad
\bat{\bbS_\vel\pair{X}{X_\MaTime}}\at\MiLagrC=
\pair{\eps^{\alpha}{X}}{\eps^{\alpha+\beta}X_\MaTime}\at{\eps^\ga\MiLagrC}.
\end{align*}
\end{example}
Following the proof of Theorem \ref{LT:Theo}, we derive the
transformation rules for the Lagrangian and Hamiltonian structures.
To this end, let $\bbL=\calL\circ\bbS_\vel$ be the transformed
Lagrangian and $\pair{\bbH}{\bssi}$ the associated Hamiltonian
structure on $TP$, i.e., $\bbH$ is the Legendre transform of $\bbL$
and $\bssi=\nnat{\mathfrak{F}{\bbL}}^\ast\mhol{\bssi}_\can$, where
$\mhol{\bssi}_\can$ is the canonical symplectic structure on
$T^\ast{P}$.
\begin{theorem}
\label{ST:Theo1} We have $\bbH=\calH\circ\bbS_\vel$ and
$\bssi=\tfrac{\d\MaTime}{\d\MiTime}\hat{\bssi}$, where
$\hat{\bssi}=\nnat{\bbS_\vel}^\ast{}\sigma$.
\end{theorem}
\begin{proof}
The definition of $\bbL$ implies
\begin{math}
\sskp{\Pi\nnat{Z}}{\cdot}=
\skp{\pi\nnat{\bbS_\vel{Z}}}{\tfrac{\d\MaTime}{\d\MiTime}\bbS_\con\cdot},
\end{math}
where $Z=\pair{X}{X_\MaTime}$ and $\Pi=\VarDer{\bbL}{X_\MaTime}$. We
conclude that
\begin{align*}
\bbH\at{Z}=\sskp{\Pi\at{Z}}{X_\MaTime}-\bbL\at{Z}=
\skp{\pi\nnat{\bbS_\vel{Z}}}{\tfrac{\d\MaTime}{\d\MiTime}\bbS_\con{Z}_\MaTime}-\calL\at{\bbS_\vel\at{Z}}=
\calH\nnat{\bbS_\vel{Z}}
\end{align*}
and
\begin{math}
\sskp{D\Pi|_{Z}\nnat{\dot{Z}}}{\cdot}= \skp{D\pi|_{\bbS_\vel{Z}}
\nnat{\bbS_\vel\dot{Z}}}{\tfrac{\d\MiTime}{\d\MaTime}\bbS_\con\cdot}
\end{math}
for all $\dot{Z}\in\tspace{Z}{TP}$. Inserting this identity into the
definition of $\bssi$, compare \eqref{GS:CanSymplForm}, we obtain
\begin{align*}
\bssi|_Z\npair{\dot{Z}}{\acute{Z}}&=
\sskp{D\Pi|_{Z}\nnat{\acute{Z}}}{\dot{X}}-
\sskp{D\Pi|_{Z}\nnat{\dot{Z}}}{\acute{Z}}
\\&=
\skp{D\pi|_{\bbS_\vel{Z}}
\nnat{\bbS_\vel\acute{Z}}}{\tfrac{\d\MaTime}{\d\MiTime}\bbS_\con\dot{X}}
-\skp{D\pi|_{\bbS_\vel{Z}}
\nnat{\bbS_\vel\dot{Z}}}{\tfrac{\d\MaTime}{\d\MiTime}\bbS_\con\acute{X}}
\\&=
\tfrac{\d\MaTime}{\d\MiTime}\skp{D\pi|_{\bbS_\vel{Z}}
\nnat{\bbS_\vel\acute{Z}}}{\bbS_\con\dot{X}}
-\tfrac{\d\MaTime}{\d\MiTime}\skp{D\pi|_{\bbS_\vel{Z}}
\nnat{\bbS_\vel\dot{Z}}}{\bbS_\con\acute{X}}
\\&=
\tfrac{\d\MaTime}{\d\MiTime}\si|_{\bbS_\vel{Z}}\pair{\bbS_\vel\dot{Z}}{\bbS_\vel\acute{Z}}=
\tfrac{\d\MaTime}{\d\MiTime}\hat{\bssi}|_Z\pair{\dot{Z}}{\acute{Z}},
\end{align*}
which is the desired result for $\bssi$.
\end{proof}
The additional scaling parameter in the formula for $\bssi$ appears
naturally due to the reparametrization of curves. More precisely,
the microscopic Hamiltonian equation is equivalent to
\begin{align*}
\hat{\bssi}|_Z\pair{\tfrac{\d}{\d\MiTime}Z}{\cdot}=\d\bbH|_z\at{\cdot},
\end{align*}
but since here the solution still depends on $\MiTime$ we
reparametrize via
$\tfrac{\d}{\d\MiTime}=\tfrac{\d\MaTime}{\d\MiTime}\tfrac{\d}{\d\MaTime}$.
%
\begin{example}
\label{Ex:TimeScaling}
Let $\calL$ be a normal Lagrangian, cf.\
\eqref{TypEx:Normal.Lagrangian}, and consider a simple \emph{time
scaling} $\MiTime\rightsquigarrow\MaTime=\eps\MiTime$ with $P=Q$ and
the two-scale ansatz $x\at{\MiTime}=X\at{\eps\MiTime}$. Then
\begin{align*}
\bbS_\vel=
\begin{pmatrix}
\scriptstyle1&\scriptstyle0\\\scriptstyle0&\scriptstyle{\eps}%
\end{pmatrix},\quad
\bbL\pair{X}{X_\MaTime}=\tfrac{1}{2}\eps^2\skp{X_\MaTime}{M{}X_\MaTime}-\calV\at{X}
\quad \text{with}\quad \Pi\pair{X}{X_\MaTime}=\eps^2{M}X_\MaTime,
\end{align*}
and a simple calculation yields
\begin{align*}
\bbH\pair{X}{X_\MaTime}=\tfrac{1}{2}\eps^2\skp{X_\MaTime}{M{}X_\MaTime}+\calV\at{X}
,\quad
\bssi=\nnat{\mathfrak{F}\bbL}^\ast\mhol{\bssi}_\can\simeq\eps^2{M}\begin{pmatrix}
\scriptstyle0&-\scriptstyle1\\\scriptstyle1&\scriptstyle0%
\end{pmatrix}.
\end{align*}
The pull-back of $\sigma$ via $\bbS_\vel$ is given by
\begin{align*}
\hat\bssi=\at{\bbS_\vel}^\ast\sigma\simeq\bbS_\vel^\trans
\begin{pmatrix}
\scriptstyle0&\scriptstyle-1\\\scriptstyle1&\scriptstyle0%
\end{pmatrix}\bbS_\vel
= \eps
\begin{pmatrix}
\scriptstyle0&\scriptstyle-1\\\scriptstyle1&\scriptstyle0%
\end{pmatrix},
\end{align*}
and differs from $\bssi$ by the factor $\eps^{-1}$.
%
%
%
%
\end{example}
In what follows we refer to $\bbL$ and $\npair{\bbH}{\bssi}$ as the
\emph{macroscopic} Lagrangian and Hamiltonian structures, but we
recall that microscopic and macroscopic structures are completely
equivalent as long as $\eps$ is a fixed but positive parameter.
Consequently, we find the following transformation rules for
solutions.
\begin{corollary}
The following equivalences are satisfied.
\begin{enumerate}
\item
A curve $\MiTime\mapsto{x}\at\MiTime\in{Q}$ solves the microscopic
Lagrangian equation to $\calL$ if and only if the transformed and
reparametrized curve
$\MaTime\mapsto{X}\at\MaTime=\calS_\con\,x\at{\MiTime\at{\MaTime}}\in{P}$
solves the macroscopic Lagrangian equation to $\bbL$.
\item
A curve $\MiTime\mapsto{z}\at\MiTime\in{TQ}$ solves the microscopic
Hamiltonian system to $\pair{\calH}{\si}$ if and only if the
transformed and reparametrized curve
$\MaTime\mapsto{Z}\at\MaTime=\calS_\vel\,z\at{\MiTime\at{\MaTime}}\in{TP}$
solves the macroscopic Hamiltonian equation to $\pair{\bbH}{\bssi}$.
\end{enumerate}
\end{corollary}
%
%
%
\subsubsection{Two-scale transformations}\label{sec:TrafoSumm}
%
Since space-time scalings depend on the parameter $\eps$, from now
on we denote a scaling transformation and its inverse by
$\calS_\con\at\eps$ and $\bbS_\con\at\eps$, respectively.
Consequently, both the macroscopic Lagrangian and Hamiltonian
structures will depend on $\eps$, and thus we write
$\bbL=\bbL\at\eps$, $\bbH=\bbH\at\eps$ and $\bssi=\bssi\at\eps$.
However, we always choose the macroscopic configuration space $P$ as
independent of $\eps$.
\bigpar
The \emph{two-scale transformations} considered in
\S\ref{sec:ChainNew} are compositions of a scaling transformation
$\calS_\con\at\eps:Q\rra{P}$, a moving-frame transformation
$\calM_\con\at\MiTime:Q\rra{Q}$ and a symmetry transformation
$\calT_\con:Q\rra{Q}$. More precisely, a general exact two-scale
transformation $\calT_\con\pair{\eps}{\MaTime}:Q\rra{P}$ and its
inverse $\bbT_\con\pair{\eps}{\MiTime}:P\rra{}Q$ are given by
\begin{align*}
\calT_\con\pair{\eps}{\MaTime}= \calS_\con\at\eps
\circ\calM_\con\bat{\MiTime\pair{\eps}\MaTime} \circ\calT_\con,\quad
\bbT_\con\pair{\eps}{\MiTime}= \calT_\con^{-1}
\circ\calM_\con\at{-\MiTime}\circ\bbS_\con\at\eps.
\end{align*}
For convenience we parametrize forward and backward transformations
by $\MaTime$ and $\MiTime$, respectively, i.e.,
\begin{align*}
\calT_\con\pair{\eps}{\MaTime\pair{\eps}{\MiTime}}
\circ\bbT_\con\pair{\eps}{\MiTime}=\mathrm{Id}_{P\to{P}},\quad
\bbT_\con\pair{\eps}{\MiTime\pair{\eps}{\MaTime}}
\circ\calT_\con\pair{\eps}{\MaTime}=\mathrm{Id}_{Q\to{Q}}.
\end{align*}
Moreover, the lifted transformations are given by
\begin{align*}
\calT_\vel\pair{\eps}{\MaTime}= \calS_\vel\at\eps\circ\calM_\vel
\bat{\MiTime\pair{\eps}\MaTime} \circ\calT_\vel,\quad
\bbT_\vel\pair{\eps}{\MiTime}=
\calT_\vel^{-1}\circ\calM_\vel\at{-\MiTime}\circ\bbS_\vel\at\eps.
\end{align*}
%
%

\begin{example}\label{ST:kdv.Example}%
The KdV reduction relies on the scaling $\MaTime=\eps^3\MiTime$,
$\MaLagr=\eps\MiLagrC$ and the two-scale ansatz
\begin{align*}
x\pair{\MiTime}{\MiLagrC}=\eps{X}\pair{\eps^3\MiTime}
{\eps\MiLagrC+\eps{}c\MiTime},
\end{align*} %
where $x\in{Q}=L^2\at{\Rset;\,\d\MiLagrC}$ and
$X\in{P}=L^2\at{\Rset;\,\d\MaLagr}$. From this ansatz we can
read-off directly the inverse two-scale transformation
\begin{math}
\bbT_\con\pair{\eps}{\MiTime}=
\calM_\con\at{-\MiTime}\circ\bbS_\con\at\eps,
\end{math}
which consist of the inverse scaling transformation
\begin{math}
\at{\bbS_\con\at\eps\,X}\at{\MiLagrC}=\eps{X}\at{\eps\MiLagrC}
\end{math} %
and the inverse of the moving-frame transformation
%
\begin{align*}
\at{\calM_\con\at{\MiTime}\,x}\at\MiLagrC=
x\at{\MiLagrC-c\,\MiTime}.
\end{align*}
Moreover, %
\begin{math}
\calI\pair{x}{x_\MiTime}=-c\,\calI_{\mathrm{space}}\pair{x}{x_\MiTime}=-c
\int_{\Rset}x_\MiTime\at\MiLagrC\,x_\MiLagrC\at\MiLagrC\,\d\MiLagrC
\end{math} %
is the integral of motion associated to $\calM_\con\at{\MiTime}$.
\end{example}%
For later purposes we prove two auxiliary results. The first lemma
describes how to restrict Lagrangian and Hamiltonian structures to
subspaces of $P$, and the second one allows us to compute $\bsSi$
from $\Si$, the matrix-valued maps corresponding to $\bssi$ and
$\si$, respectively.
\begin{lemma}
\label{TST:Lemma.SubSpaces}%
Let $\tilde{P}\subset{P}$ be a closed subspace of $P$, embedded via
a linear and continuous operator
$\tilde{\bbJ}_\con:\tilde{P}\hookrightarrow P$ with canonical lift
$\tilde{\bbJ}_\vel=\tilde{\bbJ}_\con\times\tilde{\bbJ}_\con:
T\tilde{P}\hookrightarrow TP$. Moreover, let
$\tilde{\bbL}=\bbL\circ\tilde{\bbJ}_\vel$ be the restricted
Lagrangian on $T\tilde{P}$ and $\npair{\tilde{\bbH}}{\tilde{\bssi}}$
the associated Hamiltonian structure. Then,
$\tilde{\bbH}={\bbH}\circ\tilde{\bbJ}_\vel$ and
$\tilde{\bssi}=\nnat{\tilde{\bbJ}_\vel}^\ast\bssi$.
\end{lemma}
\begin{proof}
The projector corresponding to $\tilde{\bbJ}_\vel$ is denoted by
$\bbJ_\vel:TP\twoheadrightarrow{}T\tilde{P}$ and equals the adjoint
of $\tilde{\bbJ}_\vel$. Notice that $\bbJ_\vel\circ\tilde{\bbJ}_\vel
=\mathrm{Id}_{T\tilde{P}}$ but $\mathrm{ker}\,
(\tilde{\bbJ}_\vel\circ\bbJ_\vel) \supsetneq\{0\}$ for
$\tilde{P}\subsetneq{P}$. Besides this modification the proof is
entire similar to that of Theorem \ref{LT:Theo}.
\end{proof}
\begin{lemma}
\label{TST:Lemma.SymTrafo} %
Let $\eps$ and $\MiTime$ be fixed and suppose there exist two linear
and invertible transformations \mbox{$\bsS:TP\rra{TQ}$} with
\begin{math}
\skp{\bsS\,Z}{\bsS\,\tilde{Z}}= \eps^{\mu}\sskp{Z}{\tilde{Z}}
\end{math} %
for some $\mu$ and $\bsT:TP\rra{}TP$ such that
\begin{math}
\bbT_\vel\pair{\eps}{\MiTime}=\bsS\circ\bsT.
\end{math}
Then, %
\begin{align}
\label{TST:Lemma.SymTrafo.Eqn1}
\bsSi|_Z=\tfrac{\d\MaTime}{\d\MiTime}\eps^{\mu}
\bsT^\prime\bsS^{-1}\Sigma|_{\bsS\bsT{}Z}\bsS\bsT,
\end{align}
where $\bsT^\prime$ is the adjoint to $\bsT$.
\end{lemma}
\begin{proof}
Let $Z\in{TP}$ be fixed and choose two arbitrary tangent vectors
$\dot{Z},\,\acute{Z}\in{}\tspace{Z}{TP}$. Moreover, set
$z=\bsS\bsT{Z}$ and $\dot{z}=\bsS\bsT\dot{Z}$,
$\acute{z}=\bsS\bsT\acute{Z}$. The definition of $\bssi$ and the
linearity of $\bbT_\vel\pair{\eps}{\MiTime}=\bsS\circ\bsT$ imply
\begin{math}
\bssi|_Z\npair{\dot{Z}}{\acute{Z}}=\tfrac{\d\MaTime}{\d\MiTime}\,
\si|_z\pair{\dot{z}}{\acute{z}},
\end{math}
and this gives
\begin{align*}
\sskp{\bsSi|_z\dot{Z}}{\acute{Z}}&=\tfrac{\d\MaTime}{\d\MiTime}\,
\skp{\Si|_z\dot{z}}{\acute{z}}= \tfrac{\d\MaTime}{\d\MiTime}\,
\skp{\bsS\bsS^{-1}\Si|_z{\bsS\bsT\dot{Z}}}{\bsS\bsT\acute{Z}}
\\&=%
\tfrac{\d\MaTime}{\d\MiTime}\,\eps^{\mu}
\sskp{\bsS^{-1}\Si|_z{\bsS\bsT\dot{Z}}}{\bsT\acute{Z}}=
\tfrac{\d\MaTime}{\d\MiTime}\,\eps^{\mu}
\sskp{\bsT^\prime\bsS^{-1}\Si|_z{\bsS\bsT\dot{Z}}}{\acute{Z}},
\end{align*}
the desired result.
\end{proof}
\begin{remark}
\label{TST:Remark.SymTrafo} %
For $\Si=\Si_\mtrc$ and $\bsS=\tilde{\bsS}\times\tilde{\bsS}$ with
$\tilde{\bsS}:P\to{Q}$ we have $\bsS^{-1}\Si\bsS=\bsSi_\mtrc$, where
$\Si_\mtrc$ and $\bsSi_\mtrc$ correspond to the metric symplectic
forms on $TQ$ and $TP$, respectively, see Example
\ref{GS:Ex.Lagrangian}. In this case \eqref{TST:Lemma.SymTrafo.Eqn1}
becomes
\begin{math}
\bsSi=\tfrac{\d\MaTime}{\d\MiTime}\eps^\mu
\bsT^\prime\bsSi_\mtrc\bsT.
\end{math}
\end{remark}
%
\subsection{Reduction principles}%
\label{sec:Reductions}
%
In this section we suppose that an exact two-scale transformation
has already transformed the original microscopic system into a
macroscopic one on $TP$, where $P$ is a Hilbert space with inner
product $\sskp{\cdot}{\cdot}$. As before, the macroscopic system has
Lagrangian $\bbL\at\eps$, and the associated Hamiltonian structure
on $TP$ is given by $(\bbH\at{\eps},\bssi\at\eps)$. In the previous
section we have shown how $\bbH\at\eps$ and $\bssi\at\eps$ can be
computed directly from their microscopic counterparts, but
$\bbH\at{\eps}$ is always the macroscopic Legendre transform of
$\bbL\at\eps$, and $\bssi\at\eps$ equals
$\nnat{\mathfrak{F}\bbL\at\eps}^\ast\mhol{\bssi}_\can$.
\par
In what follows we describe how the explicit dependence on $\eps$
allows for a consistent model reduction. As illustrated in
\S\ref{sec:ChainNew}, a typical two-scale transformation provides an
expansion of the macroscopic Lagrangian in powers of the scaling
parameter $\eps$, i.e., we have
\begin{align}
\label{Red:Exp.L} \bbL\at{\eps}=\eps^\kappa\,\at{\bbL_0+
\eps\bbL_1+\eps^2\bbL_2+\tdots}
\end{align}
at least on a formal level, where $\kappa$ can be positive or even
negative depending on the underlying two-scale ansatz. Recall that
such an expansion is not available for the original microscopic
system.
\par%
Since we deal only with Hamiltonian structures on tangent bundles we
benefit from Principle \ref{POCE}. In particular, the expansions
\begin{align}
\label{Red:Exp.HandSi}
\bbH\at{\eps}=\eps^\kappa\,\at{\bbH_0+\eps\bbH_1+
\eps^2\bbH_2+\tdots},\quad
\bssi\at{\eps}=\eps^\kappa\,\at{\bssi_0+\eps\bssi_1+
\eps^2\bssi_2+\tdots}
\end{align}
are consistent with \eqref{Red:Exp.L}, i.e., $\bbH_i$ is the
Legendre transform of $\bbL_i$, and we have
$\bssi_i=\nnat{\mathfrak{F}\bbL_i}^\ast\mhol{\bssi}_\can$.
\bigpar
In the simplest case the reduced model is obtained by considering
the leading order terms for $\bbL$ and $\npair{\bbH}{\bssi}$, and
ignoring all terms that contribute to higher orders in $\eps$.
However, depending on the two-scale ansatz the leading order system
can be degenerate. For this reason we distinguish the following
cases:
\begin{list}{}{
\setlength{\labelwidth}{3cm}%
\setlength{\leftmargin}{0.15\textwidth}%
\setlength{\rightmargin}{0.\textwidth}%
}
\item[Case A:]
The symplectic form $\bssi_0$ is non-degenerate, i.e. there is no
$\pair{Z}{Z_\MaTime}\in{TTP}$ with
$\bssi_0|_Z\pair{Z_\MaTime}{\cdot}\equiv0$.
\item[Case B:]
$\bssi_0$ is degenerate, but $\bbL_0$ depends on $X_\MaTime$.
\item[Case C:]
The leading order Lagrangian $\bbL_0$ is quasi-stationary, this
means independent of $X_\MaTime$, and this yields $\bbH_0=-\bbL_0$
and $\bssi_0=0$.
\end{list}
%
%
\paragraph{Reduction in Case A}
%
%
Whenever we end up with Case A, the formal reduction provides a
non-degenerate macroscopic Hamiltonian system and thus we have
established already a (formal) micro-macro transition. In
particular, the reduced Lagrangian reads $\bbL^\red=\bbL_0$ and the
associated Hamiltonian structure is given by
$(\bbH^\red,\bssi^\red)=(\bbH_0,\bssi_0)$, so that the macroscopic
Hamiltonian equation on $TP$ is given by
\begin{align}
\label{Red:MacroDynSys} \bssi^\red|_Z\pair{Z_\MaTime}{\cdot}=
\d\bbH^\red|_Z\at{\cdot}.
\end{align}
Recall, that we can neglect the pre-factor $\eps^\kappa$ as it drops
out in both the Lagrangian and Hamiltonian equations on $TP$.
\bigpar%
Since we have derived the reduced macroscopic structures by means of
formal expansions with respect to $\eps$, we are confronted with the
\emph{justification problem}. More precisely, it is not obvious that
solutions to \eqref{Red:MacroDynSys} provide (approximate) solutions
to the microscopic system. Of course, any curve
$\MaTime\mapsto{Z}\at{\MaTime}\in{}TP$ that solves
\begin{math}
\bssi|_Z\pair{Z_\MaTime}{\cdot}=\d\bbH\at{\cdot}
\end{math}
and that obeys an expansion in powers of $\eps$, must satisfy
\eqref{Red:MacroDynSys} to leading order, but the existence of such
an expansion for the solution $Z(\tau)$ must be proven. This problem
is very subtle and cannot be addressed here. Rigorous justification
results for linear and some (weakly) nonlinear systems are given in
\cite{Mie07}. For a brief discussion of the difficulties that arise
in the case of strong nonlinearities we refer to
\S\ref{sec:ChainNew.WE}, which shows that such an $\eps$-expansion
can be valid only under additional assumptions concerning the
initial data, the macroscopic time-interval under consideration and,
finally, the regularity properties of the macroscopic equation.
%
%
\paragraph{Reduction in Case B}
%
%
In contrast to Case $A$, the Cases $B$ and $C$ allow for further
reduction steps, which we explain next. We start with Case B and
refer to the KdV reduction in \S\ref{sec:ChainNew.kdv} as a typical
example. For simplicity we suppose that $\bbL_0$ depends linearly on
$X_\MaTime$, i.e., we assume that the momentum
$\Pi_0=\VarDer{\bbL_0}{X_\MaTime}$ is a function of $X$ but not of
$X_\MaTime$. As a consequence, the associated Hamiltonian structure
lives on $P$, this means $\bbH_0$ is a function on $P$ and $\bssi$
is a symplectic form on $P$. In fact,
\begin{math}
\bbH_0=\sskp{\Pi_0\at{X}}{X_\MaTime}-\bbL_0\pair{X}{X_\MaTime}
\end{math}
provides $\VarDer{\bbH_0}{X_\MaTime}=0$. This implies that the right
hand side in
\begin{align*}
{D\Pi_0}|_{\pair{X}{X_\MaTime}}\npair{\dot{X}}{{\dot{X}_\MaTime}}=
\VarDer{\Pi_0}{X}|_{X}\nnat{\dot{X}}
\end{align*}
is independent of both $X_\MaTime$ and $\dot{X}_\MaTime$, and due to
\eqref{GS:CanSymplForm} the form $\bssi_0$ lives actually on $P$.
Thus we end up with the following macroscopic model. The reduced
Lagrangian $\bbL^\red=\bbL_0$ lives on $TP$ and has a consistent
Hamiltonian structure on $P$ given by $\bbH^\red=\bbH_0|_P$ and
$\bssi^\red={\bssi_0}|_P$.
%
%
%
\paragraph{Reduction in Case C: Restriction to sub-spaces}
%
%
In some cases the leading order reduction turns out to be
quasi-stationary, i.e., $\bbL_0$ does not depend on $X_\MaTime$, and
this implies $\bbH_0=-\bbL_0$ and $\bssi_0=0$. Whenever this
happens, we obtain a reduced macroscopic model as follows. We
restrict the macroscopic configurations to
\begin{align*}
P_0=\left\{{X_0}\in{P}\;:\;0= \VarDer{\bbL_0}{X}\at{X_0}=
\VarDer{\bbH_0}{X}\at{X_0}\;\right\},
\end{align*}
and determine the reduced Lagrangian and Hamiltonian structures by
restricting the next-leading order terms $\bbL_i$ and
$\pair{\bbH_i}{\bssi_i}$ to $P_0$. However, in general we shall
expand additionally the solution $X$ in powers of $\eps$, and this
may produce correction terms in the expansions \eqref{Red:Exp.L} and
\eqref{Red:Exp.HandSi}. This problem will be discussed now, where
for our purposes we can assume that $P_0$ is a closed linear
subspace of $P$.
\par
In order to identify suitable correction terms we start with the
ansatz
\begin{align*} %
Z=\pair{X}{X_\MaTime}=Z_0+\eps{}Z_1= \pair{X_0}{X_{0\,\MaTime}} +
\eps \pair{X_1}{X_{1\,_\MaTime}},
\end{align*}
and study the Lagrangian
$\tilde{\bbL}\triple{\eps}{Z_0}{Z_1}=\bbL\pair{\eps}{Z_0+\eps{Z}_1}$
defined on $T\widetilde{P}$ with $\widetilde{P}=P_0\times{P}$.
Exploiting $\VarDer{\bbL_0}{X_\MaTime}\equiv0$ and
$\VarDer{\bbL_0}{X}\at{X_0}=0$ for all $X_0\in{P_0}$ we find
\begin{align*}
\widetilde{\bbL}\triple{\eps}{{Z_0}}{Z_1}&=
\eps^\kappa\Bat{\bbL_0\at{X_0+\eps{}X_1}+
\eps\bbL_1\at{Z_0+\eps{}Z_1}+ \eps^2\bbL_2\at{Z_0+\eps{}Z_1}+\tdots}
\\%
&= \eps^\kappa\Bat{\bbL_0\at{X_0}+ \eps\bbL_1\at{Z_0}+ \eps^2\at{
\bbL_2\at{Z_0}+\widehat{\bbL}_2\pair{Z_0}{Z_1}}+\tdots},
\end{align*}
with first correction term
\begin{align*}
\widehat{\bbL}_2\pair{Z_0}{Z_1}=\tfrac{1}{2}\sskp{\SVarDer{\bbL_0}{X}\at{X_0}X_1}{X_1}
+\sskp{\VarDer{\bbL_1}{X}\at{Z_0}}{X_1}+
\sskp{\VarDer{\bbL_1}{X_\MaTime}\at{Z_0}}{{X_1}_\MaTime}.
\end{align*}
In particular, any possible correction $\eps{Z_1}$ effects $\bbL_2$
but neither $\bbL_0$ nor $\bbL_1$.

\paragraph{Case C1: Reduced model via $\bbL_1$}
If the next-leading order Lagrangian $\bbL_1$ depends on
$X_\MaTime$, then the reduced Lagrangian is given by
$\bbL^\red=\bbL_1|_{TP_0}$, and in this case we can ignore the
correction term $\eps{}Z_1$. Moreover, according to Lemma
\ref{TST:Lemma.SubSpaces} the corresponding Hamiltonian structure is
given by $\pair{\bbH_1}{\bssi_1}|_{TP_0}$. An example for this case
is the three-wave-interaction discussed in
\S\ref{sec:ChainNew.threeWaves}.
\bigpar%
As before, the reduction to $TP_0$ is formal and must be justified
rigorously. In the simplest case the space $P_0$ is an
\emph{invariant manifold} for $\pair{\bbH}{\bssi}$. This means that
for all initial data chosen from $TP_0$ the solution
$\MaTime\mapsto{Z}\at\MaTime$ to the original problem  belongs to
$TP_0$ for all times $\MaTime>0$. In general, we expect that the
restriction to $TP_0$ provides a reasonable reduced model if $P_0$
is an \emph{approximate invariant manifold}, so that solutions to
$\npair{\bbH^\red}{{\bssi}^\red}$ are approximate solutions to
$\npair{\bbH}{\bssi}$. For the \emph{justification} in this case one
has to prove that for all initial data chosen from $TP_0$ the real
trajectory stays close to $TP_0$ (up to higher orders in $\eps$) at
least for sufficiently small \emph{macroscopic} times, see for
instance \cite{SW00,GM04,GM06,GMS07}.

\begin{example}
\label{RedCaseC:Ex}
Let $M$ be an integer, $Q=L^2\at{[0,\,M];\,\d\MiLagrC}$ the Lebesgue
space of all \emph{periodic} functions on the interval $[0,\,M]$,
and let $\calL\pair{x}{x_\MiTime}=\calK\at{x_\MiTime}-\calV\at{x}$
be defined by
\begin{align*}
\calK\at{x_\MiTime}=\tfrac{1}{2}\skp{x_\MiTime}{x_\MiTime},\quad
\calV\at{x}=-\tfrac{1}{2}\skp{\laplace{}x}{x},\quad
\skp{x}{\tilde{x}}=\int\mhintlimits_0^M x\,\tilde{x}\,\d\MiLagrC,
\end{align*}
with discrete Laplacian
$\nnat{\laplace{x}}\at{\MiLagrC}=x\at{\MiLagrC+1}+x\at{\MiLagrC-1}-2x\at{\MiLagrC}$,
so that the microscopic law of motion is the discrete wave equation
$x_{\MiTime\MiTime}=\laplace{x}$. Moreover, consider the time
scaling from Example \ref{Ex:TimeScaling}, this means
$\MaTime=\eps\MiTime$, $x=X$, $P=Q$, $\MaLagr=\MiLagrC$. Then,
$\bbL$ obeys an (exact) expansion in powers of $\eps^2$ via
\begin{align*}
\bbL\triple{\eps^2}{X}{X_\MaTime}=\bbL_0\at{X}+\eps^2\bbL_1\at{X_\MaTime},\quad
\bbL_0\at{X}=-\calV\at{X},\quad\bbL_1\at{X_\MaTime}=\calK\at{X_\MaTime}.
\end{align*}
The leading order Lagrangian and Hamiltonian equations read
$\laplace{X}=0$ and provide
\begin{align*}
P_0=\{X_0\in P\;:\;X_0\at{\MaLagr+1}=X_0\at{\MaLagr}\,\}.
\end{align*}
Exploiting the next-leading order terms corresponding to $\bbL_1$ we
find
\begin{align*}
{\bbL}^\red={\bbH}^\red=\tfrac{1}{2}\skp{{X_0}_\MaTime}{{X_0}_\MaTime},\quad
{\bssi}^\red\simeq
\begin{pmatrix}\scriptstyle0&-\scriptstyle{1}\\\scriptstyle{\scriptstyle{1}}&\scriptstyle0\end{pmatrix},
\end{align*}
so the macroscopic evolution is governed by
${X_0}_{\MaTime\MaTime}=0$. Moreover, the reduction is exact as both
microscopic and reduced dynamics are equivalent for all initial data
$\pair{X\at{0}}{X_\MaTime\at{0}}\in{}TP_0$.
\end{example}
%
%
%
%
\paragraph{Case C2: Reduced model via $\bbL_2$}
%
%
It may happen that even the next-leading Lagrangian $\bbL_1|_{TP_0}$
is quasi-stationary,
 i.e., $\bbL_1(Z_0)=\bbL_1(X_0)$ for all $Z_0=\pair{X_0}{{X_0}_\MaTime}\in{TP_0}$. Then
the general reduction procedure depends on the particular properties
of $\bbL_1$. Here we restrict to the case we meet in
\S\ref{sec:ChainNew.nls} (nlS equation), where the two-scale
transformation implies
\begin{align}
\label{Red:CaseC2a} \bbL_1|_{P_0}=-\bbH_1|_{P_0}=0.
\end{align}
For $\bbL_1|_{P_0}\neq\const$ we would restrict $X_0$ further by
imposing additionally $\VarDer{\at{\bbL_1|_{P_0}}}{X_0}=0$.
\par
Notice that \eqref{Red:CaseC2a} does \emph{not} necessarily imply
$\VarDer{\bbL_1}{X}(X_0)=0 \in \text{Lin}(P,\Rset)$ for all $X_0 \in
P_0$ and therefore we proceed as follows. Our strategy is to choose
$X_1$ in such a way that it is a stationary point of
$\widehat{\bbL}_2\pair{X_0}{X_1}$. This means we seek
$X_1=X_1\at{X_0}$  as solution to the affine equation
\begin{align*}
\SVarDer{\bbL_0}{X}\at{X_0}X_1+ \VarDer{\bbL_1}{X}\at{X_0}=0.
\end{align*}
Provided this is possible, our reduced Lagrangian on $TP_0$ is given
by
\begin{align*}
\bbL^\red\at{Z_0}=\bbL_2\at{Z_0}+\widehat{\bbL}_2\pair{X_0}{X_1\at{X_0}},
\end{align*}
and since the term $\widehat{\bbL}_2$ does not contribute to the
fiber derivative $\mathfrak{F}\bbL^\red$, one can show (similarly to
Lemma \ref{TST:Lemma.SubSpaces}) that
\begin{align*}
\bbH^\red\at{Z_0}=\bbH_2\at{Z_0}+\widehat{\bbH}_2\pair{X_0}{X_1\at{X_0}},\quad
\bssi^\red=\bssi_2|_{TP_0},
\end{align*}
where $\widehat{\bbH}_2=-\widehat{\bbL}_2$ is the corresponding
Hamiltonian structure on $TP_0$.
%
%
%
\section{Two-scale reductions for the atomic chain}
\label{sec:ChainNew}
%
The abstract framework developed in the previous section shall now
be applied to the examples mentioned in the introduction. The
microscopic system will be either the Fermi--Pasta--Ulam (FPU) chain
\begin{align}
\label{Intro:FPUChain} %
\ddot{x}_\MiLagr\at\MiTime=
\NNPot^\prime\bat{{x}_{\MiLagr+1}\at\MiTime-{x}_{\MiLagr}\at\MiTime}-
\NNPot^\prime\bat{{x}_{\MiLagr}\at\MiTime-{x}_{\MiLagr-1}\at\MiTime}
\end{align}
or the Klein--Gordon (KG) chain
\begin{align}
\label{Intro:KGChain}%
\ddot{x}_\MiLagr\at\MiTime=
\alpha\bat{{x}_{\MiLagr+1}\at{\MiTime}+{x}_{\MiLagr-1}\at{\MiTime}-
2{x}_{\MiLagr}\at{\MiTime}}-\OSPot^\prime\at{{x}_{\MiLagr}\at{\MiTime}}
\end{align}
with harmonic constant $\alpha=\NNPot^{\prime\prime}\at{0}\in\Rset$.
Without loss of generality we always assume
$0=\OSPot\at{0}=\OSPot^\prime\at{0}=\NNPot\at{0}=\NNPot^\prime\at{0}$,
and restrict our considerations to infinite chains. In the case that
the two-scale ansatz refers to small amplitudes, the linearized
atomic chain
\begin{align}
\notag
\ddot{x}_\MiLagr\at\MiTime=
\alpha\bat{{x}_{\MiLagr+1}\at{\MiTime}+{x}_{\MiLagr-1}\at{\MiTime}-
2{x}_{\MiLagr}\at{\MiTime}}-\OSPot^{\prime\prime}\at{0}{x}_{\MiLagr}
\end{align}
becomes important. This linearized chain allows for propagating
plane wave solutions
\begin{math}
\mhexp{\iu\at{\om\MiTime+\theta\MiLagr}},
\end{math}
provided that the \emph{frequency} $\om$ and the \emph{wave number}
$\theta$ satisfy the \emph{dispersion relation}
\begin{align}
\label{Intro:AtomicChain.DispRel}
\om^2={\Om}^2\at{\theta}=2\alpha\bat{1-\cos\theta}+\OSPot^{\prime\prime}\at{0}.
\end{align}
The atomic chain falls into the class of normal Hamiltonian systems,
see \eqref{TypEx:Normal.Lagrangian}, with configuration space
$Q_\discr=\ell^2\at\Zset$. The Lagrangian reads
\begin{math}
\calL_\discr\pair{x}{\dot{x}}=
\calK_\discr\at{\dot{x}}-\calV_\discr\at{x}
\end{math} %
with kinetic and potential energy given by
\begin{align}
\label{chain:energies} \calK_\discr\at{\dot{x}}=
\tfrac{1}{2}\sum_{\MiLagr\in\Zset}\,
\dot{x}_\MiLagr^2,\quad%
\calV_\discr\at{x}= \sum_{\MiLagr\in\Zset}\,\Bat{\,
\NNPot\at{x_{\MiLagr+1}-x_\MiLagr}+ \OSPot\at{x_\MiLagr}\,},
\end{align}
and Newton's equations \eqref{Intro:AtomicChain} equal the
Euler-Lagrange equations to $\calL_\discr$ on $T Q_\discr$.
Moreover, the Hamiltonian is given by
$\calH_\discr\pair{x}{\dot{x}}=\calK_\discr\at{\dot{x}}+\calV_\discr\at{x}$,
so that Newton's equations are equivalent to
\begin{align*}
\begin{pmatrix}
0&-1\\1&0
\end{pmatrix}
\frac{\d}{\d\MiTime}
\begin{pmatrix}
x\\\dot{x}
\end{pmatrix}
=
\begin{pmatrix}
\VarDer{\calH_\discr}{x}\\\VarDer{\calH_\discr}{\dot{x}}
\end{pmatrix},
\end{align*}
which is a Hamiltonian ODE on $TQ_\discr$ with metric symplectic
form, i.e., we have $\si=\si_\mtrc$ in the sense of Example
\ref{GS:Ex.Lagrangian}.
%
%
\subsection{The embedded atomic chain}\label{sec:Chain.Embed}
%
%
In order to derive effective models we start with a suitable
embedding of the atomic chain. At first we replace the discrete
lattice index $\MiLagr\in\Zset$ by a continuous variable
$\MiLagrC\in\Rset$. In addition, if the two-scale ansatz involves
oscillatory microstructure, we consider $k$ additional \emph{phase
variables} $\Phase=(\Phase_1,\,..,\,\Phase_k)$, which are supposed
to take values in the $k$-dimensional torus $T^k$. This embedding
 gives rise to the formal
identification
\begin{align*}
x_\MiLagr\at\MiTime=x\triple{\MiTime}{\MiLagr}{0},\quad
\dot{x}_\MiLagr\at\MiTime=x_\MiTime\triple{\MiTime}{\MiLagr}{0},
\end{align*}
where the instantaneous configuration
$x\triple{\MiTime}{\cdot}{\cdot}$ is for each $\MiTime$ a function
in $\MiLagrC$ and $\Phase$.
\par%
Next, we identify the Lagrangian $\calL$ of the embedded system. To
this end, we replace all sums over $\MiLagr$ in
\eqref{chain:energies} by integrals with respect to $\MiLagrC$ and
$\Phase$. This yields
\begin{align}
\label{GS:chain.Lagrangian}
\calL\pair{x}{x_\MiTime}=\calK\at{x_\MiTime}-\calV\at{x}
\end{align}
with
\begin{align}
\label{GS:Chain.energies}
\calK\at{x_\MiTime}=\!\!\!\int\mhintlimits_{\Rset\times{T^k}}\!\!\!
\tfrac{1}{2}x_\MiTime^2\,\d\MiLagrC\d\Phase,\quad \calV\at{x}=
\!\!\!\int\mhintlimits_{\Rset\times{T^k}}\!\!\!
\Bat{\NNPot\at{\nabla^+_{1,\,0}x}+
\OSPot\at{x}}\,\d\MiLagrC\d\Phase,
\end{align}
where $\nabla^+_{1,\,0}$ is a discrete differential operator, see
Remark \ref{GS:DiscreteNablas} below. Notice that the Euler-Lagrange
equation for $\calL$, i.e.\
\begin{align*}
x_{\MiTime\MiTime}=
\nabla^-_{1,\,0}\NNPot^\prime\at{\nabla^+_{1,\,0}x}
-\OSPot^\prime\bat{x},\quad{x}=x\triple{\MiTime}{\MiLagrC}{\Phase},
\end{align*}
is still fully equivalent to (an uncountable number of uncoupled
copies of) Newton's equations \eqref{Intro:AtomicChain}. However,
the embedding gives rise to additional symmetries, and thus we gain
new integrals of motion. In fact, the Lagrangian
\eqref{GS:chain.Lagrangian} is invariant under the continuous groups
of space shifts $\MiLagrC\rightsquigarrow\MiLagrC+\MiLagrC_0$ and
phase shifts $\Phase\rightsquigarrow\Phase+\Phase_0$, and Noether's
theorem provides that
\begin{align}
\notag
\calI_{\mathrm{space}}\pair{x}{x_\MiTime}=
\int\mhintlimits_{\Rset{\times}T^k}x_\MiTime\,x_\MiLagrC\,
\d\MiLagrC\d\Phase\in\Rset,\quad
\calI_{\mathrm{phase}}\pair{x}{x_\MiTime}=
\int\mhintlimits_{\Rset{\times}T^k}x_\MiTime\,x_\Phase\,
\d\MiLagrC\d\Phase\in\Rset^k
\end{align}
are conserved for any solution of the microscopic system. Recall
that $x_\MiLagrC \in \Rset $ and $x_\Phase \in \Rset^k$ denote the
derivatives of $x$ with respect to $\MiLagrC$ and $\Phase$,
respectively. These conservation laws have no counterpart within the
classical mechanics of mass points as they are a byproduct of the
embedding.
\bigpar%
\begin{remark}%
\label{GS:DiscreteNablas}%
For given $\delta\in\Rset$ and $\theta\in{T^k}$ let
\begin{align}
\notag
\nnat{\nabla^+_{\delta,\,\theta}x} \pair{\MiLagrC}{\Phase}=
x\pair{\MiLagrC+\delta}{\Phase+\theta}-
x\pair{\MiLagrC}{\Phase},\quad
\nnat{\nabla^-_{\delta,\,\theta}x}\pair{\MiLagrC}{\Phase}=
x\pair{\MiLagrC}{\Phase}-x\pair{\MiLagrC-\delta} {\Phase-\theta},
\end{align}
and
\begin{math}
\laplace_{\delta,\,\theta}=
\nabla^+_{\delta,\,\theta}-\nabla^-_{\delta,\,\theta}.
\end{math}
These definitions imply
\begin{align*}
\nabla^\pm_{-\delta,\,-\theta}=-\nabla^\mp_{\delta,\,\theta},\quad
\at{\nabla^\pm_{\delta,\,\theta}}^\ast=
-\nabla^\mp_{\delta,\,\theta},\quad
\at{\laplace_{{\delta,\,\theta}}}^\ast=
\laplace_{{\delta,\,\theta}},
\end{align*} %
where $\ast$ denotes the adjoint operator with respect to the
$L^2$--inner product.
\end{remark}

%
\subsection{From FPU to the wave equation}\label{sec:ChainNew.WE}
%
%
Here we derive the quasi-linear wave equation from Newton's equation
for FPU chains. Recall that the underlying two-scale ansatz is given
by \eqref{Intro:WE.MSAnsatz}, and involves neither a microstructure
nor a moving frame. For the embedded system this ansatz reads
\begin{align}
\label{WE:MSAnsatz}
x\pair{\MiTime}{\MiLagrC}=\eps^{-1}X\pair{\eps\MiTime}{\eps\MiLagrC},
\end{align}
and Example \ref{ST:Example1} provides
\begin{math}
Q=L^2\at{\Rset;\,\d\MiLagrC}
\end{math} %
and
\begin{math}
P=L^2\at{\Rset;\,\d\MaLagr}
\end{math} %
as well as the (lifted) inverse two-scale transformation
%
%
%
\begin{align}
\label{WE:MSTrafo.Eqn2}%
\bbT_\vel\at{\eps}: %
 \pair{X}{X_\MaTime}\at{\MaLagr}\mapsto
\pair{x}{x_\MiTime}\at{\MiLagrC}=\pair{\eps^{-1}{X}}
{{X_\MaTime}}\at{\eps\MiLagrC}.
\end{align}%
\begin{lemma} %
The two-scale transformation \eqref{WE:MSTrafo.Eqn2} yields
$\bbL=\bbK-\bbV$ and $\bbE=\bbK+\bbV$ with
\begin{align*}
\bbK\pair{\eps}{X_\MaTime}=\eps^{-1}\,\bbK_0\at{X_\MaTime},\quad
\bbK_0\at{X_\MaTime}=\!\!\!\int\mhintlimits_{\Rset}\tfrac{1}{2}X_\MaTime^2
\d\MaLagr,\quad \bbV\pair{\eps}{X}=\eps^{-1}
\!\!\!\int\mhintlimits_{\Rset}
\NNPot\at{\eps^{-1}\nabla^+_\eps{X}}\d\MaLagr.
\end{align*}
Moreover, we have $\bbH=\bbE$ and $\bssi=\eps^{-1}\bssi_\mtrc$,
where $\bssi_\mtrc$ is the metric symplectic form on $TP$, see
Example \ref{GS:Ex.Lagrangian}.
\end{lemma}
\begin{proof}
All assertions are direct consequences of \eqref{WE:MSTrafo.Eqn2}
and the abstract results from \S\ref{sec:Trafos}, see Theorem
\ref{ST:Theo1}. In particular, $\bbH=\bbE$ holds, since the
two-scale transformation does not involve a moving frame.
\end{proof}
Next we identify the leading order terms in the expansion with
respect to $\eps$. Using formal Taylor expansion
\begin{align*}
\at{\eps^{-1}\nabla^+_\eps{X}}\at\MaLagr=
\eps^{-1}\at{X\at{\MaLagr+\eps}-X\at{\MaLagr}}=
X_\MaLagr\at\MaLagr+\tfrac{1}{2}\eps{X}_{\MaLagr\MaLagr}\at\MaLagr+\DO{\eps^2},
\end{align*}
we find
\begin{math}
\bbV\pair{\eps}{X}=\eps^{-1}\bbV_0\at{X}+\DO{1}
\end{math} %
with
\begin{math}
\bbV_0\at{X}=\int_\Rset \,\NNPot\at{X_\MaLagr}\,\d\MaLagr
\end{math} %
and conclude that
\begin{align*}
\bbL^\red\pair{X}{X_\MaTime}=
\bbK_0\at{X_\MaTime}-\bbV_0\at{X},\quad
\bbH^\red\pair{X}{X_\MaTime}= \bbK_0\at{X_\MaTime}+\bbV_0\at{X}.
\end{align*}
Notice that $\bbV_0$ is defined only on $H^1\at{\Rset;\,\d\MaLagr}$,
which is dense in $P$. Finally, $\bssi^\red=\bssi_\mtrc$ completes
the leading order reduction and we end up with the following
macroscopic model:
\begin{theorem}
%
Both the formally reduced Lagrangian and Hamiltonian equations are
equivalent to
\begin{align}
\label{WE:Red.Theorem1.Eqn1}
X_{\MaTime\MaTime}-{\NNPot^\prime\at{X_\MaLagr}}_\MaLagr=0.
\end{align}%
\end{theorem}
\begin{remark}
We claimed in the introduction that the Hamiltonian two-scale
reduction is always related to the Hamiltonian structure on $TP$ but
fails if we use the canonical structure on $T^{\ast}P$. In this
example we clearly see the reason for this: The canonical momentum
corresponding to $\bbL\at{\eps}$ is given by
$\Pi=\eps^{-1}X_\MaTime$ and, replacing $X_\MaTime$ by $\Pi$, we
find
\begin{align*}
\mhol{\bbH}\triple{\eps}{X}{\Pi}=
\sskp{\VarDer{\bbL}{X_\MaTime}}{X_\MaTime}-\bbL=
\eps\int\mhintlimits_\Rset \tfrac{1}{2}\Pi^2\,\d\MaLagr+
\eps^{-1}\int\mhintlimits_\Rset{}
\NNPot\at{\eps^{-1}\nabla^+_\eps{X}}\,\d\MaLagr,
\end{align*}
the Hamiltonian on $T^\ast{P}$. As long as we fix $\eps>0$, the
canonical equations
\begin{math}
X_\MaTime=\eps\Pi
\end{math}
and
\begin{math}
\Pi_\MaTime=
\eps^{-2}\nabla^-_\eps{}\NNPot^{\prime}\at{\eps^{-1}\nabla^+_\eps{X}}
\end{math}
are fully equivalent to the Hamiltonian equations on $TP$. However,
formal expansion of $\mhol{\bbH}$ with respect to $\eps$ yields, to
leading order $\eps^{-1}$, the reduced Hamiltonian
$\mhol{\bbH}^\red\pair{X}{\Pi}=\eps^{-1}\bbV_0\at{X}$ and the
corresponding canonical equations
\begin{math}
X_\MaTime=0,
\end{math} %
\begin{math}
\Pi_\MaTime=\eps^{-1}{\NNPot^\prime\at{X_\MaLagr}}_\MaLagr
\end{math} %
are apparently different from the wave equation
\eqref{WE:Red.Theorem1.Eqn1}. Of course, here we can overcome this
problem by multiplying $\bbL$ with $\eps$, but this is not always
possible as the KdV  reduction in \S\ref{sec:ChainNew.kdv} shows.
\end{remark}
To conclude this section we discuss some aspects of
\eqref{WE:Red.Theorem1.Eqn1} which are closely related to the
justification problem. In particular, it comes out that
\eqref{WE:Red.Theorem1.Eqn1} can provide a reasonable macroscopic
model for the FPU chain only under additional assumptions and this
shows that the formal expansions from the reduction step truly need
to be justified rigorously. We introduce new variables $W=X_\MaTime$
and $R=X_\MaLagr$ and rewrite equation \eqref{WE:Red.Theorem1.Eqn1}
in the form
\begin{align}
\label{WE:Red.Theorem1.Eqn2} \MaTimeDer{R}-\MaLagrDer{W}=0,
\quad\MaTimeDer{W}-\MaLagrDer{\NNPot^\prime\at{R}}=0.
\end{align}
This is a first order system of macroscopic conservation laws with
characteristic speeds
$\la_\pm=\pm\sqrt{\NNPot^{\prime\prime}\at{R}}$ and is called the
$p$-system (with $p=-\NNPot^\prime$), see \cite{Daf00}. These
equations formally imply the conservation of energy, i.e., any
\emph{smooth} solution to \eqref{WE:Red.Theorem1.Eqn2} satisfies
$\MaTimeDer{E}-\MaLagrDer{}\at{W\NNPot^\prime\at{R}}=0$ with
$E=\tfrac{1}{2}W^2+\NNPot\at{R}$.
\par
Now suppose that $\NNPot$ is concave or, more general, restrict $R$
to the region of concavity of $\NNPot$. In this case, the system
\eqref{WE:Red.Theorem1.Eqn2} is elliptic and its initial value
problem is ill-posed. For this reason the microscopic system behaves
as follows: Even if we initialize the chain with data satisfying
$x_\MiLagr\at{0}=\eps^{-1}X_\ini\at{\eps\MiLagr}$ and
$\dot{x}_\MiLagr\at{0}=W_\ini\at{\eps\MiLagr}$, where $X_\ini$ and
$W_\ini$ are infinitely smooth macroscopic functions, the atomic
data will immediately start to oscillate on the microscopic scale,
see \cite{Her05,DH07} for numerical simulations. Therefore, the
two-scale ansatz \emph{cannot} be satisfied for any $\MaTime>0$ and
we conclude that any rigorous justification of
\eqref{WE:Red.Theorem1.Eqn1} must exclude non-convex $\NNPot$.
\par
Next suppose that $\NNPot$ is strictly convex, which provides the
strict hyperbolicity of the $p$-system, and assume for simplicity
that $\NNPot^\prime$ is also strictly convex, so that all
eigenvalues are genuinely nonlinear. However, even in this case
there exist limitations for the validity of
\eqref{WE:Red.Theorem1.Eqn2}. In fact, it is well known that the
nonlinearity of $\NNPot$  causes the following generic situation:
Given smooth initial data for \eqref{WE:Red.Theorem1.Eqn2}, there
exists a critical time $0<\MaTime_0<\infty$ at which the first
macroscopic shock is formed. In particular, there exists a smooth
solution for $0<\MaTime<\MaTime_0$, and for these times we can
expect that \eqref{WE:MSAnsatz} provides an approximate solution of
the microscopic system. However, for $\MaTime>\MaTime_0$ the
macroscopic energy $E$ is not conserved anymore and thus the
$p$-system can not be related to the macroscopic dynamics of the
chain, since the chain conserves the energy exactly. This phenomenon
is usually called the \emph{shock problem} and appears analogously
in all zero dispersion limits, compare for instance the surveys in
\cite{Lax86,Lax91,LLV93}. For the FPU chain the macroscopic dynamics
beyond the shock can be understood by Whitham's modulation theory
with periodic travelling waves, see
\cite{FV99,DHM06,AMSMSP:DHR,DH07} and \cite{HFM81,DM98,El05} for the
complete integrable Toda chain. Moreover, for harmonic lattices the
macroscopic limit under the hyperbolic scaling can be established
rigorously by means of  weak convergence methods (cf.\
\cite{Mie06,Mie07}). The transport of energies can be studied
via  Wigner-Husimi measures, see \cite{Mie06}.%
%
%
\subsection{From FPU to KdV}\label{sec:ChainNew.kdv}
%
%
To derive the KdV equation for FPU chains we rely on the two-scale
ansatz
\begin{align}
\label{kdv:MSAnsatz}%
x\pair{\MiTime}{\MiLagrC}=\eps{X}\pair{\eps^3\MiTime}
{\eps\MiLagrC+\eps{c}\MiTime},
\end{align}
which is related to a moving frame with drift velocity $c$. Example
\ref{ST:kdv.Example} provides
\begin{align}
\label{kdv:MSTrafo.Eqn2}%
\bbT_\vel\pair{\eps}{\MiTime}: %
\pair{X}{X_\MaTime}\at{\MaLagr}\mapsto
\pair{x}{x_\MiTime}\at{\MiLagrC}=\pair{\eps{X}}
{\eps^4{X_\MaTime}+\eps^2c{X}_\MaLagr}\at{\eps\MiLagrC+\eps{c}\MiTime}
\end{align}%
with $P$ and $Q$ as in \S\ref{sec:ChainNew.WE}.
%
%
%
The transformation $\bbT_\vel$ is defined only on
$H^1\at{\Rset;\,\d\MaLagr}{\times}L^2\at{\Rset;\,\d\MaLagr}$, a
dense subset of $TP$, but in order to focus on the basic features of
the reduction procedure we do not stress out this explicitly.
\begin{lemma}
\label{kdv:Lemma1}
Under the exact two-scale transformation \eqref{kdv:MSTrafo.Eqn2}
the energies $\calK$ and $\calV$ transform into their
$\eps$-parametrized counterparts
\begin{align}
\label{kdv:Lemma1.Eqn1} %
\bbK\triple{\eps}{X}{X_\MaTime}=\eps^{3}
\int\mhintlimits_{\Rset}\tfrac{1}{2}\at{\eps^2{X}_\MaTime+
{c}\,X_\MaLagr}^2\d\MaLagr,\quad%
\bbV\pair{\eps}{X}=\eps^{-1}
\int\mhintlimits_{\Rset}\NNPot\at{\eps\nabla^+_\eps{X}}\d\MaLagr
\end{align}
and the matrix $\bsSi$ corresponding to the symplectic form $\bssi$
is given by $\bsSi=\eps^5\bsSi_1+\eps^7\bsSi_2$ with
\begin{align}
\label{kdv:Lemma1.Eqn3} %
\bsSi_1=
\begin{pmatrix}%
-2\,c\,\MaLagrDer{}&0\\0&0
\end{pmatrix}%
,\quad\bsSi_2=
\begin{pmatrix}%
0&-1\\1&0
\end{pmatrix}%
.%
\end{align}
Moreover, due to the time dependence of the two-scale
transformation, $\bbH$ differs from $\bbE$ and satisfies
\begin{math}
\bbH=\bbE+\bbI
\end{math}
with
\begin{align}
\label{kdv:Lemma1.Eqn2} %
\bbI\triple{\eps}{X}{X_\MaTime}= - \eps^3c\int\mhintlimits_{\Rset}
\at{\eps^2\,X_\MaTime+cX_\MaLagr}X_\MaLagr\,\d\MaLagr.
\end{align}
\end{lemma}
\begin{proof}
For the proof of \eqref{kdv:Lemma1.Eqn1} and \eqref{kdv:Lemma1.Eqn2}
we insert the ansatz \eqref{kdv:MSAnsatz} into the definitions of
$\calV$, $\calK$ and $\calI$, cf. Formula \eqref{GS:Chain.energies}
and Example \ref{ST:kdv.Example}, and replace
$\eps^3\MiTime+\eps{c}\MiLagrC$ by $\MaLagr$ in the arising
integrals. Moreover, the identity $\bbH=\bbE+\bbI$ is provided by
Theorems \ref{MF:Lemma2} and \ref{ST:Theo1}. Finally, the linear
two-scale transformation $\bbTvp\pair{\eps}{\MiTime}$ can be
identified with
\begin{math}
\bsS\circ\bsT,
\end{math} %
where $\bsS:TP\rra{TQ}$ is given by
\begin{math}
\at{\bsS\,Z}\at{\MiLagrC}=Z\at{\eps\MiLagrC+\eps{c}\MiTime}
\end{math} %
and $\bsT$ abbreviates the operator-valued matrix
\begin{align*}
\bsT=
\begin{pmatrix}
\eps&0\\\eps^2c\MaLagrDer{}&\eps^4
\end{pmatrix}.
\end{align*}
Due to Lemma \ref{TST:Lemma.SymTrafo} and Remark
\ref{TST:Remark.SymTrafo} we find
\begin{math}
\bsSi=\frac{\d\MaTime}{\d\MiTime}\,
\frac{\d\MiLagrC}{\d\MaLagr}\,\bsT^\ast\,\bsSi_\mtrc\,{\bsT},
\end{math}
which yields \eqref{kdv:Lemma1.Eqn3} after a short computation.
\end{proof}
%
%
%
%
\paragraph{Leading order reduction}
%
At first we expand the various energies with respect to $\eps$ up to
$\DO{\eps^6}$. To this end, we define $v_i=\NNPot^{(i)}\at{0}$
so that the Taylor polynomial of $\NNPot$ reads
$\NNPot\at{x}=\tfrac{v_2}{2}x^2+\tfrac{v_3}{6}x^3+\mathrm{h.o.t.}$
\begin{lemma}
\label{kdv:Lemma2} %
The transformed energies $\bbK$, $\bbV$ and $\bbI$ satisfy
\begin{align*}
\bbK= \eps^3\bbK_0+\eps^5\bbK_1+\DO{\eps^6},\quad \bbV=
\eps^3\bbV_0+\eps^5\bbV_1+\DO{\eps^6},\quad
\bbI=\eps^3\bbI_0+\eps^5\bbI_1+\DO{\eps^6},
\end{align*}
where
\begin{align*}
\bbK_0\at{X}&=\tfrac{c^2}{2}\int\mhintlimits_{\Rset} X_\MaLagr^2\,
\d\MaLagr, \quad \bbK_1\pair{X}{X_\MaTime}=c\int\mhintlimits_{\Rset}
X_\MaTime{}X_\MaLagr\,\d\MaLagr,\quad\\
\bbV_0\at{X}&=\tfrac{v_2}{2}\int\mhintlimits_{\Rset}
X_\MaLagr^2\, \d\MaLagr, \quad%
\bbV_1\at{X}=-\tfrac{v_2}{24}\int\mhintlimits_{\Rset}
{X}_{\MaLagr\MaLagr}^2\,\d\MaLagr+ %
\tfrac{v_3}{6}\int\mhintlimits_{\Rset}
{X}_{\MaLagr}^3\,\d\MaLagr,\\
\bbI_0\at{X}&=-c^2\int\mhintlimits_{\Rset}
X_\MaLagr^2\,\d\MaLagr,\quad
\bbI_1\pair{X}{X_\MaTime}=-c\int\mhintlimits_{\Rset}
X_\MaTime{}X_\MaLagr\,\d\MaLagr.
\end{align*}
\end{lemma}
\begin{proof} The expansions for $\bbK$ and $\bbI$ follow
immediately from Lemma \ref{kdv:Lemma1}. To prove the remaining
assertions we start with
\begin{math}
\eps\nabla^+_\eps{X}=\eps^2{X}_\MaLagr+
\eps^3\tfrac{1}{2}{X}_{\MaLagr\MaLagr}+
\eps^4\tfrac{1}{6}{X}_{\MaLagr\MaLagr\MaLagr} +\DO{\eps^5}
\end{math} %
and obtain
\begin{align*}
\NNPot\at{\eps\nabla^+_\eps{X}}&= \tfrac{v_2}{2}
\at{\eps\nabla^+_\eps{X}}^2 +
\tfrac{v_3}{6}\at{\eps\nabla^+_\eps{X}}^3+\DO{|\eps\nabla^+_\eps{X}|^4}
\\
&= \tfrac{v_2}{2}
\at{\eps^4X_\MaLagr^2+\eps^5{}X_\MaLagr{X}_{\MaLagr\MaLagr}+
\eps^6\tfrac{1}{4}X_{\MaLagr\MaLagr}^2+
\eps^6\tfrac{1}{3}X_\MaLagr{X}_{\MaLagr\MaLagr\MaLagr}}
+\eps^6\tfrac{v_3}{6}X_\MaLagr^3+\DO{\eps^7}.
\end{align*}
We insert this expression into the formula for $\bbV\pair{\eps}{X}$
and due to $\int_\Rset{}X_\MaLagr{X}_{\MaLagr\MaLagr}\d\MaLagr=0$
and $\int_\Rset{}X_\MaLagr{X}_{\MaLagr\MaLagr\MaLagr}\d\MaLagr=
-\int_\Rset{}{X}_{\MaLagr\MaLagr}^2\d\MaLagr$, we obtain the
asserted expansion for $\bbV$.
\end{proof}
In the next step we can read-off the leading order terms of
$\bbL=\bbK-\bbV$, $\bbH=\bbK+\bbV+\bbI$ and $\bsSi$. However, the
order of $\eps$ at which we find the reduced Lagrangian and
Hamiltonian structures depends on the choice of the moving-frame
speed $c$. Let us start with the case $c^2\neq{}v_2$. Under this
assumption the leading order terms correspond to $\eps^3$. More
precisely, we obtain $\bssi^\red=0$ and
\begin{align*} \bbL^\red\at{X}=
-\bbH^\red\at{X}=\bbK_0\at{X}-\bbV_0\at{X}=\tfrac{1}{2}\at{c^2-v_2}
\int\mhintlimits_\Rset{X_\MaLagr^2}\,\d\MaLagr.
\end{align*}
In particular, both the reduced Lagrangian and Hamiltonian equations
turn out to be equivalent to
\begin{math}
X_{\MaLagr\MaLagr}=0
\end{math}
and have no non-trivial solutions at all. Thus, we assume
\begin{align}
\label{kdv:Condition} %
c^2=v_2=\NNPot^{\prime\prime}\at{0},
\end{align}
i.e., the moving-frame speed equals the sound velocity of the
linearized FPU chain. In this case we find
$\bbK_0=\bbV_0=-\tfrac{1}{2}\bbI_0$ and this leads to
\emph{cancelations} in $\bbL$ and $\bbH$. Consequently, the leading
order terms in the Lagrangian and Hamiltonian structure correspond
to $\eps^5$ and we end up with the following macroscopic model:
\begin{theorem}
%
%
With \eqref{kdv:Condition} the reduced Lagrangian and Hamiltonian
structures are given by
\begin{align}
\label{kdv:RedTheorem.Eqn1} \bbL^\red\pair{X}{X_\MaTime}=
\bbK_1\pair{X}{X_\MaTime}-\bbV_1\at{X},\quad \bbH^\red\at{X}=
\bbV_1\at{X},\quad\bsSi^\red=\bsSi_1.
\end{align}
In particular, both the reduced Lagrangian and Hamiltonian equations
are equivalent to
\begin{align}
\label{kdv:Eula2}%
2\,c{X}_{\MaTime\MaLagr}-
\tfrac{1}{12}\,v_2\,X_{\MaLagr\MaLagr\MaLagr\MaLagr}-
v_3X_\MaLagr{X}_{\MaLagr\MaLagr}=0,
\end{align}
which is a KdV equation for $X_\MaLagr$.
\end{theorem}
\begin{proof}
The identities \eqref{kdv:RedTheorem.Eqn1} can be read-off from
Lemma \ref{kdv:Lemma2} and \eqref{kdv:Eula2} follows by a direct
calculation.
\end{proof}
The KdV reduction with \eqref{kdv:Condition} is an example for Case
B from Section \S\ref{sec:Reductions}, i.e., the reduced Hamiltonian
structure lives on $P$ and not on $TP$. Moreover, the term $\bbI_0$,
which produces the cancelation in $\bbL_0$ and $\bbH_0$, equals up
to the sign the macroscopic integral of motion
$\bbI^\red\at{X}=\int\mhintlimits_{\Rset}X_\MaLagr^2\d\MaLagr$,
associated to the invariance under shift in the $y$-direction.
\begin{remark}
As before, the formal two-scale reduction relies on the Hamiltonian
structure on $TP$ but fails if we use the canonical structure on
$T^\ast{P}$. Even worse, here we cannot overcome this problem by a
simple rescaling of $\bbL$. To understand this, we consider the
rescaled Lagrangian (for $c^2=v_2$)
\begin{align*}
\tilde{\bbL}\triple{\eps}{X}{X_\MaTime}&=
\eps^{-5}\bbL\triple{\eps}{X}{X_\MaTime} %
=
\bbK_1\pair{X}{X_\MaTime}+\eps^2\bbK_2\at{X_\MaTime}
+\tilde{\bbV}\pair{\eps}{X},
\end{align*}
where \begin{align*}\bbK_2\at{X_\MaTime}=c\!\int\mhintlimits_{\Rset}
X_\MaTime{}X_\MaLagr\,\d\MaLagr,\quad
\tilde{\bbV}\pair{\eps}{X}=\eps^{-5}
\at{\bbV\pair{\eps}{X}-\eps^3\bbV_0\at{X}}\approx\bbV_1\at{X}=\DO{1}.
\end{align*}
The canonical momenta are given by
$\tilde{\Pi}=c{X_\MaLagr}+\eps^2{X}_\MaTime$ and computing the
associated Hamiltonian on $T^\ast{P}$ we find
\begin{align*}
\tilde{\mhol{H}}\triple{\eps}{X}{\Pi}&=
\eps^{-2}\tilde{\mhol{H}}_0\pair{X}{\Pi}\;+\;\tilde{\bbV}\pair{\eps}{X},\quad
\tilde{\mhol{H}}_0\pair{X}{\Pi}=
\int\mhintlimits_{\Rset}\tfrac{1}{2}\at{\Pi^2-c^2\,X_\MaLagr^2}\,
\d\MaLagr.
\end{align*}
In particular, the canonical equations corresponding to the leading
order Hamiltonian $\tilde{\mhol{\bbH}}_0$ do not equal
\eqref{kdv:Eula2}.
\end{remark}
Finally, we mention that rigorous justification results for the KdV
reduction can be found in \cite{SW00,FP99}. However, these results
do not use the reduced Lagrangian or Hamiltonian structures, but
work on the equation of motion directly.
%
%
%
%
\subsection{From KG to nlS}\label{sec:ChainNew.nls}
%
We start with the two-scale ansatz \eqref{Intro:nls.MSAnsatz} for a
modulated pulse in the KG chain \eqref{Intro:KGChain} with
$\alpha=1$ and aim to show that the complex amplitude $A$ satisfies
the nlS equation. Recall that the plane waves appearing in
\eqref{Intro:nls.MSAnsatz} model a microstructure of harmonic
oscillations, and thus we can regard the nlS equation as a
macroscopic modulation equation.
\par
In contrast to the previous examples, here the two-scale ansatz does
not provide immediately an exact two-scale transformation, but we
can setup the problem as follows: We embed the discrete lattice
$\Zset$ into the cylinder $\Rset{\times}T^1$ and identify each
microscopic configuration with a function
$x\in{Q}=L^2\at{\Rset{\times}T^1;\,\d\MiLagrC\d\Phase}$ depending on
the microscopic continuous space variable $\MiLagrC$ and a periodic
phase variable $\Phase\in{T}^1\cong\ccinterval{0}{2\pi}$. Moreover,
in accordance to the scaling, we choose
$P=L^2\at{\Rset{\times}T^1;\,\d\MaLagr\d\Phase}$ and make the
two-scale ansatz
\begin{align}
\label{nls:MSAnsatz} x\triple{\MiTime}{\MiLagrC}{\Phase}
=\eps{X}\triple{\eps^2\MiTime} {\eps\MiLagrC-\eps{}c\MiTime}
{\Phase+\om\MiTime+\theta\MiLagrC},
\end{align}
which gives rise to the inverse two-scale transformation
\begin{align}
\label{nls:MSTrafo}
\Bat{\bbT_\vel\pair{\eps}{\MiTime}\pair{X}{X_\MaTime}}
\pair{\MiLagrC}{\Phase}=\pair{\eps{X}}
{\eps^3{X_\MaTime}-\eps^2c{X}_\MaLagr+\eps\om{X}_\Phase}
\pair{\eps\MiLagrC-\eps{c}\MiTime}
{\Phase+\om\MiTime+\theta\MiLagrC}.
\end{align}
In this section we show that this transformation implies both the
particular form of the microstructure and the {nlS} equation.
\begin{remark}
\label{nls:Remark:MSTrafo}
From \eqref{nls:MSAnsatz} we read-off the identity
\begin{math}
\bbT_\con\pair{\eps}{\MiTime}= \calT_\con^{-1}
\circ\calM_\con\at{-\MiTime} \circ{\bbS_\con}\at\eps
\end{math}, %
where
\begin{math}
\at{\bbS_\con\at\eps{X}}\pair{\MiLagrC}{\Phase}=
\eps{}X\pair{\eps\MiLagrC}{\Phase}
\end{math} %
denotes the inverse scaling transformation. Moreover,
\begin{math}
\at{\calM_\con\at\MiTime\,x}\pair{\MiLagrC}{\Phase}= x
\pair{\MiLagrC+c\,\MiTime}{\Phase-\om\MiTime}
\end{math}
is a moving frame transformation with associated integral of motion
\begin{math}
\mathcal{I}\pair{x}{x_\MiTime}=
\int_{\Rset{\times}T^1}x_\MiTime\at{c{x}_\MiLagrC-\om{x}_\Phase}
\d\MiLagrC\d\Phase
\end{math} %
and
\begin{math}
\nnat{\calT_\con\,x}
\pair{\MiLagrC}{\Phase}=x\pair{\MiLagrC}{\Phase-\theta\MiLagrC}
\end{math} %
corresponds to a weak symmetry transformation.
\end{remark}
\begin{lemma}
\label{nls:Lemma0}
Under \eqref{nls:MSTrafo} the transformed energies $\bbK$, $\bbV$
and $\bbI$ take the form
\begin{align*}%
\bbK\triple{\eps}{X}{X_\MaTime}&=
\eps^{-1}\int\mhintlimits_{\Rset{\times}T^1}
\tfrac{1}{2}\at{\eps^3{X}_\MaTime-
\eps^2{c}\,X_\MaLagr+\eps\om{X}_\Phase}^2\,\d\MaLagr\d\Phase,
\\%
\bbV\pair{\eps}{X}&= \eps^{-1}\int\mhintlimits_{\Rset{\times}T^1}
\at{-\eps^2\tfrac{1}{2}X\laplace_{\eps,\,\theta}X+
\OSPot\at{\eps{X}}}\,\d\MaLagr\d\Phase,
\\%
\bbI\triple{\eps}{X}{X_\MaTime}&=
\eps^{-1}\int\mhintlimits_{\Rset{\times}T^1} \at{\eps^3{X}_\MaTime-
\eps^2{c}\,X_\MaLagr+\eps\om{X}_\Phase} \at{
\eps^2{c}\,X_\MaLagr-\eps\om{X}_\Phase}\,\d\MaLagr\d\Phase,
\end{align*}
and we have $\bbL=\bbK-\bbV$, $\bbE=\bbK+\bbV$ and $\bbH=\bbE+\bbI$.
Moreover, the matrix $\bsSi$ corresponding to the symplectic form
$\bssi$ satisfies
\begin{math}
\bsSi=\eps^3\bsSi_2+\eps^4\bsSi_3+ \eps^5\bsSi_4,
\end{math}
with
\begin{align*}
\bsSi_2=
\begin{pmatrix}%
-2\,\om\,\PhaseDer{}&0\\0&0
\end{pmatrix}%
,\quad \bsSi_3=
\begin{pmatrix}%
2\,c\,\MaLagrDer{}&0\\0&0
\end{pmatrix}%
,\quad \bsSi_4=
\begin{pmatrix}%
0&-1\\1&0
\end{pmatrix},%
\end{align*}
and is non-degenerate due to $\bsSi_4$.
\end{lemma}
\begin{proof}
At first we study the time dependent transformation
$\calM_\con\at{\MiTime}\circ\calT_\con:Q\to{Q}$, see Remark
\ref{nls:Remark:MSTrafo}, and write $\hat{x}=\calT_\con{x}$ and
$\tilde{x}=\calM_\con\at\MiTime \hat{x}$. According to Theorem
\ref{LT:Theo}, the transformation $\calT_\con$ transforms
$\npair{{\calH}}{\si}$ into $\npair{\hat{\calH}}{\hat\si}$ with
$\hat{\calH}=\calH\circ\calT_\vel^{-1}$ and
$\hat{\sigma}=\nnat{\calT_\vel^{-1}}^\ast\sigma$. Then we apply the
moving frame transformation $\calM_\vel\at\MiTime$ and Theorem
\ref{MF:Lemma2} provides the Hamiltonian structure
$\npair{\tilde{\calH}}{\tilde{\sigma}}$ with
$\tilde{\sigma}=\nnat{\calM_\vel\at{-\MiTime}}^\ast\hat\sigma$ and
\begin{math}%
\tilde{\calH}
=\nnat{\hat{\calE}+\hat{\calI}}\circ\calM_\vel\at{-\MiTime}
\end{math}
with $\hat{\calE}=\hat{\calH}$ and
\begin{align}
\label{nls:Lemma0.Eqn1} \hat{\calI}\pair{\hat{x}}{\hat{x}_\MiTime}=
\int\mhintlimits_{\Rset{\times}T^1}\hat{x}_\MiTime\at{c{\hat{x}}_\MiLagrC-\om{\hat{x}}_\Phase}
\d\MiLagrC\d\Phase.
\end{align}
Moreover, exploiting Theorem \ref{ST:Theo1} for the scaling
transformation $\calS_\con\at\eps=\bbS_\con^{-1}\at\eps$, we find
\begin{align*}
\bbH=\nnat{\calE+\calI}\circ\calT_\vel^{-1}\circ\calM_\vel\at{-\MiTime}
\circ{\bbS_\vel}\at\eps,\quad \bssi=\eps^2
\nnat{\calT_\vel^{-1}\circ\calM_\vel\at{-\MiTime}
\circ{\bbS_\vel}\at\eps}^\ast\sigma
\end{align*}
with $\calI=\hat{\calI}\circ\calT_\vel$. According to
\eqref{GS:Chain.energies} and \eqref{nls:Lemma0.Eqn1}, the
microscopic energies are given by
\begin{math}
\calK\at{x_\MiTime}=\tfrac{1}{2}\int_{\Rset\times{T^1}}
x_\MiTime^2\,\d\MiLagrC\d\Phase
\end{math} %
and
\begin{align*}
\calV\at{x}=
\!\!\!\int\mhintlimits_{\Rset\times{T^1}}\!\!\!%
\Bat{- \tfrac{1}{2}x\laplace_{1,\,0}x +
\OSPot\at{x}}\,\d\MiLagrC\d\Phase,\quad \calI\pair{x}{x_\MiTime}=
\int\mhintlimits_{\Rset{\times}T^1}x_\MiTime\at{c{x}_\MiLagrC-c\theta{x}_\Phase-\om{x}_\Phase}
\d\MiLagrC\d\Phase,
\end{align*}
where the discrete operators $\nabla$ and $\laplace$ are defined in
Remark \ref{GS:DiscreteNablas}. The expressions for $\bbK$, $\bbV$,
$\bbL$, $\bbE$ and $\bbI$ now follow by inserting
\eqref{nls:MSAnsatz} into the formulas for $\calK$, $\calV$ and
$\calI$. For the computation of $\bsSi$ we identify the linear
two-scale transformation $\bbTvp\pair{\eps}{\MiTime}$ with
\begin{math}
\bsS\circ\bsT,
\end{math} %
where $\bsS:TP\rra{TQ}$ is given by
\begin{math}\at{\bsS{Z}}\pair{\MiLagrC}{\Phase}=
Z\pair{\eps\MiLagrC-\eps{c}\MiTime}{\Phase+\om\MiTime+\theta\MiLagrC}
\end{math}
and $\bsT$ abbreviates the operator-valued matrix
\begin{align*}
\bsT=
\begin{pmatrix}
\eps&0\\-\eps^2\,c\,\MaLagrDer{}+\eps\,\om\,\PhaseDer{}&\eps^3
\end{pmatrix}
\end{align*}
with components in $\mathrm{Lin}\at{P,\,P}$. Finally, Remark
\ref{TST:Remark.SymTrafo} yields
\begin{math} \bsSi=\frac{\d\MaTime}{\d\MiTime}\,
\frac{\d\MiLagrC}{\d\MaLagr}\,\bsT^\prime\,\bsSi_\mtrc\,\bsT
\end{math}
and this implies the desired result.
\end{proof}
%
%
\paragraph{Leading order reduction}
%
%
%
Next we derive the formal expansions with respect to $\eps$. To this
end we introduce the constants $v_i=\OSPot^{(i)}\at{0}$ and find,
due to $v_0=v_1=0$,
\begin{align}
\label{nls:expansion}\OSPot\at{\eps{X}}=
\eps^2\tfrac{v_2}{2}\,{X}^2+
\eps^3\tfrac{v_3}{6}\,{X}^3+ %
\eps^4\tfrac{v_4}{24}\,{X}^4+ \DO{\eps^5}.
\end{align}
\begin{lemma}
\label{nls:Lemma1}
The transformed energies satisfy
\begin{enumerate}
\item%
\begin{math}
\bbI= \eps\bbI_0+\eps^2\bbI_1+ \eps^3\bbI_2+\DO{\eps^4}
\end{math}
with %
\begin{math}
\bbI_0\at{X}= -\om^2\!\!\!\int\mhintlimits_{\Rset{\times}T^1}\!\!\!
X_\Phase^2\, \d\MaLagr\d\Phase,
\end{math}
\begin{align*}
\bbI_1\at{X}=2\om{}\,c\!\!\!
\int\mhintlimits_{\Rset{\times}T^1}\!\!\!
X_\MaLagr{X}_\Phase\,\d\MaLagr\d\Phase,\quad%
\bbI_2\pair{X}{X_\MaTime}&=
-c^2\!\!\!\int\mhintlimits_{\Rset{\times}T^1}\!\!\!
X_\MaLagr^2\,\d\MaLagr\d\Phase%
-\om\!\!\!\int\mhintlimits_{\Rset{\times}T^1}\!\!\!
X_\MaTime{}X_\Phase
\,\d\MaLagr\d\Phase,%
\end{align*}
\item%
\begin{math}
\bbK= \eps\bbK_0+\eps^2\bbK_1+ \eps^3\bbK_2+\DO{\eps^4}
\end{math}
with $\bbK_0\at{X}=-\tfrac{1}{2}\,\bbI_0\at{X}$,
$\bbK_1\at{X}=-\tfrac{1}{2}\,\bbI_1\at{X}$ and
\begin{align*}
\bbK_2\pair{X}{X_\MaTime}=-\tfrac{1}{2}\,\bbI_2\at{X}
+\tfrac{1}{2}\,\om\!\!\! \int\mhintlimits_{\Rset{\times}T^1}\!\!\!
X_\MaTime{X}_\Phase\,\d\MaLagr\d\Phase,
\end{align*}
\item%
\begin{math}
\bbV= \eps\bbV_0+\eps^2\bbV_1+ \eps^3\bbV_2+\DO{\eps^4}
\end{math}
with
\begin{align}
\label{nls:Lemma1.Pot}
\begin{array}{lclcl}
\bbV_0\at{X}&=&\displaystyle -\tfrac{1}{2}\!\!\!
\int\mhintlimits_{\Rset{\times}T^1}\!\!\! X\laplace_{0,\,\theta}X\,
\d\MaLagr\d\Phase, &+&\displaystyle \tfrac{v_2}{2}\!\!\!
\int\mhintlimits_{\Rset{\times}T^1}\!\!\! X^2\, \d\MaLagr\d\Phase,
\\%
\bbV_1\at{X}&=&\displaystyle %
-\tfrac{1}{2}\!\!\!
\int\mhintlimits_{\Rset{\times}T^1}\!\!\!%
X\at{\nabla^{+}_{0,\theta}X_\MaLagr+
\nabla^-_{0,\theta}X_\MaLagr}\,\d\MaLagr\d\Phase%
&+&\displaystyle%
\tfrac{v_3}{6}\!\!\! \int\mhintlimits_{\Rset{\times}T^1}\!\!\!
X^3\,\d\MaLagr\d\Phase,
\\%
\bbV_2\at{X}&=&\displaystyle -\tfrac{1}{4}\!\!\!
\int\mhintlimits_{\Rset{\times}T^1}\!\!\!
X\at{\nabla^{+}_{0,\theta}-
\nabla^-_{0,\theta}+2\,\mathrm{Id}}X_{\MaLagr\MaLagr}\,
\d\MaLagr\d\Phase &+&\displaystyle \tfrac{v_4}{24}
\!\!\!\int\mhintlimits_{\Rset{\times}T^1}\!\!\!
X^4\,\d\MaLagr\d\Phase.
\end{array}
\end{align}
\end{enumerate}
\end{lemma}
\begin{proof}
The expansions for $\bbK$ and $\bbI$ follow directly from Lemma
\ref{nls:Lemma0}. Moreover, Taylor expansion with respect to $\eps$
yields
%
%
\begin{align}
\notag
\laplace_{\eps,\,\theta}X=\laplace_{0,\,\theta}X+
\eps\at{\nabla^{+}_{0,\theta}+\nabla^-_{0,\theta}}X_\MaLagr
+\eps^2\,\tfrac{1}{2}\at{\nabla^{+}_{0,\theta}-
\nabla^-_{0,\theta}+2\mathrm{Id}}X_{\MaLagr\MaLagr}+\DO{\eps^3},
\end{align}
and this gives rise to the first kind of integrals in
\eqref{nls:Lemma1.Pot}. Finally, inserting \eqref{nls:expansion}
into
$\int_{\Rset{\times}{T^1}}\OSPot\at{\eps{X}}\,\d\MiLagrC\d\Phase$
completes the proof.
\end{proof}
\bigpar%
According to Lemma \ref{nls:Lemma1} the leading order Lagrangian and
Hamiltonian equations are given by
\begin{align}
\label{nls:ZeroEuLa}%
-\om^2X_{\Phase\Phase}+\laplace_{0,\,\theta}X-v_2X=0,
\end{align}
and, using Fourier transform with respect to $\Phase$, we conclude
that this equation has nontrivial solutions if and only if $\om$ and
$\theta$ satisfy $m^2\om^2=\Om^2\at{m\theta}$ for some integer $m$,
where $\Om$ is the dispersion relation for the linearized chain with
$\alpha=1$, that is
\begin{align}
\label{nls:DispRel}%
\Om^2\at{\theta}=v_2+2\at{1-\cos\theta}.
\end{align}
In what follows we always assume $\om^2=\Omega\at\theta^2$ as well
as the \emph{non-resonance} condition
\begin{align}
\label{nls:NonResonance}%
m^2 \omega^2 \neq \Omega(m\theta)^2 \text{ for } m\in
\Zset\setminus\{-1,1\},
\end{align}
which imply that the solution space to \eqref{nls:ZeroEuLa} in
$L^2\at{T^1;\d\Phase}$ is spanned by $\cos\Phase$ and $\sin\Phase$.

\begin{theorem}
\label{nls:Red.Theorem1} Suppose $\om^2=\Omega^2\at{\theta}$ and
\eqref{nls:NonResonance}. Then, the leading order Lagrangian and
Hamiltonian equations \eqref{nls:ZeroEuLa} are quasi-stationary
(corresponding to $\bsSi_0=0$) and have solutions
\begin{align}
\label{nls:ZeroRedEqn}%
X_0\pair{\MaLagr}{\Phase}=
\pi^{-1/2}\Bat{B_1\at{\MaLagr}\cos\at{\Phase}+
B_2\at{\MaLagr}\sin\at{\Phase}}
\end{align}
with $B_1,\,B_2\in{}L^2\at{\Rset;\d\MaLagr}$ arbitrary. Moreover,
\eqref{nls:ZeroRedEqn} implies $0=\bbL_0\at{X_0}=\bbH_0\at{X_0}$.
\end{theorem}
\begin{proof}
All results follow from Lemma \ref{nls:Lemma1}.
%
%
%
\end{proof}
Introducing a complex valued amplitude $A$ by
$2\sqrt{\pi}A=B_1-\iu{B}_2$, Equation \eqref{nls:ZeroRedEqn}
transforms into
\begin{align}
\label{nls:ModAnsatz}%
X_0\pair{\MaLagr}{\Phase}=
2\Re\at{A\at{\MaLagr}\mhexp{\iu\Phase}},
\end{align}
and is hence equivalent the original two-scale ansatz
\eqref{Intro:nls.MSAnsatz}.
%
\paragraph{Elimination of the microstructure}

The leading order reduction determines the structure of the
microscopic oscillations together with the dispersion relation. As
discussed in Case C2 of \S\ref{sec:Reductions}, this allows for a
further reduction step that yields the macroscopic modulation
equation for the amplitudes $B_1$ and $B_2$, or, equivalently, for
the complex-valued amplitude $A$. Let $P_0$ be the $L^2$--space of
complex-valued functions depending on $\MaLagr$, i.e.,
\begin{align*}
P_0=\left\{\pair{B_1}{B_2}\in{L^2}\at{\Rset;\d\MaLagr}{\times}{L^2}\at{\Rset;\d\MaLagr}\right\}\cong
\left\{A\in{L^2}\at{\Rset;\,\Cset}\right\},
\end{align*}
which can be viewed as a closed and proper subset of $P$ due to
\eqref{nls:ZeroRedEqn}. By construction, each element of $P_0$
satisfies the leading order equations exactly, and thus we can use
the next-leading order terms in order to derive the effective
macroscopic dynamics on $TP_0$.
\par%
Below we choose the moving frame speed $c$ appropriately, and this
yields $\bbL_1|_{P_0}\equiv0$ as well as $\bbH_1|_{P_0}\equiv0$ due
to cancelations. Consequently, the reduced structures are related to
$\bbL_2$, and hence we must take care of the correction terms
$\widehat{\bbL}_2$ and $\widehat{\bbH}_2$  coming from the ansatz
$X=X_0+\eps{X_1}$, see Case C2 in  \S\ref{sec:Reductions}.

\begin{lemma}
\label{nls:Lemma2}
With $X=X_0+\eps{X}_1$ we have
\begin{align*}
\widehat{\bbL}_2\pair{X_0}{X_1}=\bbL_0\at{X_1}+\!\!\!
\int\mhintlimits_{\Rset{\times}T^1}\!\!\!
X_1\at{%
2\,\om{}\,c{X_0}_{\MaLagr\Phase}+
\at{\nabla^{+}_{0,\theta}{X_0}_\MaLagr+\nabla^-_{0,\theta}{X_0}_\MaLagr}%
- \tfrac{v_3}{2} X_0^2}\,\d\MaLagr\d\Phase
\end{align*}
and
$\widehat{\bbH}_2\pair{X_0}{X_1}=-\widehat{\bbL}_2\pair{X_0}{X_1}$.
Moreover, all corrections to the symplectic structure are of order
$\eps^3$ and do not contribute to $\bsSi_2$.
\end{lemma}
\begin{proof}
The correction terms for $\bbI$, $\bbK$, and $\bbV$ can be read-off
from Lemma \ref{nls:Lemma1}. More precisely, we find
\begin{enumerate}
\item%
\begin{math}
{\bbI}_i\at{X_0+\eps\,X_1}={\bbI}_i\at{X_0}+\widehat{\bbI}_i\pair{X_0}{X_1}
\end{math} %
with $\widehat{\bbI}_0\pair{X_0}{X_1}=0$ and
\begin{align*}
\widehat{\bbI}_1\pair{X_0}{X_1}&=%
2\om^2\!\!\!\int\mhintlimits_{\Rset{\times}T^1}X_1{X_0}_{\Phase\Phase}\,\d\MaLagr\d\Phase
,\\%
\widehat{\bbI}_2\pair{X_0}{X_1}&=%
\bbI_0\at{X_1}%
-4\om{}\,c\!\!\! \int\mhintlimits_{\Rset{\times}T^1}\!\!\!
X_1{X_0}_{\MaLagr\Phase}\,\d\MaLagr\d\Phase, %
\end{align*}
\item%
\begin{math}
{\bbK}_i\at{X_0+\eps\,X_1}={\bbK}_i\at{X_0}+\widehat{\bbK}_i\pair{X_0}{X_1}
\end{math}
with $\widehat{\bbK}_0\pair{X_0}{X_1}=0$ and
\begin{align*}
\widehat{\bbK}_1\pair{X_0}{X_1}= %
-\tfrac{1}{2}\widehat{\bbI}_1\pair{X_0}{X_1}%
,\quad%
\widehat{\bbK}_2\pair{X_0}{X_1}=-\tfrac{1}{2}\widehat{\bbI}_2\pair{X_0}{X_1},
\end{align*}
\item%
\begin{math}
{\bbV}_i\at{X_0+\eps\,X_1}={\bbV}_i\at{X_0}+\widehat{\bbV}_i\pair{X_0}{X_1}
\end{math}
with $\widehat{\bbV}_0\pair{X_0}{X_1}=0$ and
\begin{align*}
\widehat{\bbV}_1\pair{X_0}{X_1}&= %
-\!\!\! \int\mhintlimits_{\Rset{\times}T^1}\!\!\!
X_1\laplace_{0,\,\theta}X_0\, \d\MaLagr\d\Phase%
+v_2\!\!\! \int\mhintlimits_{\Rset{\times}T^1}\!\!\! X_1{X}_0\,
\d\MaLagr\d\Phase,
\\%
\widehat{\bbV}_2\pair{X_0}{X_1}&=\bbV_0\at{X_1}%
-\!\!\!
\int\mhintlimits_{\Rset{\times}T^1}\!\!\!%
X_1\at{\nabla^{+}_{0,\theta}{X_0}_\MaLagr+
\nabla^-_{0,\theta}{X_0}_\MaLagr}\,\d\MaLagr\d\Phase%
+\displaystyle%
\tfrac{v_3}{2}\!\!\! \int\mhintlimits_{\Rset{\times}T^1}\!\!\!
{X_1}X_0^2\,\d\MaLagr\d\Phase.
\end{align*}
\end{enumerate}
Finally, due to $\bbL=\bbK-\bbV$ and $\bbH=\bbK+\bbV+\bbI$ all
assertions are direct consequences of these identities.
\end{proof}
\begin{lemma}
\label{nls:Lemma3}
If the moving frame speed $c$ is given by
\begin{align}
\label{nls:Lemma3.Eqn1}%
c\,\om=-\Om^\prime\at{\theta}\Om\at\theta=-\sin\theta
\end{align}
then ${\bbL}_1|_{P_0}=-\bbH_1|_{P_0}=0$. Otherwise the Lagrangian
and Hamiltonian equations to ${\bbL}_1$ and ${\bbH}_1$ have no
non-trivial solution at all.
\end{lemma}
\begin{proof}
A direct calculation shows

\begin{align*}
\bbI_1\at{X_0}=-2\,\bbK_1\at{X_0}=2\,\om\,c\int\mhintlimits_\Rset
\at{{B_1}_\MaLagr{B_2}-{B_1}{B_2}_\MaLagr}\, \d\MaLagr =4\,\om\,c
\int\mhintlimits_\Rset{B_1}_\MaLagr{}B_2\,\d\MaLagr,
\end{align*}
and using
\begin{align*}
X\at{\nabla^{+}_{0,\theta}
X_\MaLagr+\nabla^{-}_{0,\theta}X_\MaLagr}=%
2\sin\theta\at{B_1\cos \Phase+B_2\sin\Phase} \at{-{B_1}_\MaLagr\sin
\Phase+{B_2}_\MaLagr\cos\Phase}
\end{align*}
as well as
\begin{math}
0=\int\mhintlimits_{T^1}
\at{B_1\cos{\Phase}+B_2\sin{\Phase}}^3\d\Phase
\end{math}
we find
\begin{align*}
{\bbV}_1\at{X_0}=-\sin\theta \int\mhintlimits_\Rset
\at{{B_1}{B_2}_\MaLagr-{B_1}_\MaLagr{B_2}}\,\d\MaLagr=
2\sin\theta\int\mhintlimits_\Rset {B_1}_\MaLagr{B_2}\,\d\MaLagr,
\end{align*}
so that ${\bbL}_1|_{P_0}=-\bbH_1|_{P_0}=0$ for
\eqref{nls:Lemma3.Eqn1}. Finally, for other values of $c$ the
Lagrangian equations for $\bbL_1|_{P_0}$ equal
${B_1}_\MaLagr={B_2}_\MaLagr=0$.
\end{proof}
Condition \eqref{nls:Lemma3.Eqn1} implies that the moving frame
moves with the negative \emph{group velocity} associated to the
dispersion relation \eqref{nls:DispRel} (the negative sign appears
since our phase definition is $\Phase=\om\MiTime+\theta\MiLagrC$).
Compare this with the case $c^2\neq{\NNPot^{\prime\prime}}\at{0}$
from \S\ref{sec:ChainNew.kdv}.
\bigpar%
Next we prove that the non-resonance condition
\eqref{nls:NonResonance} provides the higher order correction $X_1$
in dependence of the first order solution $X_0$.
\newcommand{\wt}{\widetilde}
\newcommand{\wh}{\widehat}
\begin{lemma}
\label{nls:Lemma4} %
Suppose $\om^2=\Om^2\at{\theta}$,
$c\,\om=-\Om^\prime\at{\theta}\Om\at{\theta}$, and the non-resonance
condition \eqref{nls:NonResonance}, and let $X_0$ be fixed. Then,
each solution $X_1$ to the equation
\begin{align*}%
\VarDer{\wh{\bbL}_2\pair{X_0}{X_1}}{X_1}=0
\end{align*}
satisfies $X_1-\wh{X}_1\at{X_0}\in{P_0}$, where the special solution
$\wh{X}_1\at{X_0}$ is given in the proof. Moreover, for each
$\wt{X}_0\in{P_0}$ we have
\begin{align*}
\widehat{\bbL}_2\pair{X_0}{\wh X_1\at{X_0}+\wt{X}_0}=-
\widehat{\bbH}_2\pair{X_0}{\wh X_1\at{X_0}+\wt{X}_0}=
-\overline{\bbV}_2\at{X_0}
\end{align*}
with
\begin{align*} \overline{\bbV}_2\at{X_0}=
C\!\!\!\int\mhintlimits_{\Rset{\times}T^1}\!\!\!
\at{B_1^2+B_2^2}^2\,\d\MaLagr\d\Phase,\quad\quad
C=\frac{v_3^2}{8\pi}\at{\frac{1}{4\at{4\om^2-\Om^2\at{2\theta}}}-\frac{1}{2v_2}}.
\end{align*}
\end{lemma}
\begin{proof}
The choice of $c$ implies
\begin{align*}
\widehat{\bbL}_2\pair{X_0}{X_1}=\bbL_0\at{X_1}-\tfrac{v_3}{2} \!\!\!
\int\mhintlimits_{\Rset{\times}T^1}\!\!\!
X_1{}X_0^2\,\d\MaLagr\d\Phase,
\end{align*}
and hence the equation for $X_1$ becomes
\begin{align}
\label{nls:Lemma4.Eqn1} %
\om^2{X_1}_{\Phase\Phase}-\laplace_{0,\,\theta}{X_1}+v_2{X_1}=
-\tfrac{v_3}{2}X_0^2=-\tfrac{v_3}{2}\at{\tfrac{B_1^2+B_2^2}{2\pi}+\tfrac{B_1^2-B_2^2}{2\pi}\cos{2\Phase}+\tfrac{B_1B_2}{\pi}\sin{2\Phase}}.
\end{align}
This equation can be solved explicitly by Fourier transform with
respect to $\Phase$ and noting that the operator
$\om^2\partial_{\Phase\Phase}-\laplace_{0,\,\theta}+v_2$ is
symmetric with kernel orthogonal to $X_0^2$. Some elementary
analysis shows that each solution \eqref{nls:Lemma4.Eqn1} can be
written as $X_1=\wh{X}_1\at{X_0}+\wt{X}_0$, where $\wt{X}_0\in{P_0}$
and
\begin{align*}
\wh{X}_1\at{X_0} =C_1X_0^2+C_2,\quad%
C_1=\frac{v_3}{2\at{4\om^2-\Om^2\at{2\theta}}}, \quad%
C_2=-\at{C_1+\frac{v_3}{2v_2}}\frac{B_1^2+B_2^2}{2\pi}.%
\end{align*}
Multiplying \eqref{nls:Lemma4.Eqn1} by $X_1$ and integrating over
$\Rset\times T^1$ gives
\begin{align*}
2\,\bbL_0\at{\wh{X}_1\at{X_0}+\wt{X}_0}=\tfrac{v_3}{2} \!\!\!
\int\mhintlimits_{\Rset{\times}T^1}\!\!\!
\at{\wh{X}_1\at{X_0}+\wt{X}_0}X_0^2\,\d\MaLagr\d\Phase=
\tfrac{v_3}{2} \!\!\! \int\mhintlimits_{\Rset{\times}T^1}\!\!\!
\wh{X}_1\at{X_0}X_0^2\,\d\MaLagr\d\Phase
\end{align*}
for all $\wt{X}_0\in{P_0}$, and hence we find
\begin{align*}
\widehat{\bbL}_2\pair{X_0}{\wh{X}_1\at{X_0}+\wt{X}_0}=
\widehat{\bbL}_2\pair{X_0}{\wh{X}_1\at{X_0}}=-\tfrac{v_3}{4}\int\mhintlimits_{\Rset{\times}T^1}\!\!\!
\at{C_1X_0^2+C_2}X_0^2\,\d\MaLagr\d\Phase
\end{align*}
which implies the desired result.
\end{proof}
Finally, we combine all results and obtain the macroscopic model on
$TP_0$.

\begin{theorem}
\label{nls:Red.Theorem2}
Under the assumptions made in  Lemma \ref{nls:Lemma4} the reduced
Lagrangian and Hamiltonian are given by
\begin{align}
\notag
{\bbL}^\red\pair{A}{A_\MaTime}=
{\bbK}^\red\pair{A}{A_\MaTime}-{\bbV}^\red\at{A},\quad
{\bbH}^\red\pair{A}{A_\MaTime}={\bbV}^\red\at{A},
\end{align}
with
\begin{align}
\notag
\begin{split}
{\bbK}^\red\pair{A}{A_\MaTime}&= \iu2\pi\om\int\mhintlimits_\Rset
\at{A\mhol{A}_\MaTime-A_\MaTime\mhol{A}}\,\d\MaLagr,
\\%
{\bbV}^\red\at{A}&= 2\pi\rho_1 \int\mhintlimits_\Rset
\abs{A_\MaLagr}^2\,\d\MaLagr +
2\pi\rho_2\int\mhintlimits_\Rset\abs{A}^4\,\d\MaLagr,
\end{split}
\end{align}
where the constants $\rho_1$ and $\rho_2$ are given in
\eqref{nls:Red.Theorem2.Eqn8}, and
$A=\at{B_1-\iu{}B_2}/\at{2\sqrt{\pi}}$ is the complex-valued
amplitude. Moreover, in terms of $\pair{A}{A_\MaTime}$ the reduced
symplectic matrix ${\bsSi}^\red$ is given by
\begin{align}
\label{nls:Red.Theorem2.Eqn4}%
{\bsSi}^\red=\,4\,\pi\om\iu
\begin{pmatrix}1&0\\0&0
\end{pmatrix},
\end{align}
and both the reduced Lagrangian and Hamiltonian equations are
equivalent to
\begin{align}
\label{nls:Red.Theorem2.Eqn3}
\iu2\,\om{}A_\MaTime=\rho_1A_{\MaLagr\MaLagr}- 2\rho_2\abs{A}^2A,
\end{align}
which is a nonlinear Schr\"{o}dinger equation.
\end{theorem}
\begin{proof}
Using the results from Lemma \ref{nls:Lemma1} and Lemma
\ref{nls:Lemma2} we end up with
\begin{align*}
\widetilde{\bbL}\pair{X_0}{X_1}&=\eps^3\nnat{{\bbK}_2\at{X_0}-{\bbV}_2\at{X_0}+
\widehat{\bbL}_2\pair{X_0}{X_1}}+\DO{\eps^4},
\\%
\widetilde{\bbH}\pair{X_0}{X_1}&=\eps^3\nnat{{\bbK}_2\at{X_0}+
{\bbV}_2\at{X_0}+{\bbI}_2\at{X_0}+
\widehat{\bbH}_2\pair{X_0}{X_1}}+\DO{\eps^4},
\end{align*} %
where we have used that ${\bbL}_i\at{X_0}={\bbH}_i\at{X_0}=0$ for
$i=0,1$ and $X_0\in{P_0}$, compare Theorem \ref{nls:Red.Theorem1}
and Lemma \ref{nls:Lemma3}. Moreover, due to Lemma \ref{nls:Lemma4}
we can eliminate $X_1$, and this yields
\begin{align}
\label{nls:Red.Theorem2.Eqn5}
\begin{split}
{\bbL}^\red\at{X_0}&=\eps^3\nnat{{\bbK}_2\at{X_0}-{\bbV}_2\at{X_0}-\overline{\bbV}_2\at{X_0}},
\\%
{\bbH}^\red\at{X_0}&=\eps^3\nnat{{\bbK}_2\at{X_0}+
{\bbV}_2\at{X_0}+{\bbI}_2\at{X_0}+ \overline{\bbV}_2\at{X_0}},
\end{split}
\end{align} %
compare Case C2 in \S\ref{sec:Reductions}. Inserting
\eqref{nls:ZeroRedEqn} into the formulas from Lemma \ref{nls:Lemma1}
gives
\begin{align*}
{\bbI}_2\at{X_0}&=-c^2 \int\mhintlimits_\Rset
\at{\at{{B_1}_\MaLagr}^2+\at{{B_2}_\MaLagr}^2}\,\d\MaLagr-
\om\int\mhintlimits_\Rset
\at{{B_1}_\MaTime{B_2}-{B_1}{B_2}_\MaTime}\,\d\MaLagr,
\\%
{\bbK}_2\at{X_0}&=\tfrac{1}{2}c^2 \int\mhintlimits_\Rset
\at{\at{{B_1}_\MaLagr}^2+\at{{B_2}_\MaLagr}^2}\,\d\MaLagr+
\om\int\mhintlimits_\Rset
\at{{B_1}_\MaTime{B_2}-{B_1}{B_2}_\MaTime}\,\d\MaLagr,
\\%
{\bbV}_2\at{X_0}&= \tfrac{1}{2}\,\cos{\theta} \int\mhintlimits_\Rset
\at{\at{{B_1}_\MaLagr}^2+\at{{B_2}_\MaLagr}^2}\,\d\MaLagr+
\tfrac{v_4}{32\pi}
\int\mhintlimits_\Rset\at{B_1^2+B_2^2}^2\,\d\MaLagr.
\end{align*}
By construction, $B_1$ and $B_2$ satisfy
$B_1=\sqrt{\pi}\at{A+\mhol{A}}$, $B_2=\iu\sqrt{\pi}\at{A-\mhol{A}}$,
and thus we find
\begin{math}
{B_1}_\MaTime{B_2}-{B_1}{B_2}_\MaTime=
\iu2\pi\,\at{A\mhol{A}_\MaTime-A_\MaTime\mhol{A}}
\end{math}
as well as
\begin{align*}
\at{{B_1}_\MaLagr}^2+\at{{B_2}_\MaLagr}^2=
4\pi\,A_\MaLagr\,\mhol{A}_\MaLagr=4\pi\abs{A_\MaLagr}^2,\quad
\at{{B_1}^2+{B_2}^2}^2=16\pi^2\,\at{A\,\mhol{A}}^2=16\pi^2\abs{A}^4.
\end{align*}
We define
\begin{align}
\label{nls:Red.Theorem2.Eqn8}
\rho_1:=\Om\at\theta\Om^{\prime\prime}\at\theta=\cos\theta-c^2,\quad
\rho_2:=\frac{v_4}{4}-\frac{v_3^2}{2v_2}+\frac{v_3^2}{4\at{4\om^2-\Om^2\at{2\theta}}}
\end{align}
and
\begin{align*}
\bbK^\red&:={\om}\int\mhintlimits_\Rset
\at{{B_1}_\MaTime{B_2}-{B_1}{B_2}_\MaTime}\,\d\MaLagr,
\\%
\bbV^\red&:=\tfrac{\rho_1}{2} \int\mhintlimits_\Rset
\at{\at{{B_1}_\MaLagr}^2+\at{{B_2}_\MaLagr}^2}\,\d\MaLagr+
\tfrac{\rho_2}{8\pi}\int\mhintlimits_\Rset\at{B_1^2+B_2^2}^2\,\d\MaLagr,
\end{align*}
and this implies the formulas for $\bbL^\red$ and $\bbH^\red$. To
compute ${\bsSi}^\red$, recall $\bsSi=\eps^3\bsSi_2+\DO{\eps^4}$
independent of $X_1$, and notice that the ansatz
\eqref{nls:ZeroRedEqn} can be written as
\begin{align*}
\begin{pmatrix}
X\\X_\MaTime
\end{pmatrix}
={\bsT}_0
\begin{pmatrix}
B_1\\B_2\\{B_1}_\MaTime\\{B_2}_\MaTime
\end{pmatrix},\quad
{\bsT}_0=\frac{1}{\sqrt{\pi}}\begin{pmatrix}
\cos\Phase&\sin\Phase&0&0\\0&0&\cos\Phase&\sin\Phase
\end{pmatrix}
\end{align*}
with ${\bsT}_0:T{P}_0\to{TP}$. The adjoint operator
${\bsT}_0^\prime:{TP}\to TP_0$ reads
\begin{align*}
{\bsT}_0^\prime\begin{pmatrix} X\\X_\MaTime
\end{pmatrix}=
\frac{1}{\sqrt\pi}\begin{pmatrix} %
\displaystyle\int_{T^1}X\cos\Phase\,\d\Phase,&%
\displaystyle\int_{T^1}X\sin\Phase\,\d\Phase,&%
\displaystyle\int_{T^1}X_\MaTime\cos\Phase\,\d\Phase,&%
\displaystyle\int_{T^1}X_\MaTime\sin\Phase\,\d\Phase%
\end{pmatrix}^\trans
\end{align*}
and with respect to the variables $B_1$, $B_2$, ${B_1}_\MaTime$,
${B_2}_\MaTime$ we find
\begin{align*}
{\bsSi}^\red={\bsT}_0^\prime
\begin{pmatrix}2\om\PhaseDer{}&0\\0&0
\end{pmatrix}{\bsT}=
2\,\om
\begin{pmatrix}0&-1&0&0\\1&0&0&0\\0&0&0&0\\0&0&0&0
\end{pmatrix},
\end{align*}
which implies \eqref{nls:Red.Theorem2.Eqn4}.
From this and \eqref{nls:Red.Theorem2.Eqn5} we conclude that both
the reduced Lagrangian and Hamiltonian equations on $P_0$ read
\begin{align*}
 -2\,\om{B_2}_\MaTime&=- \rho_1{B_1}_{\MaLagr\MaLagr}+
\tfrac{1}{2\pi}\rho_2\at{B_1^2+B_2^2}B_{1},\\
+2\,\om{B_1}_\MaTime&=- \rho_1{B_2}_{\MaLagr\MaLagr}+
\tfrac{1}{2\pi}\rho_2\at{B_1^2+B_2^2}B_{2}
\end{align*}
and rewriting this in terms of $A$ we find
\eqref{nls:Red.Theorem2.Eqn3}.
\end{proof}
As before, Theorem \ref{nls:Red.Theorem2} concerns a reduced
macroscopic model on $P_0$ that is obtained by  means of
\emph{formal expansions}. In particular, it is not obvious that the
nlS equation \eqref{nls:Red.Theorem2.Eqn3} combined with the
modulation ansatz \eqref{nls:ModAnsatz} yields approximate solutions
for the KG chain. However, the careful residual analysis from
\cite{GM04,GM06} provides rigorous justification results, and thus
we can regard $P_0$ as an approximate invariant manifold.
\begin{remark}
Like for the KdV example, the terms $\bbI_0$ and $\bbI_1$, which
cause the cancelations in $\bbL_0$ and $\bbL_1$, provide macroscopic
conservation laws. In fact, with some calculations we find
\begin{align*}
\bbI_0\sim\int\mhintlimits_{\Rset}  |A|^2
\,\d\MaLagr\quad\text{and}\quad \bbI_1\sim\int\mhintlimits_{\Rset}
 \mathrm{Im} ( A_{\MaLagr}\overline{A}) \,\d\MaLagr,
\end{align*}
which equal the macroscopic integrals of motion associated with the
symmetries under phase shifts $A\mapsto\mhexp{\iu{}s}A$, and shifts
in the $\MaLagr$-direction, respectively.
\end{remark}
%
%
\subsection{Three-wave-interaction for the KG chain}
\label{sec:ChainNew.threeWaves}
%
%
%
Here we discuss the interaction of three pulses in the KG chain, see
\eqref{Intro:TW.MSAnsatz}. More precisely, we consider three pulses
$p_1$, $p_2$, and $p_3$,  and aim to understand how the resulting
microstructure is modulated on the hyperbolic scale for space and
time.
%
%
\paragraph{Pulses in the KG chain}
%
%
We briefly summarize some aspects of  \emph{pulses}, and refer to
\cite{GM04,Gia08,GMS07} for more details. A plane wave is a solution
to the linearized chain
\begin{align*}
x_j(t)=A\mhexp{\iu\nnat{\om\MiTime+\theta j}}+
\cc=A\mhexp{+\iu\nnat{\om\MiTime+\theta j}}+
\mhol{A}\mhexp{-\iu\nnat{\om\MiTime+\theta j}}, \quad j\in \Zset,
\end{align*}
with complex amplitude $A$, frequency $\om$, and wave number
$\theta$. Notice that $\npair{-\theta}{-\om}$ gives the same pulse
as $\npair{\theta}{\om}$, whereas $\npair{\theta}{-\om}$ is the
pulse that travels in opposite direction. Obviously, each plane wave
must satisfy the dispersion relation
\begin{align}
\notag
\om^2= \Om^2\at{\theta}=v_2+2\alpha\at{1-\cos\theta} \quad
\text{with } v_2=\Phi''_0(0).
\end{align}
Since the amplitude $A$ can always be chosen arbitrarily, we can
identify each plane wave with a point in
\begin{align*}
\calP=\{\npair{\theta}{\om{}}\,:\;\om^2=\Om^2\nnat{\theta}\} \subset
T^1\times \Rset.
\end{align*}
In what follows we assume the stability condition
\begin{equation}\label{3wi:StabCond}
\min\{4\alpha+v_2,\,v_2 \}=
\min\limits_{\theta\in[0,2\pi)}\Om^2(\theta) >0,
\end{equation}
so that each single plane wave is a stable solution to the
linearized chain.
\bigpar
A simple \emph{pulse} is a modulation of a plane wave by a slowly
varying amplitude
\begin{align*}
x^{(k)}_j(t)= \eps \,A_k(\eps t,\eps j)\,
\mhexp{\iu\nnat{\om_k\MiTime+\theta_k j}}+ \cc.
\end{align*} %
On the hyperbolic scale $\tau=\eps t$ and $y=\eps j$ a pulse will
simply travel with group velocity $c_k=\Omega'(\theta_k)$. However,
if different pulses associated with $p_k\in \calP$ meet each other
they interact in case their frequencies and wave vectors are in
resonance. Three plane waves $p_1,p_2,p_3\in\calP$ are called in
\emph{three-wave resonance} if there exists a choice of three signs
$m_i\in\{-1,\,+1\}$ such that ${m_1}p_1+m_2p_2+m_3p_3=\pair{0}{0}
\in T^1\times \Rset$. By using complex conjugates and replacing
$p_k$ by $-p_k$ if necessary, we can always assume that
\begin{equation}\label{eq:TWRes}
p_1+p_2+p_3=0 ,\quad\text{i.e.,}\quad \left\{
\begin{array}{cl} \theta_1+\theta_2 +\theta_3=0 &\in T^1,\\
\omega_1+\omega_2+\omega_3 =0 &\in \Rset. \end{array}\right.
\end{equation}
This {\it resonance condition} arises naturally as it is equivalent
to the \emph{cancelation of oscillations} via
\begin{align*}
\mhexp{\iu\,\nnat{{\om_1}\MiTime+\theta_1 j}}
\mhexp{\iu\,\nnat{{\om_2}\MiTime+\theta_2 j}}
\mhexp{\iu\,\nnat{{\om_3}\MiTime+\theta_3 j}}=1 \quad \text{for all
}t\in \Rset \text{ and } j\in \Zset,
\end{align*}
and guarantees that the product of two pulses contains oscillatory
terms that appear also in the third pulse.
\par%
Of course, the KG chain allows for resonances between more than
three pulses, but in our context these can be ignored for the
following reason. According to \eqref{Intro:TW.MSAnsatz}, the pulse
amplitudes scale with $\eps$, so that three-pulse resonances, which
are related to quadratic products such as $x_1x_2$, correspond to
the power $\eps^2$. Interactions of more than three pulses, however,
contribute to order $\DO{\eps^3}$, and are thus not relevant on the
hyperbolic scale.
\par

However, to make the presentation as simple as possible we now
assume that the three plane waves $p_1,p_2$, and $p_3$ are chosen
such that except for \eqref{eq:TWRes} there are no further
resonances. More precisely, we define
\begin{align*}
\calZ=\{\, (k_1,k_2)\in \Zset^2\: : \: (\omega_1k_1{+}\omega_2
k_2)^2 = \Omega^2(\theta_1k_1{+}\theta_2k_2)\,\},
\end{align*}
and make the following assumption.
\begin{assumption}
\label{3wi:AssPrms}
The vectors $p_1, p_2\in \calP$ are chosen such that
\begin{align}
\notag
 \calZ=\{\: (1,0),\ (-1,0),\ (0,1),\
(0,-1),\ (1,1),\ (-1,-1)\:\}.
\end{align}
\end{assumption}
Obviously, we always have $(\pm1,0)\in \calZ$ by $p_1\in \calP$ and
similarly $(0,\pm1)\in \calZ$ by $p_2\in \calP$. If additionally
$p_3\in \calP$ satisfies the three-wave resonance condition
\eqref{eq:TWRes}, then $\pm(1,1)$ also lies in $\calZ$. Thus,
Assumption \ref{3wi:AssPrms} already implies \eqref{eq:TWRes} and
additionally excludes any further resonances involving these three
plane waves.

\begin{remark}
$\at{i}$ According to \cite{Gia08}, the resonance condition
\eqref{eq:TWRes} is equivalent to
\begin{equation*}
 \alpha^2\mu_1\mu_2(\mu_1{+}\mu_2)
 +(2\alpha\mu_1\mu_2{+}\delta)\sqrt{(\alpha\mu_1{+}\delta)(\alpha\mu_2{+}\delta)}
 +\delta\alpha(\mu_1\mu_2{+}\mu_1{+}\mu_2)+\tfrac{5}{4}\delta^2
=0
\end{equation*}
with $\mu_i=(1-\cos\theta_i)/2$ and $\delta=v_2/4$. Hence, for
$\alpha>0$ (attractive nearest-neighbour interactions) the resonance
condition cannot be satisfied as the stability condition
\eqref{3wi:StabCond} implies $\delta>0$. However, for $\alpha\in
(-\tfrac{1}{4}v_2,-\tfrac{3}{16}v_2)$ (repulsive case) the stability
condition is still satisfied, but now there exists a one-parameter
family of solutions $\pair{\mu_1}{\mu_2}$.
\\$\at{ii}$ In general, it is not easy to check the
\emph{non-resonance conditions} implied by Assumption
\ref{3wi:AssPrms}, i.e., to prove that no further plane waves are
contained in $\calZ$. The mapping $\Zset^2\ni k \mapsto
(k{\cdot}(\theta_1,\theta_2), k{\cdot}(\omega_1,\omega_2)) \in T^1
\times \Rset $ may have a dense image and hence comes close to the
set $\calP$ very often, giving rise to a small divisor problem.
However, by varying also $\alpha$ and $v_2$, it is possible to
choose $\theta_1,\theta_2$ as rational multiples of $\pi$ and to
make $\omega_1/\omega_2$ rational as well. Then, the image of the
above mapping hits every bounded set in finitely many points. Then,
Assumption \ref{3wi:AssPrms} appears very reasonable.
\\$\at{iii}$ In Remark \ref{3wi:AssPrms.Weak} below we provide a weaker
variant of Assumption \ref{3wi:AssPrms}.
\end{remark}
%
%

%
%
%
\paragraph{Invertible two-scale ansatz}
%
The resonance and non-resonance conditions imposed by Assumption
\ref{3wi:AssPrms} imply that there exist exactly two independent
phase variables. Therefore, concerning the embedding of the
microscopic system, it is necessary and sufficient to introduce a
two-dimensional phase variable
$\Phase=\pair{\Phase_1}{\Phase_2}\in{T^2}$, i.e.,
\begin{align*}
Q=L^2\at{\Rset{\times}T^2;\,\d\MiLagrC\d\Phase},\quad
P=L^2\at{\Rset{\times}T^2;\,\d\MaLagr\d\Phase}.
\end{align*}
Similarly to the nlS example we start with the invertible two-scale
ansatz
\begin{equation}\label{3wi:MSAnsatz}
x\triple{\MiTime}{\MiLagrC}{\Phase} =\eps{X}\triple{\eps\MiTime}
{\eps\MiLagrC} {\Phase+\om\MiTime+\theta\MiLagrC},
\end{equation}
so that the corresponding inverse two-scale transformation
$\bbT_\vel\pair{\eps}{\MiTime}:TP\to TQ$  reads
\begin{equation}\label{3wi:MSTrafo}
\Bat{\bbT_\vel\pair{\eps}{\MiTime}\pair{X}{X_\MaTime}}
\pair{\MiLagrC}{\Phase}=\pair{\eps{X}}
{\eps^2{X_\MaTime}+\eps\om\cdot X_\Phase} \pair{\eps\MiLagrC}
{\Phase+\om\MiTime+\theta\MiLagrC}
\end{equation}
with $\theta=\pair{\theta_1}{\theta_2}$, $\om=\pair{\om_1}{\om_2}$,
and
$\partial_\Phase=\pair{\partial_{\Phase_1}}{\partial_{\Phase_2}}$.
\begin{remark}
The ansatz \eqref{3wi:MSAnsatz} provides
\begin{math}
\bbT_\con\pair{\eps}{\MiTime}= \calT_\con^{-1}
\circ\calM_\con\at{-\MiTime} \circ{\bbS_\con}\at\eps :P\to Q
\end{math} %
with inverse scaling transformation
\begin{math}
\at{\bbS_\con\at\eps{X}}\pair{\MiLagrC}{\Phase}=
\eps{}X\pair{\eps\MiLagrC}{\Phase}
\end{math},
weak symmetry transformation
\begin{math}
\nnat{\calT_\con\,x}
\pair{\MiLagrC}{\Phase}=x\pair{\MiLagrC}{\Phase-\theta\MiLagrC},
\end{math}
and moving frame transformation
$\at{\calM_\con\at\MiTime\,x}\pair{\MiLagrC}{\Phase}$
$=x\pair{\MiLagrC}{\Phase-\om\MiTime}$ associated to the integral of
motion \begin{math} \mathcal{I}\pair{x}{x_\MiTime}=
-\om\cdot\calI_{\mathrm{phase}}\pair{x}{x_\MiTime}=
-\int_{\Rset{\times}T^2}x_\MiTime\at{\om\cdot x_\Phase}
\d\MiLagrC\d\Phase.
\end{math}
\end{remark}
%
%
%

\paragraph{Leading order reduction}
%
We start with the computation of the transformed structures.
\begin{lemma}
\label{3wi:Lemma1}
The transformation \eqref{3wi:MSTrafo} provides $\bbL=\bbK-\bbV$,
$\bbE=\bbK+\bbV$, and $\bbH=\bbE+\bbI$, as well as the following
expansions:
\begin{enumerate}
\item
\begin{math}
\bbI= \eps\bbI_0+\eps^2\bbI_1+ \DO{\eps^3}
\end{math}
with
\begin{align*}
\bbI_0\at{X}= -\!\!\!\int\mhintlimits_{\Rset{\times}T^2}\!\!\!
(\om\cdot X_\Phase)^2\, \d\MaLagr\d\Phase, \quad
\bbI_1\pair{X}{X_\MaTime}=
-\!\!\!\int\mhintlimits_{\Rset{\times}T^2}\!\!\!
X_\MaTime{}(\om\cdot X_\Phase)\,\d\MaLagr\d\Phase,
\end{align*}
\item
\begin{math}
\bbK= \eps\bbK_0+\eps^2\bbK_1+ \DO{\eps^3}
\end{math}
with
$\bbK_0=-\tfrac{1}{2}\,\bbI_0$ and $\bbK_1=-\bbI_1$,
\item
\begin{math}
\bbV= \eps\bbV_0+\eps^2\bbV_1+ \DO{\eps^3}
\end{math}
with
\begin{align*}
\begin{array}{lclcl}
\bbV_0\at{X}&=&\displaystyle -\tfrac{\alpha}{2}\!\!\!
\int\mhintlimits_{\Rset{\times}T^2}\!\!\! X\laplace_{0,\,\theta}X\,
\d\MaLagr\d\Phase &+&\displaystyle \tfrac{v_2}{2}\!\!\!
\int\mhintlimits_{\Rset{\times}T^2}\!\!\! X^2\, \d\MaLagr\d\Phase,
\\%
\bbV_1\at{X}&=&\displaystyle %
-\tfrac{\alpha}{2}\!\!\!
\int\mhintlimits_{\Rset{\times}T^2}\!\!\!%
X\at{\nabla^{+}_{0,\theta}+\nabla^-_{0,\theta}}X_\MaLagr
\,\d\MaLagr\d\Phase%
&+&\displaystyle%
\tfrac{v_3}{6}\!\!\! \int\mhintlimits_{\Rset{\times}T^2}\!\!\!
X^3\,\d\MaLagr\d\Phase,
\end{array}
\end{align*}
where $\alpha=\NNPot^{\prime\prime}(0)$,
$v_2=\OSPot^{\prime\prime}(0)$, and
$v_3=\OSPot^{\prime\prime\prime}(0)$.
\end{enumerate}
Moreover, the  matrix $\bsSi$ corresponding to $\bssi$ obeys the
exact expansion
\begin{align*}
\bsSi=\eps^2\bsSi_1+\eps^3\bsSi_2,\quad
 \bsSi_1=
 \begin{pmatrix}
 -2\,\om\cdot\partial_\Phase &0\\0&0
 \end{pmatrix}
 ,\quad \bsSi_2=
 \begin{pmatrix}
 0&-1\\1&0
 \end{pmatrix},
\end{align*}
so that $\bsSi$ is non-degenerate due to $\bsSi_2$.
\end{lemma}
\begin{proof}
Analogously to the proof of Lemma \ref{nls:Lemma0} we find the
equations for $\bbL$ and $\bbH$ along with
\begin{align*}%
 \bbI\triple{\eps}{X}{X_\MaTime}
=\calI\circ\bbT_\vel\pair{\eps}{\MiTime}
 &=\eps^{-1}\int\mhintlimits_{\Rset{\times}T^2} \at{\eps^2{X}_\MaTime
 +\eps\om\cdot X_\Phase}
 \at{-\eps\om\cdot X_\Phase}\,\d\MaLagr\d\Phase,
 \\
 \bbK\triple{\eps}{X}{X_\MaTime}
=\calK\circ\bbT_\vel\pair{\eps}{\MiTime}
 &=\eps^{-1}\int\mhintlimits_{\Rset{\times}T^2}
 \tfrac{1}{2}\at{\eps^2{X}_\MaTime+\eps\om\cdot\partial_\Phase X}^2
 \,\d\MaLagr\d\Phase,
 \\
 \bbV\pair{\eps}{X}
=\calV\circ\bbT_\vel\pair{\eps}{\MiTime}
 &= \eps^{-1}\int\mhintlimits_{\Rset{\times}T^2}
 \at{-\eps^2\tfrac{\alpha}{2}X\laplace_{\eps,\,\theta}X+
 \OSPot\at{\eps{X}}}\,\d\MaLagr\d\Phase.
\end{align*}
Moreover, the expansions with respect to $\eps$ follow from direct
calculations, and using
\begin{align*}
\bsT=
\begin{pmatrix}
\eps&0\\\eps\,\om\cdot\PhaseDer{}&\eps^2
\end{pmatrix}
\end{align*}
the matrix $\bsSi$ can be calculated by means of Remark
\ref{TST:Remark.SymTrafo}.
\end{proof}
As a consequence of Lemma \ref{3wi:Lemma1} we obtain $\bsSi_0=0$ and
$\bbL_0\at{X}=-\bbH_0\at{X}$, and the leading order equation
\begin{equation}
\notag
(\om\cdot\partial_\Phase)^2 X -\alpha \laplace_{0,\,\theta}X+v_2
X=0.
\end{equation}
is again quasi-stationary. Applying Fourier transformation with
respect to $\Phase$,
 a general function $X$ has the form
$X(y,\Phase)= \sum_{k\in \Zset^2} F_k(y) \mhexp{\iu k\cdot \Phase} $
and solves the above equation if and only if $F_k =0$ for all
$k\not\in \calZ$.
\begin{lemma}\label{l:LeadOrdRed}
Under Assumption \ref{3wi:AssPrms} the leading order Lagrangian and
Hamiltonian equations are quasi-stationary, and all solutions are
given by
\begin{equation}\label{MicrSol}
X_0\pair{\MaLagr}{\Phase}=\sum_{n=1}^3
A_n\at{\MaLagr}\mhexp{\iu\Phase_n}+\cc,
\end{equation}
with $\Phase_3=-\Phase_1-\Phase_2$ and arbitrary
$A_n\in{L}^2\at{\Rset;\Cset}$, $n=1,2,3$. Moreover, we have
$\bbL_0\at{X_0}=\bbH_0\at{X_0}=0$ for all $X_0$ with
\eqref{MicrSol}.
\end{lemma}
%
%

\paragraph{Elimination of the microstructure}
As outlined in \S\ref{sec:Reductions}, we derive the reduced
macroscopic model by restricting the next-leading order terms to the
space
\begin{equation*}
P_0 =\{X_0\in P\ :\ \partial_X\bbL_0(X_0)=\partial_X\bbH_0(X_0)=0\}
\cong \{A=(A_1,A_2,A_3)\in ({L^2}\at{\Rset;\,\Cset})^3\}.
\end{equation*}
Notice that, in contrast to the nlS example from
\S\ref{sec:ChainNew.nls}, here $\bbL_1|_{TP_0}$ and $\bbH_1|_{TP_0}$
do not vanish, and provide the reduced Lagrangian and Hamiltonian.
In particular, we need not care for the correction terms coming from
$X=X_0+\eps{X_1}$.
\begin{theorem}
Under Assumption \ref{3wi:AssPrms} the reduced Lagrangian and
Hamiltonian are given by
\begin{align*}
{\bbL}^\red\pair{A}{A_\MaTime}
={\bbK}^\red\pair{A}{A_\MaTime}-{\bbV}^\red\at{A},\quad
{\bbH}^\red\pair{A}{A_\MaTime}={\bbV}^\red\at{A},
\end{align*}
with
\begin{align*}
{\bbK}^\red\pair{A}{A_\MaTime}
&=\iu\sum_{n=1}^3\om_n\int\mhintlimits_{\Rset}
A_n{\mhol{A}_n}_\tau\,\d\MaLagr+\cc
\\
{\bbV}^\red\at{A} &=v_3\int\mhintlimits_{\Rset}A_1A_2A_3
\,\d\MaLagr+\iu\sum_{n=1}^3\om_n\om_n^\prime\int\mhintlimits_{\Rset}
 A_n\overline{A_n}_y\,\d\MaLagr
+\cc,
\end{align*}
where $\om_n\om'_n=\Omega(\theta_n)\Omega'(\theta_n)$. Moreover,
\begin{align*}
{\bsSi}^\red =-2\iu
\begin{pmatrix}\mathbf{\Omega}&0
\\0&0\end{pmatrix},
\qquad
\mathbf{\Omega}=\begin{pmatrix}\om_1&0&0\\0&\om_2&0\\0&0&\om_3\end{pmatrix},
\end{align*}
is the reduced symplectic matrix, and the reduced Lagrangian and
Hamiltonian equations are equivalent to the three-wave-interaction
equations \eqref{Intro:TW.MSequations}.
\end{theorem}
\begin{proof}
According to \S\ref{sec:Reductions} we have
\begin{align*}
\bbL^\red=c\bbK_1|_{TP_0}-c\bbV_1|_{TP_0},\quad
\bbH^\red=c\bbK_1|_{TP_0}+c\bbV_1|_{TP_0}+c\bbI_1|_{TP_0},\quad
\bssi^\red=c{\bssi_1}|_{TP_0},
\end{align*}
where for convenience we introduced a trivial scaling by
$c=1/\at{4\pi^2}$. Inserting \eqref{MicrSol} into the formulas from
Lemma \ref{3wi:Lemma1}, and exploiting Assumption \ref{3wi:AssPrms}
we obtain
\begin{align*}%
c\,\bbK_1\pair{A}{A_\MaTime}&=-c\bbI_1\pair{A}{A_\MaTime}
=\sum_{n=1}^3 \int\mhintlimits_{\Rset}\iu\om_n
A_n\overline{A_n}_\tau \,\d\MaLagr+\cc,
\\
c\bbV_1\at{A} &= \Big(\sum_{n=1}^3\int\mhintlimits_{\Rset}
\iu\om_n\om_n^\prime A_n\overline{A_n}_y\,\d\MaLagr
+v_3\int\mhintlimits_{\Rset}A_1A_2A_3 \,\d\MaLagr\Big)+\cc,
\end{align*}
where we used $\alpha\sin\theta_n=\om_n\om_n^\prime$ and the
properties of $\nabla$ and $\laplace$, see Remark
\ref{GS:DiscreteNablas}. Concerning $\bsSi^\red$ we observe that the
ansatz \eqref{MicrSol} can be written as
\begin{align*}
\begin{pmatrix}
X\\X_\MaTime
\end{pmatrix}
={\bsT}_0
\begin{pmatrix}
A_1\\A_2\\A_3\\{A_1}_\MaTime\\{A_2}_\MaTime\\{A_3}_\MaTime
\end{pmatrix}+\mathrm{c.c.}
,\quad {\bsT}_0=
\begin{pmatrix}
\mhexp{\iu\phi_1}&\mhexp{\iu\phi_2}&\mhexp{-\iu(\phi_1+\phi_2)}&0&0&0\\
0&0&0&\mhexp{\iu\phi_1}&\mhexp{\iu\phi_2}&\mhexp{-\iu(\phi_1+\phi_2)}
\end{pmatrix}
\end{align*}
with ${\bsT}_0:TP_0\to TP$. The adjoint operator
${\bsT}_0^\prime:TP\to TP_0$ reads
\begin{align*}
{\bsT}_0^\prime\begin{pmatrix} X\\X_\MaTime\end{pmatrix} =
\int\mhintlimits_{T^2}
\begin{pmatrix}\displaystyle\mhexp{-\iu\phi_1}&0
\\\displaystyle\mhexp{-\iu\phi_2}&0
\\\displaystyle\mhexp{\iu(\Phase_1+\Phase_2)}&0
\\0&\displaystyle\mhexp{-\iu\phi_1}
\\0&\displaystyle\mhexp{-\iu\phi_2}
\\0&\displaystyle\mhexp{\iu(\Phase_1+\Phase_2)}
\end{pmatrix}
\begin{pmatrix} X\\X_\MaTime\end{pmatrix}
\,\d\Phase,
\end{align*}
and we find
\begin{align*}
\bsSi^\red=c\,{\bsT}_0^\prime{\bsSi}_1{\bsT}_0
&=c\,{\bsT}_0^\prime\begin{pmatrix}-2\om\cdot\PhaseDer{}&0\\0&0\end{pmatrix}
{\bsT}_0 =-2\iu
\begin{pmatrix}\mathbf{\Omega}&\mathbf{0}
\\\mathbf{0}&\mathbf{0}\end{pmatrix}.
\end{align*}
%
Finally, the Lagrangian equations to ${\bbL}^\red$ are given by
\begin{align*}
\partial_\tau
\left(\partial_{\overline{A}_\tau}{\bbL}^\red\pair{A}{A_\MaTime}\right)
-\partial_{\overline{A}}{\bbL}^\red\pair{A}{A_\MaTime}=0
\end{align*}
and equal %
\begin{align*}\displaystyle -2\,\iu\mathbf{\Omega}
A_\tau =\partial_{\mhol{A}}{\bbH}^\red(A),
\end{align*}
which is the Hamiltonian equation to $({\bbH}^\red,{\bssi}^\red)$.
Finally, both equations coincide with \eqref{Intro:TW.MSequations}.
\end{proof}
As mentioned in the introduction, one can obtain the macroscopic
equations \eqref{Intro:TW.MSequations} also by inserting the
two-scale ansatz \eqref{Intro:TW.MSAnsatz} into the Klein--Gordon
chain \eqref{Intro:KGChain} and requiring the coefficients of the
terms $\eps^2\mhexp{\iu \at{\om_n\MiTime+\theta_n\MiLagr}}$ to
vanish. Based on this formal expansion one can then justify the
validity of \eqref{Intro:TW.MSequations}, see \cite{Gia06,Gia08} and
\S7.2 in \cite{AMSMSP:GHM}.

\begin{remark}
Inserting \eqref{MicrSol} into the formulas from Lemma
\ref{3wi:Lemma1} and exploiting the resonance condition, we obtain
\begin{align*}
c\bbK_0|_{P_0}=-\tfrac{1}2c\bbI_0|_{P_0}=c\bbV_0|_{P_0}
=\tfrac{1}2c\bbE_0|_{P_0}
=\sum_{n=1}^3\om_n^2\int\mhintlimits_{\Rset}|A_n|^2\,\d\MaLagr.
\end{align*}
These equalities reflect
the cancelation in $\bbL_0$ and manifest the equipartition of energy
for plane-wave solutions. Moreover, the total energy $\bbE_0$ is the
first integral associated to the invariance under phase shifts.
\end{remark}
\begin{remark} \label{3wi:AssPrms.Weak}
Assumption \ref{3wi:AssPrms}, which excludes all other possible
resonances except for $p_1+p_2+p_3=0$, can be weakened as follows.
As already mentioned, on the hyperbolic scale we can ignore
resonances of more than three pulses. We shall, however, exclude the
possibility that further pulses are created via three-pulse
resonance, because otherwise we expect the three-pulse solution that
involves $p_1$, $p_2$, and $p_3$ to be unstable on the hyperbolic
scale. This gives rise to the non-resonance conditions
\begin{align*}
2p_1,\,2p_2,\,2p_3,\,p_1-{p_2},\,p_1-{p_3},\,p_2-{p_3}\,\notin\calP\setminus
\{\pm p_1,\,\pm p_2,\,\pm p_3\}.
\end{align*}
Assuming this, it can happen that there exist further degenerate
three-pulse resonances between $p_1$, $p_2$, and $p_3$, as for
instance
$2p_1+p_3=0$ or $2p_1-p_2=0$. In this case we
still obtain a
stable three-pulse solution, but the reduction procedure
provides a different modulation equation. In fact, such degenerate
resonances give rise to further cubic coupling terms in the formula
for $\bbV^\red$, as for instance
$\frac{v_3}{2}\int_\Rset (A_1^2+ A_2^2)A_3\d\MaLagr+\cc$
or
$\frac{v_3}{2}\int_\Rset A_1^2\mhol{A_2}\d\MaLagr+\cc$,
respectively.
Altogether, in order to guarantee
that \eqref{Intro:TW.MSequations} is a reasonable macroscopic model
it is sufficient to assume the resonance condition
\begin{align*}
\pair{1}{0},\,\pair{0}{1},\,\pair{1}{1}\in\calZ
\end{align*}
and the non-resonance conditions
\begin{align*}
\pair{0}{2},\,\pair{2}{0},\,\pair{1}{-1},\,\pair{2}{1},\,\pair{1}{2},\,
\pair{2}{2}\notin\calZ.
\end{align*}

\end{remark}

\subsection{Outlook to further examples}\label{sec:ChainNew.FurtherEx}
%
Finally, we give a brief overview on two other classes of
micro-macro transitions that can also be studied with respect to
Hamiltonian and Lagrangian reductions. However, since these examples
lead to additional problems, their investigation is left for a
forthcoming study.

\paragraph{Coupled systems}
%
describe the interactions between modulated pulses and waves with
long wave-length. The interesting feature here is that the
corresponding two-scale ansatz
\begin{align}
\label{Intro:CS.MS.Ansatz} %
x\triple{\MiTime}{\MiLagrC}{\Phase}
=\eps^\alpha{X}\pair{\eps\MiTime}{\eps\MiLagr}+
\eps^\beta{A}\pair{\eps\MiTime}
{\eps\MiLagr}\mhexp{+\iu\at{\om\MiTime+\theta\MiLagr}}+ \cc
\end{align}
combines contributions with different orders of magnitude. For
instance, if we derive the effective macroscopic model for
$\alpha=0$ and $\beta=1$ by inserting \eqref{Intro:CS.MS.Ansatz}
into the microscopic equation of motion, we find
\begin{align}%
\label{Intro:CoupledSystem1} \MaTimeDerS{}{X}&=c_{\mathrm
m}^2\,\MaLagrDerS{X},\quad \iu\, \MaTimeDer{A}=\iu\,
c_\mathrm{gr}\,\MaLagrDer{A}-\rho_0\,\MaLagrDer{X}\,A.
\end{align}
However, the asymmetric coupling between both equations prevents
\eqref{Intro:CoupledSystem1} from being the Euler-Lagrange equation
of a suitable chosen macroscopic Lagrangian with variables $X$ and
$A$, and we conclude that the reduction of Lagrangian and
Hamiltonian structures yields a different reduced model.

%
%
%
%
\paragraph{Whitham's modulation theory}
%
%
is another example postponed to our forthcoming paper. This theory
was originally developed in the context of PDEs, see
\cite{Whi74,Kam06}, but can also be applied to discrete systems, see
for instance \cite{HLM94,FV99}. The main building block for
Whitham's modulation theory are \emph{periodic travelling waves}.
These are exact solutions to \eqref{Intro:AtomicChain} satisfying
$x_\MiLagr\at\MiTime=\mathfrak{X}\at{\om\MiTime+\theta\MiLagr}$ with
$\Phase=\om\MiTime+\theta\MiLagr$. For the atomic chain the profile
$\mathfrak{X}$ must fulfil the following advance-delay differential
equation
\begin{align}
\notag
\omega^2\, \mathfrak{X}_{\Phase\Phase}\at{\Phase} =
\NNPot^\prime\Bat{\mathfrak{X}\at{\Phase+\theta}-\mathfrak{X}\at{\Phase}}
-\NNPot^\prime\Bat{\mathfrak{X}\at{\Phase}-\mathfrak{X}\at{\Phase-\theta}}
-\OSPot^{\prime}\at{\mathfrak{X}\at\Phase}.
\end{align}
In case that both $\OSPot^\prime$ and $\NNPot^\prime$ are linear, we
can solve this equation by means of Fourier transformation, and will
recover plane waves with \eqref{Intro:AtomicChain.DispRel}, but for
nonlinear potentials more sophisticated methods are necessary,
compare for instance \cite{DHM06} and references therein. The basic
ideas behind Whitham's modulation theory can be summarized as
follows: We consider the KG chain and start with the following
two-scale ansatz
\begin{align}
\notag
x_\MiLagr\at\MiTime=
\mathfrak{X}\triple{\eps\MiLagr}{\eps\MiTime}{\eps^{-1}
\ModPhase\pair{\eps\MiLagr}{\eps\MiTime}}.
\end{align}
Here, $\ModPhase$ is the \emph{modulated phase} and provides the
fields of wave number and frequency via
\begin{math}
\om\pair{\MaTime}{\MaLagr}=
\VarDer{\ModPhase}{\MaTime}\pair{\MaTime}{\MaLagr}
\end{math} %
and
\begin{math}
\theta\pair{\MaTime}{\MaLagr}=
\VarDer{\ModPhase}{\MaLagr}\pair{\MaTime}{\MaLagr},
\end{math} %
and for each $\pair{\MaTime}{\MaLagr}$ the function
$\Phase\mapsto\mathfrak{X}\triple{\MaTime}{\MaLagr}{\Phase}$ is
assumed to be a periodic travelling wave. Whitham's approach to the
Lagrangian reduction allows to derive easily the corresponding
macroscopic model. For the KG chain we find two nonlinear
conservation laws
\begin{align}
\label{Intro:Whitham.ModSys}%
\MaTimeDer{\theta}\pair{\MaTime}{\MaLagr}-\MaLagrDer{\om}%
\pair{\MaTime}{\MaLagr}=0,\quad
\quad\MaTimeDer{S}\pair{\MaTime}{\MaLagr}+
\MaLagrDer{g}\pair{\MaTime}{\MaLagr}=0,
\end{align}
which are closed by the Gibbs equation $\d{L}=S\,\d\om+g\,\d\theta$
and the equation of state $L=L\pair{\theta}{\om}$, which provides
the action of a travelling wave as a function of $\om$ and $\theta$.
Moreover, it can be shown that \eqref{Intro:Whitham.ModSys} is a
system of Hamiltonian PDEs.
\par
The new feature appearing in this example is that the corresponding
two-scale transformation depends on the modulated phase $\ModPhase$,
which in turn depends on the solution to the macroscopic equation.
In other words, within Whitham's modulation theory we do not know
the two-scale transformations a priori and this complicates the
reduction of Lagrangian and Hamiltonian structures. Finally, the
modulation theory for FPU chains leads to further complications,
since the Galilean invariance of \eqref{Intro:FPUChain} causes a
coupling between macroscopic waves and modulated oscillations, see
\cite{FV99,Her05,AMSMSP:GHM,AMSMSP:DHR,DH07}.

%

%
\providecommand{\bysame}{\leavevmode\hbox to3em{\hrulefill}\thinspace}
\providecommand{\MR}{\relax\ifhmode\unskip\space\fi MR }
\providecommand{\MRhref}[2]{%
  \href{http://www.ams.org/mathscinet-getitem?mr=#1}{#2}
}
\providecommand{\href}[2]{#2}

\end{document}